\documentclass[journal,12pt,onecolumn,draftcls]{IEEEtran}
\usepackage{graphics}
\usepackage{multicol}
\usepackage{lipsum}
\usepackage{color}
\usepackage[cmex10]{amsmath}
\usepackage{amsthm}
\usepackage{amssymb}
\usepackage{mathrsfs}
\usepackage{mathtools}
\usepackage{amsbsy}
\usepackage[colorlinks=true,bookmarks=false,citecolor=blue,urlcolor=blue]{hyperref}
\usepackage{xcolor}
\usepackage{mathrsfs}
\usepackage{setspace}
\usepackage{float}
\usepackage{graphicx}
\usepackage{pstool}
\usepackage{lettrine}
\usepackage{newclude}
\usepackage[normalem]{ulem}
\usepackage{latexsym}
\usepackage{algpseudocode}
\usepackage{algorithm,algpseudocode}
\usepackage{algorithmicx}
\usepackage{multirow}
\usepackage{amsmath}
\usepackage{wasysym}
\usepackage{mathrsfs}
\usepackage{amsmath,cite,amsfonts,amssymb,color}

\def\E{\mathbb{E}}
\ifCLASSOPTIONcompsoc
\usepackage[caption=false,font=normalsize,labelfon
t=sf,textfont=sf]{subfig}
\else
\usepackage[caption=false,font=footnotesize]{subfig}
\fi
\allowdisplaybreaks

\newtheorem{theorem}{Theorem}

\newtheorem{Remark}{Remark}
\graphicspath{{figures/}}
\allowdisplaybreaks

\newcommand{\bbeta}{\boldsymbol{\eta}}

\graphicspath{{./Figures/1_CDF_Diff_DUEs/}{./Figures/2_CDF_Diff_CUEs/}{./Figures/3_D2Ddistance/}{./Figures/./4_ADCresolution/}{./Figures/SystemModel/}{./Figures/Bios/}{./Figures/FinalResults/}{./Figures/MajorRevise/}}
\definecolor{azure}{rgb}{0.0, 0.5, 1.0}
\definecolor{violet}{rgb}{0.58, 0.0, 0.83}

\begin{document}

%\title{D2D Underlaid Cell-Free Massive MIMO with Low Resolution ADC}
%\title{Cell-Free Massive MIMO with Low Resolution ADC and Underlaid D2D Communications}
\title{Coexistence of D2D Communications and Cell-Free Massive MIMO Systems with Low Resolution ADC for Improved Throughput in Beyond-5G Networks}
%\title{Cell-Free Massive MIMO Underlaid D2D Communications with Low Resolution ADCs}

\author{Hamed~Masoumi,
        Mohammad~Javad~Emadi, and 
        Stefano~Buzzi, {\emph{Senior Member, IEEE}}
\thanks{% Manuscript received December 9, 2018; revised April 17, 2019,
% August 1, 2019, and October 8, 2019; accepted October 13, 2019. This
% work was supported in part by the Iran National Science Foundation
% (INSF). The associate editor coordinating the review of this article and
% approving it for publication was C.-K. Wen. \textit{(Corresponding author:
% Mohammad Javad Emadi.)}.

H. Masoumi and M. J. Emadi are with the Department of Electrical Engineering, Amirkabir University of Technology (Tehran Polytechnic), Tehran, Iran (E-mails:\{hamed\_masoomy, mj.emadi\}@aut.ac.ir). S. Buzzi is with the Department of Electrical and Information Engineering, University of Cassino and Southern Latium, Cassino, Italy, and with Consorzio Nazionale Interuniversitario per le Telecomunicazioni
(CNIT), Parma, Italy (E-mail: buzzi@unicas.it).}}

\markboth{IEEE Transactions on Communications ,~Vol. 70, No. 2, FEBRUARY 2022}%
{Shell \MakeLowercase{\textit{et al.}}: Bare Demo of IEEEtran.cls for IEEE Journals}

\maketitle

% under the assumption
\vspace{-1cm}
\begin{abstract}
In this paper, uplink transmission of a cell-free massive multiple-input multiple-output (CF-mMIMO) system coexisting with device-to-device (D2D) communication links is investigated, under the assumption  that access points (APs) are equipped with low resolution analog-to-digital converters (ADCs). Lower bounds of achievable rates for both D2D users (DUEs) and CF-mMIMO users (CFUEs) are derived in closed-form, with perfect and imperfect channel state information. Next, in order to reduce pilot contamination,  greedy and graph coloring-based pilot allocation algorithms are proposed and analyzed for the considered scenario. Furthermore,  to control interference and improve the performance, two power control strategies are designed and their complexity and convergence are also discussed. The first power control strategy aims at maximizing CFUEs' sum spectral efficiency (SE) subject to quality of service constraints on DUEs, while the second one maximizes the weighted product of CFUEs' and DUEs' signal-to-interference-plus-noise-ratios (SINRs).  Numerical results show that the proposed pilot and power allocations bring a considerable improvement to the network SE. Also, it is revealed that the activation of D2D links has a positive effect on the system throughput, i.e. the network offloading ensured by the D2D links overcomes the increased interference brought by D2D communications. 

%\textcolor{red}{Our results indicate that at high signal-to-interference-plus-noise ratio regime and sufficiently large fronthaul capacity, estimating channels at the CPU rather than APs results in a smaller estimation error.} 
\end{abstract}
\IEEEpeerreviewmaketitle

\section{Introduction}
\lettrine{T}{he} use of large-scale antenna arrays at base stations, a solution commonly known as massive MIMO (mMIMO), has been one of the main technological innovations of fifth-generation systems. Indeed, in rich scattering environments, having a large number of antennas at a base station permits multiplexing, with simple beamforming schemes, several users on the same time-frequency slot, thus leading to remarkable improvements in the network throughput. Unfortunately, mMIMO is not capable of solving the problem of user performance disparity, since there is usually a large gap between the achievable rate for users that are located in the inner part of a radio cell, and that of users located at the cell borders, where large interference levels may be present. In order to overcome this problem, a new deployment architecture, named CF-mMIMO, proposed in recent years \cite{ngo2017cell,nayebi2017precoding}, is seriously considered as one of the main building blocks of future  beyond fifth generation and sixth generation (6G) wireless networks  \cite{rajatheva2020white,zhang2020prospective}. In CF-mMIMO, a large number of distributed antennas or APs are deployed in the coverage area to create macro diversity and to provide increased performance uniformly across users. The APs are connected to a central processing unit (CPU) through fronthaul links, and the use of the time-division-duplex (TDD) protocol permits avoiding channel estimation on the downlink. Moreover, in CF-mMIMO uplink channel estimates are retained at the APs and are used to compute locally the beamformers, thus avoiding an excessive load on the fronthaul links. 

Besides CF-mMIMO, D2D communications, originally introduced in 3GPP LTE  Release 12, have gained more and more importance over the years, and now there is a general consensus that they will be present in 6G networks as well \cite{liu2020graph,zhang2020envisioning}. 
 Indeed, in future densely populated network the chance of having users in a close proximity that want to communicate will not be negligible. So, by allowing these devices to communicate directly the performance of the communication is improved due to shorter distance between these devices compared with their distance from APs. This also contributes to reducing the network load and to improve the data rate and the delay w.r.t. the case in which communication flows through the APs \cite{zhang2020prospective,rajatheva2020white,zhang2020envisioning}.

Finally, in a CF-mMIMO system with dense AP deployment the distance between APs and UEs is not so large, and so D2D communications may happen at close distance from APs. The interference caused by D2D communications is thus larger than the traditional cellular deployment with macro-BS, and it is thus important to consider  co-existence issues between D2D links and CF-mMIMO links.
Otherwise stated, the simultaneous operation for CF-mMIMO and D2D links on the same carrier frequency  causes mutual interference and performance degradation, whereby proper resource management algorithms are to be employed\footnote{Although CF-mMIMO makes the user-AP distance smaller, in practice it is not always the case because we cannot mount the AP anywhere that we want. So, by deploying D2D communications not only the users will be able to enjoy low path loss in poorly covered regions but also the load on the cell-free infrastructure and its fronthaul will be alleviated by offloading the traffic of the users with the possibility of  establishing D2D connections in areas with densely active users.}. This consideration motivates the study that is here presented. \textcolor{black}{In the following three subsections related works on CF-mMIMO and D2D are reviewed and our motivation and contributions are presented.}

\subsection{CF-mMIMO Related Works}
Primary groundbreaking works on CF-mMIMO started with the seminal papers \cite{ngo2017cell} and  \cite{nayebi2017precoding}, which revealed its potential with respect to classical network deployments. The follow-up studies considered various aspects of CF-mMIMO, including its performance with different approaches \cite{buzzi2019user,bjornson2019making,zhang2019cellR,mai2018pilot,liu2019tabu,liu2020graph,jin2020spectral,interdonato2020local,liu2019spectral}, under different non-ideality circumstances \cite{hu2019cell,masoumi2019performance,zhang2018performanceR,zhang2020rfR,zheng2020efficientR}, and  its functioning in combination with other technologies \cite{alonzo2019energy,rezaei2020underlaid,d2020analysis,wang2019performance,bashar2019performance}. To be specific, \cite{buzzi2019user} studies a user-centric approach along with resource allocation strategies for uplink and downlink data rates and it shows tangible performance improvements compared to the cell-free scenario. In \cite{bjornson2019making}, a comprehensive investigation is conducted on the performance of different levels of cooperation among APs, and it turns out that with  global or local minimum mean square error (MMSE), CF-mMIMO outperforms the classical cellular counterpart significantly. 
%From EE perspectives, assuming per AP power and per user SE constraints, \cite{van2020joint} minimizes total downlink power consumption by joint power allocation and active AP selection for load balancing. 
In \cite{mai2018pilot}, authors present a pilot power allocation problem aimed at optimizing the channel estimation normalized total mean square error, with random pilot assignment and largest large-scale fading-based AP selection 
scheme. Also, \cite{liu2019tabu} and \cite{liu2020graph} apply tabu-search and graph coloring pilot allocation for CF-mMIMO, respectively. \cite{jin2020spectral} exploits beamformed downlink pilots in a correlated Rician fading CF-mMIMO system and proposes power optimization for the downlink of this system. \cite{interdonato2020local} proposes locally implementable zero-forcing (ZF) precoders and derives SE for the downlink while \cite{liu2019spectral} utilizes ZF combining for the uplink of CF-mMIMO and obtains SE expressions in the closed-form for perfect and imperfect channel state information (CSI).%Uplink SE of the CF-mMIMO using zero-forcing (ZF) detector with perfect and imperfect channel state information (CSI) is investigated in \cite{liu2019spectral} and several asymptotic results are derived. 
%In order to achieve green CF-mMIMO, \cite{femenias2020access} explores multiple heuristic AP on/off strategies for optimizing uplink and downlink energy efficiency of the system based on number and location of the active users.
\cite{hu2019cell} investigates the effect of low resolution ADCs in both APs and user equipment (UEs) for the downlink of CF-mMIMO and presents a max-min power control. In \cite{masoumi2019performance}, the uplink of CF-mMIMO with limited fronthaul capacity and hardware impairments at both APs and UEs are considered and the sum rate maximization problem is investigated. Finally, in \cite{zhang2018performanceR,zhang2020rfR,zheng2020efficientR} the effects of different hardware impairments in CF-mMIMO such as low-resolution ADCs, radio frequency impairments are studied.

In parallel, many researchers studied CF-mMIMO in coexistence with other technologies. 
%The total energy efficiency for the uplink and downlink of both a CF-mMIMO system and a user-centric system in millimeter wave frequency bands is maximized in \cite{alonzo2019energy}, assuming hybrid beamforming. 
In a spectrum sharing scenario, the downlink performance of CF-mMIMO system as a secondary network that is underlaid below a co-located mMIMO system with non-orthogonal multiple access (NOMA) technique is scrutinized in \cite{rezaei2020underlaid}. \cite{d2020analysis} inspects the support for unmanned aerial vehicles as well as ground users in CF-mMIMO networks for the uplink and downlink transmissions along with max-min power allocation. Moreover, SE of CF-mMIMO with full-duplex APs is analyzed and deterministic equivalents for uplink and downlink sum rates are also presented in \cite{wang2019performance}. Furthermore, authors of \cite{bashar2019performance} examine an adaptive mode switching between NOMA and orthogonal multiple access for the downlink of CF-mMIMO with max-min power control. 
%\textcolor{azure}{Also, \cite{rezaei2020rateR} investigates a NOMA-enabled CF-mMIMO with three linear precoders}. A deep learning approach has been recently employed for sum SE maximization by controlling transmit power of users in a fronthaul limited CF-mMIMO \cite{bashar2020exploiting}. Finally, in \cite{jin2019channel} authors explore the application of deep learning for channel estimation in millimeter wave bands in CF-mMIMO.

\subsection{D2D Related Works}
D2D communications have attracted a large share of interest in the recent past \cite{srinivasan2019joint,xu2016pilot,yang2016downlink,pan2017resourceR,he2017spectral,lin2015interplay,liu2018performance}. \cite{srinivasan2019joint,xu2016pilot} study the uplink of a single-cell mMIMO network with underlaid D2D users. \cite{srinivasan2019joint} maximizes the uplink sum data rates of cellular users with perfect CSI by jointly optimizing power and the resources subject to energy 
consumption constraint at the base station equipped with a low-resolution ADC, outage probability constraint for D2D users, and maximum transmit power. However, \cite{xu2016pilot} estimates the channel of D2D pairs using pilots that are orthogonal with the pilots of cellular users and are reused among D2D pairs and applies graph coloring strategy for pilot assignment, and proposes an optimization problem for minimizing sum power consumption of D2D transmitters subject to quality-of-service (QoS) for cellular users. 
%Then, instantaneous EE of D2D users subject to QoS constraints for cellular users is maximized. 
For a similar setting, the sum SE of D2D users is maximized in \cite{yang2016downlink}, with cellular users assumed to operate in downlink mode. In \cite{pan2017resourceR} for a power domain single-cell NOMA-based system with underlaid D2D users and full CSI knowledge, power allocation as well as channel assignments are applied to maximize the sum of instantaneous data rate for D2D pairs. \cite{he2017spectral} addresses open-loop power control for the uplink of multi-cell mMIMO systems with underlaid D2D pairs and without considering channel estimation or pilot transmission. 
%\cite{xu2017pilot} studies the sum SE maximization of D2D users subject to QoS constraints for cellular users in the uplink of single cell mMIMO, with channel estimation using pilots that are reused among D2D pairs but orthogonal with cellular users.
In \cite{lin2015interplay}, uplink multi-cell mMIMO system with underlaid D2D pairs is investigated; in particular, asymptotic and non-asymptotic SE of cellular and D2D users with perfect and imperfect CSI using orthogonal pilots and without power allocation is analysed. For D2D-based vehicle-to-vehicle (V2V) communications underlaid in the uplink of a sigle-cell mMIMO system, the SE of V2V users and cellular users with perfect CSI and using ZF and maximum-ratio combining (MRC) are derived in \cite{liu2018performance}. Next, a power optimization problem to maximize the sum SE of V2V users subject to QoS for cellular users is proposed. 
%In \cite{ghazanfari2018optimized}, uplink of multi-cell mMIMO system with underlaid D2D users is considered. Also, orthogonal pilots are used for channel estimation, where MRC and ZF are used for symbol detection and to derive the SE of both cellular UEs and D2D users. Finally, max-min fairness of the  users and maximization of the product of SINRs are addressed. 

\subsection{Contribution}
In the considered dense CF-mMIMO system, D2D and cell-free users operate in the same resources and there is increased likelihood that they happen to be in close proximity of each other and some APs. This can create a strong mutual interference between D2D and cell-free users not only in data transmission phase but also during channel training due to pilot contamination. In the related prior works, the behaviour of such system is not investigated, and it is not clear how this mutual coupling will affect the overall system performance. To this end, for the considered system we have analysed the performance of both cell-free and D2D users by deriving their SE in the closed form. These expressions reveal the mutual effect of cell-free and D2D users on each other's performance and therefore to control interference and improve the system performance we have proposed pilot assignment algorithms and power control optimization problems. Specifically, we present and solve two optimization problems: the first one maximizes the SE of CFUEs with QoS constraints on DUEs' SE, while the second one maximizes the weighted product of SINRs of DUEs and CFUEs. We also assume that there is a limited number of orthogonal pilots which are reused among DUEs and CFUEs for channel estimation, and thus, two pilot assignment algorithms are considered to manage pilot contamination. Also, low resolution ADCs are used to make cell-free system energy- and cost-efficient. While using low resolution ADCs negatively impacts the SE and channel estimation quality, our results indicate that utilizing moderate resolution ADCs (around 4 bits) can reduce their degrading effect to a large extent. User-centric approach also considered to make the system more realistic and scalable. The contribution of this paper can be summarized as follows.
\begin{itemize}
    \item For the uplink of user-centric CF-mMIMO with \textit{underlaid D2D users} and \textit{low resolution ADCs} at the APs, \textit{closed-form SE} formulas for both CFUEs and DUEs with perfect and imperfect CSI are derived.
    \item Since a limited number of orthogonal pilots are \textit{reused} among all the users, two pilot allocation algorithms, i.e. a \textit{greedy-based algorithm} for CFUEs and a \textit{graph coloring-based algorithm} for DUEs, are adopted in order to limit the pilot contamination effects.
\item \textit{Two power allocation strategies} are proposed to further improve the system performance. In the first one, sum SE of CFUEs are maximized subject to QoS for DUEs and maximum transmit power. In the second one, the weighted product of SINRs of CFUEs and DUEs is maximized subject to maximum transmit power of the users. For both problems, solutions based on geometric programming and on successive convex lower bound  maximization are proposed and their complexity and convergence are analyzed.
\item Finally, numerical results are provided to evaluate the performance of the system and the proposed resource allocation problems in the considered scenario.
\end{itemize}

%%%%%%%%%%%%%%%%%%%%%%%%%
\emph{Organization}: In the remainder of the article we present the system model in Section \ref{sec2sysmodel}. The performance analysis is carried out in Section \ref{sec3performance}, and pilot assignment and power control are addressed in Section \ref{sec4PilotandPower} and \ref{Sec:powerallocation}, respectively. Finally, numerical results are presented in Section \ref{sec5numerical}, while concluding remarks are given in Section \ref{sec6conclusion}.

\emph{Notation}: For matrices and vectors we use boldface uppercase and boldface lowercase letters, respectively. $\boldsymbol{x}\in \mathbb{C}^{N\times 1}$ denotes a vector in a $N$-dimensional complex space, $\delta_{ij}$ equals $1$ for $i=j$ and $0$ otherwise. Moreover, $(.)^{*}$, $(.)^{T}$ and $(.)^{H}$ are used for denoting conjugate, transpose and conjugate-transpose operators. Finally, $\mathcal{CN}(0,\sigma^{2})$ represents the zero-mean circularly symmetric complex Gaussian (CSCG) distribution with variance $\sigma^{2}$. 

%%%%%%%%%%%%%%%%%%%%%%%%%%%%%%%%%%%%%%%%%%%%%%%%%%%%%%%%%%%%%%%%%%%%%%%%%%%%%%%%%%%%%%%%%%%%%%%%%%%%%%%%%%%%%%%%%%%%%%%%%%%%%%%%%%%
\section{System Model}\label{sec2sysmodel}
We consider the uplink of CF-mMIMO system with underlaid D2D communications in which $K$ single-antenna CFUEs communicate with $M$ distributed single-antenna APs; simultaneously, $L$ D2D pairs communicate in the considered system as shown in Fig. \ref{fig:systmdl}. Similar to \cite{xu2017pilot} and \cite{lin2015interplay}, we assume a single-antenna transmitter, for instance DUE$_{l}^{\text{tx}}$, and an $N$-antenna receiver counterpart , i.e. DUE$_{l}^{\text{rx}}$, for D2D communications\footnote{Please note that the results can be straightforwardly extended to the full MIMO scenario.}. Note that $K\!\! \ll \!\!M$ and all the communications take place in the same time-frequency resource. The TDD protocol is used to exploit the channel reciprocity for reducing channel estimation overhead. Also, APs are assumed to be equipped with low resolution ADCs for deployment cost reduction. In addition, CPU is an aggregation node where the resource allocation is performed and the received signal from different APs are collected to estimate the transmitted symbol of each use.
\subsection{Channel Model}
We consider Rayleigh fading channel model which is constant in each coherence interval of length $T$ [samples], and changes independently from one coherence interval to another. The channel between the $k$th CFUE for $k \in \mathcal{K} =\{1,2,...,K\}$ or the transmitter of $l$th DUE pair for $l \in \mathcal{L} =\{1,2,...,L\}$ and the $m$th AP (AP$_m$) is modeled by $h_{mk}^{c}\sim\mathcal{CN}(0,\beta_{mk}^{c})$  and $h_{ml}^{d}\sim\mathcal{CN}(0,\beta_{ml}^{d})$, respectively. Moreover, the channel between those transmitters and the receiver of the $l^{\prime}$th DUE pair is given by $\boldsymbol{g}_{l^{\prime}k}^{c}\sim\mathcal{CN}(0,\psi_{l^{\prime}k}^{c}\boldsymbol{I}_{N})$ and $\boldsymbol{g}_{l^{\prime}l}^{d}\sim\mathcal{CN}(0,\psi_{l^{\prime}l}^{d}\boldsymbol{I}_{N})$, where $\boldsymbol{I}_{N}$ is the $N\times N$ identity matrix and $\beta_{ml}^{d},$ $\beta_{mk}^{c},$ $\psi_{l^{\prime}l}^{d},$ $\psi_{l^{\prime}k}^{c}$ account for the large-scale fading coefficients. 

\begin{figure}[!t]
\centering
\psfragfig[scale=.3]{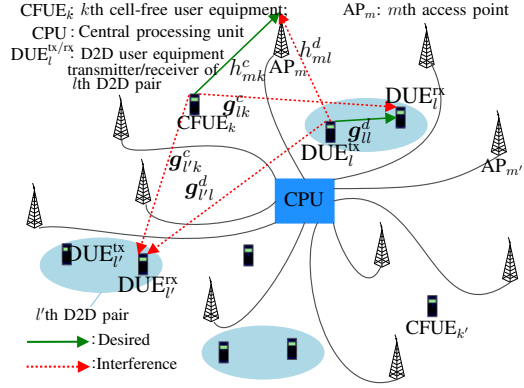}{\psfrag{gmk}[][][0.73]{\ \ $h_{mk}^{c}$}
\psfrag{gklpc}[][][0.73]{\!\!$\boldsymbol{g}_{l^{\prime}k}^{c}$}
\psfrag{gklc}[][][0.73]{$\boldsymbol{g}_{lk}^{c}$}
\psfrag{APm}[][][0.63]{\ AP\textsubscript{$m$}}
\psfrag{gmld}[][][0.73]{\!\!\!\!$h_{ml}^{d}$}
\psfrag{UEk}[][][0.63]{CFUE\textsubscript{$k$}}
\psfrag{Dlt}[][][0.73]{DUE$_{l}^{\text{{tx}}}$}
\psfrag{Dlr}[][][0.73]{\!DUE$_{l}^{\text{{rx}}}$}
\psfrag{glpld}[][][0.73]{$\boldsymbol{g}_{l^{\prime}l}^{d}$}
\psfrag{Dlpt}[][][0.73]{DUE$_{l^{\prime}}^{\text{{tx}}}$}
\psfrag{Dlpr}[][][0.73]{DUE$_{l^{\prime}}^{\text{{rx}}}$}
\psfrag{APmp}[][][0.63]{AP\textsubscript{$m^{\prime}$}}
\psfrag{UEkp}[][][0.63]{CFUE\textsubscript{$k^{\prime}$}}
\psfrag{glld}[][][0.73]{$\boldsymbol{g}_{ll}^{d}$}
\psfrag{CPU}[][][0.65]{CPU}
\psfrag{UEk1}[][][0.57]{CFUE\textsubscript{$k$}}
\psfrag{Dltr}[][][0.57]{DUE$_{l}^{\text{{tx}}/\text{{rx}}}$}
\psfrag{CPU1}[][][0.57]{CPU}
\psfrag{CPU11111111111}[][][0.55]{\ \ \ \ \ Central processing unit}
\psfrag{UEk11111111111}[][][0.55]{\ \ \ \ \ \ \ \ \ \ \ \ \ \ \ \ \ \ \ \ \ \ \ \ \ \ \ \ \ \ \ \ \ \ \ \ \ \ \ \ \ \ \ \ \ \ \ \ \  $k$th cell-free user equipment;\ \ \ \ \ \ \ \ \ AP$_m$: $m$th access point}
\psfrag{Dltr1111111111}[][][0.53]{\ \ \ \ \ D2D user equipment}
\psfrag{Dltr1111111}[][][0.53]{transmitter/receiver of}
\psfrag{Dltr111111}[][][0.53]{$l$th D2D pair}
\psfrag{Dlp}[][][0.53]{$l^{\prime}$th D2D pair}
\psfrag{Desired}[][][0.53]{Desired}
\psfrag{Interference}[][][0.53]{Interference}}
\vspace{-0.2cm}
\caption{{Cell-free mMIMO with underlaid D2D communications system model.}}\label{fig:systmdl}
\end{figure}

%%%%%%%%%%%%%%%%%%%%%%%%%%%%%%%%%%%%%%%%%%%%%%%%%%%%%%%%%%%%%
\subsection{Modelling Impacts of Low Resolution ADC}
The received signal at the AP$_m$ and the receiver of $l$th DUE pair are respectively given by
\begin{subequations}
\begin{align}
\bar{y}_{m}^{c} = \sqrt{\rho^{c}}\sum\limits_{k =1}^{K}\sqrt{\eta_{k}^{c}}h_{mk}^{c}s_{k}^{c} + \sqrt{\rho^{d}}\sum\limits_{l^{\prime} =1}^{L}\sqrt{\eta_{l^{\prime}}^{d}}h_{ml^{\prime}}^{d}s_{l^{\prime}}^{d} + n_{m}^{c},\label{eqn:IdlADCap}\\
\boldsymbol{y}_{l}^{d} = \sqrt{\rho^{c}}\sum\limits_{k =1}^{K}\sqrt{\eta_{k}^{c}}\boldsymbol{g}_{lk}^{c}s_{k}^{c} + \sqrt{\rho^{d}}\sum\limits_{l^{\prime} =1}^{L}\sqrt{\eta_{l^{\prime}}^{d}}\boldsymbol{g}_{ll^{\prime}}^{d}s_{l^{\prime}}^{d} + \boldsymbol{n}_{l}^{d},\label{eqn:IdlADCd2d}
\end{align}
\end{subequations}
where $s_{k}^{c}\sim\mathcal{CN}(0,1)$ and $s_{l^{\prime}}^{d}\sim\mathcal{CN}(0,1)$ are the transmitted symbol by the $k$th CFUE and the transmitter of $l^{\prime}$ D2D pair, which is assumed to be independent and identically distributed (i.i.d) for different users. Furthermore, $\rho^{c},\ \rho^{d}$ and $\eta_{k}^{c},\ \eta_{l^{\prime}}^{d}$ denote the corresponding maximum transmit power and the power control coefficients of the users, respectively. Also,  $n_{m}^{c}\sim\mathcal{CN}(0,N_{0})$ is the additive white Gaussian noise (AWGN) at the $m$th AP and $\boldsymbol{n}_{l}^{d}$ is an $N\times 1$ vector of i.i.d AWGN random variables distributed according to $\mathcal{CN}(0,N_{0})$.
Since low resolution ADCs are used at the APs, the received signal \eqref{eqn:IdlADCap} for the $m$th AP is actually written as follows \cite{zhang2017mixedADCrician}
\begin{equation}\label{eqn:LADC}
  y_{m}^{c} = \xi \bar{y}_{m}^{c} + q_{m} = \xi\sqrt{\rho^{c}}\sum\limits_{k =1}^{K}\sqrt{\eta_{k}^{c}}h_{mk}^{c}s_{k}^{c} + \xi\sqrt{\rho^{d}}\sum\limits_{l^{\prime} =1}^{L}\sqrt{\eta_{l^{\prime}}^{d}}h_{ml^{\prime}}^{d}s_{l^{\prime}}^{d} + \xi n_{m}^{c} + q_{m}.
\end{equation}
In \eqref{eqn:LADC} we have used the so called additive quantization noise model (AQNM). Here, $q_{m}$ accounts for the quantization noise which is uncorrelated with $\bar{y}_{m}$. Given the channel realizations, for a non-uniform quantizer the variance of $q_{m}$ is computed as follows \cite{zhang2017mixedADCrician,zhou2016fronthaul}
%has the following variance $Q_{m} = \mathbb{E}\{q_{m}^{*}q_{m}|\{h_{mk}^{c},h_{ml}^{d}\}\}$ \cite{zhang2017mixedADCrician,zhou2016fronthaul}; with
\begin{equation}\label{eqn:Qvariance}
  Q_{m} = (1-\xi)\xi\mathbb{E}\{\bar{y}_{m}^{c^*}\bar{y}_{m}^{c}|\{h_{mk}^{c},h_{ml}^{d}\}\}
  = (1-\xi)\xi\left(\rho^{c}\sum\limits_{k =1}^{K}\eta_{k}^{c}\left|h_{mk}^{c}\right|^2 + \rho^{d}\sum\limits_{l^{\prime} =1}^{L}\eta_{l^{\prime}}^{d}\left|h_{ml^{\prime}}^{d}\right|^{2} + N_{0}\right).
\end{equation}
In \eqref{eqn:Qvariance}, $\xi$ can be specified in terms of the number of ADC quantization bits $b$; for $b>5$, $\xi$ is computed as $\xi = 1 - \frac{\pi\sqrt{3}}{2}2^{-2b}$ and for other values of $b$ it can be obtained from \cite[Table I]{zhang2017mixedADCrician}. 
% \textcolor{purple}{ Using the presented models in this section, now we can analyze and derive the SE of the system which is a crucial performance metric.}%Table \ref{tab:Quantization} \cite{zhang2017mixedADCrician}.
%\begin{table}[!t]
%\centering
%\caption{Values of $\xi$ for $b \in \{1,2,3,4,5\}$.}
%\label{tab:Quantization}
%\scalebox{0.9}{
%\begin{tabular}{c|ccccc}
%\hline
%$b$   & $1$      & $2$      & $3$       & $4$        & $5$        \\ \hline
%$\xi$ & $0.6366$ & $0.8825$ & $0.96546$ & $0.990503$ & $0.997501$ \\ \hline
%\end{tabular}%
%}
%\end{table}
\section{Achievable Rate Analysis}\label{sec3performance}
In this section, the uplink achievable rate of the system for the CFUEs and for the DUEs are derived under perfect CSI. Next, imperfect CSI is obtained using uplink channel training and the corresponding achievable data rates are derived. \textcolor{black}{These derivations are important in designing resource allocation algorithms for power control and pilot assignment as well as analysing and gaining insights on the impacts of using low resolution ADCs.}  % and several special scenarios of the system are investigated.
\subsection{Uplink Achievable Rate with Perfect CSI}
\subsubsection{Achievable rate of CFUEs} When perfect CSI is available\footnote{When low resolution ADCs are used, there will be an error floor for the estimated channels even with high power orthogonal pilots \cite{zhang2018low}. However, by assuming that one uses high resolution ADCs only in the channel estimation phase, perfect CSI can be achieved by sending high power and orthogonal pilots. Hence, here studying the perfect CSI case provides  a performance upper limit to the performance of the system with low resolution ADCs during channel estimation and can be used as a benchmark to evaluate the performance of the low resolution ADC-aware channel estimation algorithms in the future works.}, by using \eqref{eqn:LADC} and MRC receiver the following approximation of transmitted symbol for the $k$th CFUE at the CPU can be derived.
\begin{equation}\label{eqn:PCSI_MRC}
\begin{split}
  r_{k}^{c} &= \sum\limits_{m\in\mathcal{M}_{k}}^{}h_{mk}^{c^{*}}y_{m}^{c} =  \underbrace{\xi\sqrt{\eta_{k}^{c}\rho^{c}}\sum\limits_{m\in\mathcal{M}_{k}}^{}|h_{mk}^{c}|^{2}}_{\text{DS\textsubscript{$k$}: desired signal}}s_{k}^{c}   +   \mathcal{I}_{k}^{c},\\
  \mathcal{I}_{k}^{c} &= \underbrace{\xi\!\sqrt{\!\rho^{c}}\sum\limits_{k^{\prime}\neq k} ^ {K}\!\! \sqrt{\eta_{k^{\prime}}^{c}}\sum\limits_{m\in\mathcal{M}_{k}}^{}h_{mk}^{c^*}h_{mk^{\prime}}^{c}s_{k^{\prime}}^{c}}_{\text{ICFUE\textsubscript{$k$}: interference from CFUEs}} + \underbrace{\xi\!\sqrt{\!\rho^{d}}\sum\limits_{l^{\prime}=1} ^ {L}\!\! \sqrt{\!\eta_{l^{\prime}}^{d}}\sum\limits_{m\in\mathcal{M}_{k}}^{}\!\!h_{mk}^{c^*}h_{ml^{\prime}}^{d}s_{l^{\prime}}^{d}}_{\text{IDUE\textsubscript{$k$}: interference from DUEs}} +  \underbrace{\xi\!\sum\limits_{m\in\mathcal{M}_{k}}^{}\!h_{mk}^{c^*}n_{m}}_{\text{TN\textsubscript{$k$}: total noise}} +\!\!\!\!\! \underbrace{\sum\limits_{m\in\mathcal{M}_{k}}^{}\!h_{mk}^{c^*}q_{m}}_{\text{QN\textsubscript{$k$}: quantization noise}}\!\!\!,
  \end{split}
\end{equation}
 where $\mathcal{I}_{k}^{c}$ consists of inter-user interference from both CFUEs and D2D transmitters, the additive channel noise and the quantization noise as the result of deploying low resolution ADCs at the APs. It can be shown that for a given channel realization all the terms in \eqref{eqn:PCSI_MRC} are mutually uncorrelated. 
 To make the system scalable, each UE is served by a limited number of APs.
 This approach is known as user-centric (UC) method \cite{buzzi2019user}. So the $m$th AP serves only $\mathcal{K}_{m}\subset \mathcal{K}$ users based on the strength of the channel coefficients. We also denote by $\mathcal{M}_{k}$ the set of APs serving the $k$th user. In this paper, the set $\mathcal{K}_{m}$ is determined by $\sum\limits_{i=1}^{|\mathcal{K}_{m}|}\frac{\check{\beta}_{mi}^{c}}{\sum_{i=1}^{K}\beta_{mi}^{c}}\geq \delta$\footnote{We have used this criterion to circumvent the prohibitive computation of combinatorial user-AP association. This criterion is also used for other tasks; like distinguishing among weak cell-edge and strong cell-center users in \cite{zhu2015soft}.}
% \begin{equation*}
%     \sum\limits_{i=1}^{|\mathcal{K}_{m}|}\frac{\check{\beta}_{mi}^{c}}{\sum_{i=1}^{K}\beta_{mi}^{c}}\geq \delta,
% \end{equation*}
where ${\check{\beta}_{m1}^{c}>\check{\beta}_{m2}^{c}> ... >\check{\beta}_{mK}^{c}}$ are the sorted version of the large-scale fading coefficients between the users and the $m$th AP in descending order. Furthermore, the threshold $\delta \in [0,1]$. 
%\textcolor{violet}{ \textbf{\large To be removed:} So, the ergodic achievable rate of the $k$th CFUE with perfect CSI is expressed as
% \begin{equation}\label{eqn:PCSI_Rate1}
%   R_{k}^{\text{CFUE\textsubscript{P}}} =  \mathbb{E}\left\lbrace \log_{2}\left(1 + \dfrac{\left\lvert \text{DS}_{k}\right\rvert^{2}}{ \text{Var}(\mathcal{I}_{k}^{c}) }\right) \right\rbrace,
% \end{equation}
% where the superscript P, stands for \textit{perfect} CSI and Var$(.)$ indicates the variance operator. Given the perfect CSI, the variance of $\mathcal{I}_{k}^{c}$ can be written as
% \begin{equation}\label{PCSI_Interference}
% \begin{split}
%     \text{Var}(\mathcal{I}_{k}^{c}) &=\!
%     \xi^{2}\rho^{c}\sum\limits_{k^{\prime}\neq k} ^ {K}\!\! \eta_{k^{\prime}}^{c}\left\lvert\sum\limits_{m=1}^{M}h_{mk}^{c^*}h_{mk^{\prime}}^{c}\right\lvert^{2} + \xi^{2}\!\rho^{d}\sum\limits_{l^{\prime}=1} ^ {L}\! \!\eta_{l^{\prime}}^{d}\left\lvert\sum\limits_{m=1}^{M}\!\!h_{mk}^{c^*}h_{ml^{\prime}}^{d}\right\rvert^{2} +  N\xi\!\sum\limits_{m=1}^{M}\!\left\lvert h_{mk}^{c}\right\lvert^{2}\\
%     &+  (1-\xi)\xi\left(\rho^{c}\sum\limits_{k =1}^{K}\eta_{k^{\prime}}^{c}\sum\limits_{m=1}^{M}|h_{mk}^{c}|^{2}\left|h_{mk^{\prime}}^{c}\right|^2 + \rho^{d}\sum\limits_{l^{\prime} =1}^{L}\eta_{l^{\prime}}^{d}\sum\limits_{m=1}^{M}|h_{mk}^{c}|^{2}\left|h_{ml^{\prime}}^{d}\right|^{2}\right).
%     \end{split}
% \end{equation}
% Notice that in computing the above variance we also used \eqref{eqn:Qvariance}.} 
Next, in order to obtain a \textit{closed-form} expression for the achievable rate, the well-known \textit{use and then forget} (UatF) technique \cite{bjornson2017massive} is applied to the statistic \eqref{eqn:PCSI_MRC}, and results in 
\begin{equation}\label{eqn:PCSI_UatF1}
r_{k} = \E\{\text{DS}_{k}\}s_{k} + \text{BU}_{k}s_{k} + \mathcal{I}_{k}^{c},    
\end{equation}
where $\text{BU}_{k} = \{\text{DS}_{k} -\E\{\text{DS}_{k}\}\}$ stands for beamforming uncertainty which is caused by using only channel statistics for data detection. Since all the terms in \eqref{eqn:PCSI_UatF1} are also mutually uncorrelated, by considering that the last two interfering terms follow the worst case Gaussian distribution, the achievable rate is given by 
\begin{equation}\label{eqn:PCSI_UatF}
    R_{k}^{\text{CFUE\textsubscript{p}}} = \log_{2}\left(1 + \dfrac{\lvert\E\{\text{DS}_{k}\}\rvert^{2}}{\text{Var}(\text{BU}_{k}) + \text{Var}(\mathcal{I}_{k}^{c}) }\right).
\end{equation}
Throughout the paper the subscripts ``p'' and ``ip'' stand for \textit{perfect} and \textit{imperfect} CSI, respectively, and Var$(.)$ indicates the variance operator.
\begin{theorem} The closed-form  achievable rate of $k$th CFUE with perfect CSI is
 \begin{equation}\label{eqn:PCSI_apx}
     \!\!R_{k}^{\text{CFUE\textsubscript{p}}} =\! \log_{2}\!\!\left(\!\!1\! +\! \dfrac{\xi\eta_{k}^{c}\rho^{c}\left(\sum\limits_{m\in\mathcal{M}_{k}}^{}\beta_{mk}^{c}\right)^{2}}{ \rho^{c}\! \sum\limits_{k^{\prime}=1} ^ {K}\!\! \eta_{k^{\prime}}^{c}\!\sum\limits_{m\in\mathcal{M}_{k}}^{}\!\! \beta_{mk}^{c}\beta_{mk^{\prime}}^{c} \!+\! \rho^{d}\! \sum\limits_{l^{\prime}=1} ^ {L}\! \eta_{l^{\prime}}^{d}\!\!\sum\limits_{m\in\mathcal{M}_{k}}^{}\!\! \beta_{mk}^{c}\beta_{ml^{\prime}}^{d} \!+\! (1\!-\!\xi)\rho^{c}\eta_{k}^{c}\!\!\sum\limits_{m\in\mathcal{M}_{k}}^{}\! \beta_{mk}^{c^{2}} \!+\! N_{0}\!\!\sum\limits_{m\in\mathcal{M}_{k}}^{}\!\! \beta_{mk}^{c}}\!\!\right)\!\!.
 \end{equation}
\end{theorem}
\begin{proof}
See Appendix A.
\end{proof}
\begin{Remark}
When the bits of ADCs tend to infinity, i.e. $b\rightarrow\infty$, the SE of CFUEs reduces to 
\begin{equation}\label{eqn:PCSI_infb}
     \!\!R_{k,\text{\textbf{I}}}^{\text{CFUE\textsubscript{p}}} =\! \log_{2}\!\!\left(\!\!1\! +\! \dfrac{\eta_{k}^{c}\rho^{c}\left(\sum\limits_{m\in\mathcal{M}_{k}}^{}\beta_{mk}^{c}\right)^{2}}{ \rho^{c}\! \sum\limits_{k^{\prime}=1} ^ {K}\!\! \eta_{k^{\prime}}^{c}\!\sum\limits_{m\in\mathcal{M}_{k}}^{}\!\! \beta_{mk}^{c}\beta_{mk^{\prime}}^{c} \!+\! \rho^{d}\! \sum\limits_{l^{\prime}=1} ^ {L}\! \eta_{l^{\prime}}^{d}\!\!\sum\limits_{m\in\mathcal{M}_{k}}^{}\!\! \beta_{mk}^{c}\beta_{ml^{\prime}}^{d} \!+\! N_{0}\!\!\sum\limits_{m\in\mathcal{M}_{k}}^{}\!\! \beta_{mk}^{c}}\!\!\right)\!\!.
 \end{equation}
The above result follows from the fact that $\lim\limits_{b \to \infty} \xi = 1$. Additionally, by setting $\mathcal{M}_k=\{1,2,...,M\}$ and $\rho^{d}=0$, equation \eqref{eqn:PCSI_infb} reduces to the SE of original CF-mMIMO \cite{ngo2017cell}.
\end{Remark}
\begin{Remark}
When users' transmit power increases without bound, i.e. $\rho^{c}\rightarrow\infty$ and $\rho^{d}\rightarrow\infty$ with power control coefficients equal to one, CFUEs' SE reduces to \begin{equation}\label{eqn:PCSI_infp}
     \!\!R_{k,\text{\textbf{II}}}^{\text{CFUE\textsubscript{p}}} =\! \log_{2}\!\!\left(\!\!1\! +\! \dfrac{\xi\left(\sum\limits_{m\in\mathcal{M}_{k}}^{}\beta_{mk}^{c}\right)^{2}}{ \! \sum\limits_{k^{\prime}=1} ^ {K}\sum\limits_{m\in\mathcal{M}_{k}}^{}\!\! \beta_{mk}^{c}\beta_{mk^{\prime}}^{c} \!+\!  \sum\limits_{l^{\prime}=1} ^ {L}\sum\limits_{m\in\mathcal{M}_{k}}^{}\!\! \beta_{mk}^{c}\beta_{ml^{\prime}}^{d} \!+\! (1\!-\!\xi)\!\!\sum\limits_{m\in\mathcal{M}_{k}}^{}\! \beta_{mk}^{c^{2}}}\!\!\right)\!\!.
 \end{equation}
From \eqref{eqn:PCSI_infp} we observe that by using low resolution ADCs not only an additional interference is added to the denominator of the  SINR, the numerator also scales down linearly by $0<\xi< 1$ and its impact becomes more severe for coarser ADCs. Even by ignoring inter-user interference, degrading effect of low resolution ADCs will be present when users transmit with high power.
\end{Remark}
\subsubsection{Achievable rate of DUEs} By having perfect CSI, MRC combining technique can be applied at the receiver of the $l$th D2D pair, i.e. equation \eqref{eqn:IdlADCd2d}, which leads to
\begin{equation}\label{eqn:PCSI_MRCd2d}
  r_{l}^{d} = \boldsymbol{g}_{ll}^{d^{H}}\boldsymbol{y}_{l}^{d} = \underbrace{\sqrt{\rho^{d}\eta_{l}^{d}}\left\|\boldsymbol{g}_{ll}^{d}\right\|^{2}}_{\text{DS\textsubscript{$l$}}}s_{l}^{d} \!+\! \underbrace{\sqrt{\!\rho^{c}}\sum\limits_{k =1}^{K}\!\!\sqrt{\eta_{k}^{c}}\boldsymbol{g}_{ll}^{d^{H}}\boldsymbol{g}_{lk}^{c}s_{k}^{c}}_{\text{ICFUE\textsubscript{$l$}}} + \underbrace{\sqrt{\rho^{d}}\sum\limits_{l^{\prime} \neq l}^{L}\!\!\sqrt{\eta_{l^{\prime}}^{d}}\boldsymbol{g}_{ll}^{d^{H}}\boldsymbol{g}_{ll^{\prime}}^{d}s_{l^{\prime}}^{d}}_{\text{IDUE\textsubscript{$l$}}} + \underbrace{\boldsymbol{g}_{ll}^{d^{H}}\boldsymbol{n}_{l}^{d}}_{\text{TN\textsubscript{$l$}}}.
\end{equation}
Since for given channel realizations, the interference terms, i.e. ICFUE\textsubscript{$l$}, IDUE\textsubscript{$l$} and TN\textsubscript{$l$}, and the desired signal follow a Gaussian distribution and are mutually independent from one another, the ergodic achievable rate for the receiver of $l$th D2D pair is derived as
 \begin{equation}\label{eqn:PCSI_d2d1}
     R_{l}^{\text{DUE\textsubscript{p}}} =\! \mathbb{E}\left\lbrace\log_{2}\!\left(\!1\! +\! \dfrac{\rho^{d}\eta_{l}^{d}\left\|\boldsymbol{g}_{ll}^{d}\right\|^{4}}{\rho^{c}\sum\limits_{k =1}^{K}\!\eta_{k}^{c}\left|\boldsymbol{g}_{ll}^{d^{H}}\boldsymbol{g}_{lk}^{c}\right|^{2} \!+\! \rho^{d}\sum\limits_{l^{\prime} \neq l}^{L}\!\eta_{l^{\prime}}^{d}\left|\boldsymbol{g}_{ll}^{d^{H}}\boldsymbol{g}_{ll^{\prime}}^{d}\right|^{2} \!+\! N_{0}\left\|\boldsymbol{g}_{ll}^{d}\right\|^{2}}\!\right)\!\right\rbrace.
 \end{equation}
 Next, we use the Jensen's inequality\footnote{Note that, for D2D users the UatF bounding leads to a considerably underestimated achievable rate mainly because D2D users are not equipped with large number of antennas. In contrast, thanks to the large number of access points in the cell-free system, the UatF provides a tight bound for the achievable rate of cell-free users \cite{marzetta2016fundamentals}}, given by \eqref{eqn:jensen}, to obtain a closed-form expression for \eqref{eqn:PCSI_d2d1} which is provided in Theorem 2.
 \begin{equation}\label{eqn:jensen}
     \log\left(1 + \dfrac{1}{\mathbb{E}\left\lbrace x \right\rbrace}\right) \leq \mathbb{E}\left\lbrace \log\left(1 + \dfrac{1}{x}\right) \right\rbrace,
 \end{equation}
 \begin{theorem} The achievable rate of $l$th DUE with perfect CSI and $N\geq 2$ is
 \begin{equation}\label{eqn:PCSI_apx}
     \!\!\tilde{R}_{l}^{\text{DUE\textsubscript{p}}} =\! \log_{2}\!\!\left(\!\!1\! +\! \dfrac{\rho^{d}\eta_{l}^{d}\psi_{ll}^{d}(N-1)}{\rho^{c}\sum\limits_{k =1}^{K}\!\eta_{k}^{c}\psi_{lk}^{c} \!+ \rho^{d}\sum\limits_{l^{\prime} \neq l}^{L}\!\eta_{l^{\prime}}^{d}\psi_{ll^{\prime}}^{d} \!+\! N_{0}}\!\!\right)\!\!.
 \end{equation}
\end{theorem}
\begin{proof}
See Appendix B.
\end{proof}
%%%%%%%%%%%%%%%%%%%%%%%%%%%%%%%%%%%%%%%%%%%%%%%%%%%%%%%%%%%%%%%%%%%%%%%%%%%%%%%%%%%%%%%%%%%%%%%%%%%%%%%%%%%%%%%%%%%%%%%%%%%%
\subsection{Uplink Achievable Rate with Imperfect CSI}
In this subsection we first present the uplink channel estimation procedure, and then the achievable data rates of CFUEs and DUEs are derived using the obtained estimates.
\subsubsection{Channel estimation}
For obtaining channel estimates, $\tau$-length orthogonal pilot sequences, denoted by $\boldsymbol{\Phi} = \left\lbrace\boldsymbol{\phi}_{1}, \boldsymbol{\phi}_{2}, ..., \boldsymbol{\phi}_{\tau}\right\rbrace$, are considered, where $\boldsymbol{\phi}_{u}^{H}\boldsymbol{\phi}_{v}=\delta_{uv}$  and $\boldsymbol{\phi}_{u}\!\in\!\mathbb{C}^{\tau\times 1},\ \{u,v\} = 1,2,...,\tau$. Hence, the channel estimation overhead is $\varsigma = \frac{T-\tau}{T}$. The assigned pilots for CFUE $k$ and DUE $l$ are denoted by $\boldsymbol{\omega}_{k}\in\boldsymbol{\Phi}$ and $\boldsymbol{\theta}_{l}\in\boldsymbol{\Phi}$, respectively. Also, the total transmit power and the power control coefficients of the $k$th CFUE and the $l$th DUE are indicated by $\rho_{p}^{c},\ \mu_{k}^{c}$ and $\rho_{p}^{d},\ \mu_{l}^{d}$, respectively. Thus, the $m$th AP receives a $\tau\times 1$ vector $\boldsymbol{y}_{p,m}^{c}$, and the receiver of the $l$th D2D pair receives an $N\times\tau$ matrix $\boldsymbol{Y}_{p,l}^{d}$ as follows
\begin{subequations}
\begin{align}
\boldsymbol{y}_{p,m}^{c} &= \xi\sqrt{\tau\rho_{p}^{c}} \sum\limits_{k =1}^{K}\sqrt{\mu_{k}^{c}}h_{mk}^{c}\boldsymbol{\omega}_{k} + \xi\sqrt{\tau\rho_{p}^{d}}\sum\limits_{l^{\prime} =1}^{L}\sqrt{\mu_{l^{\prime}}^{d}}h_{ml^{\prime}}^{d}\boldsymbol{\theta}_{l^{\prime}} + \xi \boldsymbol{n}_{p,m}^{c} + \boldsymbol{q}_{p,m},\label{eqn:Rxpilot}\\
\boldsymbol{Y}_{p,l}^{d} &= \sqrt{\tau\rho_{p}^{c}} \sum\limits_{k =1}^{K}\sqrt{\mu_{k}^{c}}\boldsymbol{g}_{lk}^{c}\boldsymbol{\omega}_{k}^{H} + \sqrt{\tau\rho_{p}^{d}}\sum\limits_{l^{\prime} =1}^{L}\sqrt{\mu_{l^{\prime}}^{d}}\boldsymbol{g}_{ll^{\prime}}^{d}\boldsymbol{\theta}_{l^{\prime}}^{H} +  \boldsymbol{N}_{p,l}^{d}.\label{eqn:Rxpilotd2d}
\end{align}
\end{subequations}
In the above equations $\boldsymbol{q}_{p,m}$ is a $\tau\times 1$ quantization noise vector whose covariance matrix is defined as $\boldsymbol{Q}_{p,m} = (1-\xi)\xi\mathbb{E}\left\lbrace \bar{\boldsymbol{y}}_{p,m}^{c} \bar{\boldsymbol{y}}_{p,m}^{c^{H}}\right\rbrace$, with $\bar{\boldsymbol{y}}_{p,m}^{c}$ the received signal at the $m$th AP for the case of infinite resolution ADCs. Furthermore, the $\tau\times 1$ vector $\boldsymbol{n}_{p,m}^{c}$ and the $N\times\tau$ matrix $\boldsymbol{N}_{p,l}^{d}$ are additive white noises contributions with i.i.d entries distributed according to $\mathcal{CN}(0,N_{0})$. After projecting the received signals onto the used pilot sequences the channel between $k$th CFUE and AP $m$ is estimated using linear MMSE (LMMSE)\footnote{Some algorithms based on least-squares, expectation-maximization, maximum likelihood, and joint channel-and-data are used to alleviate low resolution ADCs impact on channel estimation \cite{zhang2018low}. Many of these algorithms rely on long pilot sequence lengths or they have high complexity \cite{zhang2018low}. Here, we focus on the co-existence of CF-mMIMO and D2D communications where low resolution ADCs are used to make the system cost- and power-efficient. From this perspective, our results provide a worst-case system performance using simpler MMSE estimation technique, which can be used as a benchmark for evaluating the effectiveness of other schemes in future works.} as follows %\cite[chapter 7]{kay1993fundamentals}
\begin{equation}\label{eqn:ChEst}
 \hat{h}_{mk}^{c} = \frac{\sqrt{\tau\rho_{p}^{c}\mu_{k}^{c}} \beta_{mk}^{c}}{\tau\rho_{p}^{c}\sum\limits_{k^{\prime} =1}^{K}\mu_{k^{\prime}}^{c}\beta_{mk^{\prime}}^{c}\left|\boldsymbol{\omega}_{k}^{H}\boldsymbol{\omega}_{k^{\prime}}\right|^{2} + \tau\rho_{p}^{d}\sum\limits_{l^{\prime} =1}^{L}\mu_{l^{\prime}}^{d}\beta_{ml^{\prime}}^{d}\left|\boldsymbol{\omega}_{k}^{H}\boldsymbol{\theta}_{l^{\prime}}\right|^{2} + N_{0}}y_{p,mk}^{c} = \lambda_{mk}^{c}y_{p,mk}^{c},
\end{equation}
where $y_{p,mk}^{c} \!=\! \boldsymbol{\omega}_{k}^{H}\boldsymbol{y}_{p,m}^{c}$ and the variance of the channel estimate is given by $\gamma_{mk}^{c} \!=\! \xi\!\sqrt{\tau\rho_{p}^{c}\mu_{k}^{c}} \beta_{mk}^{c}\lambda_{mk}^{c}$.
\begin{Remark}
Mean square error (MSE) of the estimated channel of any $k$th CFUE tend to the following limit when transmitted pilot power grows without bound.
\begin{equation}\label{eqn:mseLimit}
 \lim\limits_{\{\rho_{p}^{c}, \rho_{p}^{d}\} \to \infty} \text{MSE}_{mk} = \beta_{mk}^{c}\left(1-\xi\frac{ \beta_{mk}^{c}}{\sum\limits_{k^{\prime} =1}^{K}\!\!\beta_{mk^{\prime}}^{c}\!\left|\boldsymbol{\omega}_{k}^{H}\boldsymbol{\omega}_{k^{\prime}}\right|^{2} + \sum\limits_{l^{\prime} =1}^{L}\!\!\beta_{ml^{\prime}}^{d}\!\left|\boldsymbol{\omega}_{k}^{H}\boldsymbol{\theta}_{l^{\prime}}\right|^{2}}\right)\!,
\end{equation}
where $\text{MSE}_{mk}=\mathbb{E}\left\lbrace\left|h_{mk}^{c}-\hat{h}_{mk}^{c}\right|^2\right\rbrace$. Based on \eqref{eqn:mseLimit}, in addition to the pilot contamination, using low resolution ADCs also leads to channel estimation degradation  that does not disappear by increasing the pilot power. Even by using orthogonal pilots, there will be an error floor due to the low resolution ADC utilization. When orthogonal pilots are used, the error floor is equal to $\beta_{mk}^{c}\left(1-\xi\right)$ which will vanish only when high precision ADCs are used.
\end{Remark}
The estimation of the D2D channels are obtained from \eqref{eqn:Rxpilotd2d} as follows
\begin{equation}\label{eqn:ChEstd2d}
 \hat{\boldsymbol{g}}_{ll}^{d} = \frac{\sqrt{\tau\rho_{p}^{d}\mu_{l}^{d}} \psi_{ll}^{d}}{\tau\rho_{p}^{c}\sum\limits_{k =1}^{K}\mu_{k}^{c}\psi_{lk}^{c}\left|\boldsymbol{\omega}_{k}^{H}\boldsymbol{\theta}_{l}\right|^{2} + \tau\rho_{p}^{d}\sum\limits_{l^{\prime} =1}^{L}\mu_{l^{\prime}}^{d}\psi_{ll^{\prime}}^{d}\left|\boldsymbol{\theta}_{l^{\prime}}^{H}\boldsymbol{\theta}_{l}\right|^{2} + N_{0}}\boldsymbol{y}_{p,ll}^{d} = \lambda_{ll}^{d}\boldsymbol{y}_{p,ll}^{d},
\end{equation}
where $\boldsymbol{y}_{p,ll}^{d} = \boldsymbol{Y}_{p,l}^{d}\boldsymbol{\theta}_{l}$ and $\mathbb{E}\left\lbrace\hat{\boldsymbol{g}}_{ll}^{d}\hat{\boldsymbol{g}}_{ll}^{d^{H}}\right\rbrace = \sqrt{\tau\rho_{p}^{d}\mu_{l}^{d}} \psi_{ll}^{d}\lambda_{ll}^{d}\boldsymbol{I}_{N\times N} = \gamma_{ll}^{d}\boldsymbol{I}_{N\times N}$.
\subsubsection{Achievable rate}
In the following, based on the estimated channels and by applying a MRC receiver, the achievable data rate of the CFUEs and DUEs are presented.
 \begin{theorem} The closed form achievable data rate for $k$th CFUE with imperfect CSI is given by
 \begin{equation}\label{eqn:IPCSI_apx}
     \!\!R_{k}^{\text{CFUE\textsubscript{ip}}} \!=\!\varsigma\!\log_{2}\!\!\left(\!\!1\! +\! \dfrac{\xi^{2}\eta_{k}^{c}\rho^{c}\left(\sum\limits_{m\in\mathcal{M}_{k}}^{}\gamma_{mk}^{c}\right)^{2}}{\splitfrac{\xi\rho^{c}\!\! \sum\limits_{k^{\prime}=1} ^ {K}\!\! \eta_{k^{\prime}}^{c}\!\!\!\sum\limits_{m\in\mathcal{M}_{k}}^{}\!\!\! \gamma_{mk}^{c}\beta_{mk^{\prime}}^{c} + \xi\rho^{d}\! \sum\limits_{l^{\prime}=1} ^ {L}\!\! \eta_{l^{\prime}}^{d}\!\!\!\sum\limits_{m\in\mathcal{M}_{k}}^{}\!\!\! \gamma_{mk}^{c}\beta_{ml^{\prime}}^{d}+ \xi N_{0}\!\!\!\sum\limits_{m\in\mathcal{M}_{k}}^{}\!\!\!\gamma_{mk}^{c}\!+\! (1\!-\!\xi^{2})\rho^{c}\eta_{k}^{c}\!\!\sum\limits_{m\in\mathcal{M}_{k}}^{}\!\!\! \gamma_{mk}^{c^{2}}}{\!\!\!\!\!\!\!\!\!\!\!\!\!\!\!\!\!\!\!\!\!\!\!\!\!\!\!\!\!\!\!\!\!\!\!\!\!\!\!\!\!\!\!\!\!\!\!\!\!\!\!\!\!\!\!\!\!\!\!\!+\!\rho^{c}\!\sum\limits_{k^{\prime}\neq k} ^ {K}\!\! \eta_{k^{\prime}}^{c}\!\!\!\sum\limits_{m\in\mathcal{M}_{k}}^{}\!\!\! \left(\!\gamma_{mk}^{c}\frac{\sqrt{\mu_{k^{\prime}}^{c}}\beta_{mk^{\prime}}^{c}}{\sqrt{\mu_{k}^{c}}\beta_{mk}^{c}}\!\right)^{2}\!\!\!\left|\boldsymbol{\omega}_{k}^{H}\boldsymbol{\omega}_{k^{\prime}}\right|^{2}\!\!+\!\rho^{d}\! \sum\limits_{l^{\prime}=1} ^ {L}\!\! \eta_{l^{\prime}}^{d}\!\!\!\sum\limits_{m\in\mathcal{M}_{k}}^{}\!\! \!\left(\!\!\gamma_{mk}^{c}\frac{\sqrt{\mu_{l^{\prime}}^{d}\rho_{p}^{d}}\beta_{ml^{\prime}}^{d}}{\sqrt{\mu_{k}^{c}\rho_{p}^{c}}\beta_{mk}^{c}}\!\right)^{2}\!\!\!\left|\boldsymbol{\omega}_{k}^{H}\boldsymbol{\theta}_{l^{\prime}}\right|^{2}}}\!\!\right)\!\!.
 \end{equation}
\end{theorem}
\begin{proof}[Sketch of Proof]
By applying the UatF technique and following the approach similar to that of perfect CSI and using the estimated channels, the achievable data rate is derived.
\end{proof}
In order to derive the achievable data rate of the $l$th DUE receiver we write $\boldsymbol{g}_{ll}^{d} = \hat{\boldsymbol{g}}_{ll}^{d} +  \boldsymbol{\varepsilon}_{ll}^{d},$  where $\boldsymbol{\varepsilon}_{ll}^{d}$ is the LMMSE estimation error and it is independent from the estimated channel. Thus, the combined signal using MRC and imperfect CSI is given by
\begin{equation}\label{eqn:IPCSI_MRCd2d}
\begin{split}
  r_{l}^{d} &= \hat{\boldsymbol{g}}_{ll}^{d^{H}}\boldsymbol{y}_{l}^{d} = \underbrace{\sqrt{\rho^{d}\eta_{l}^{d}}\left\|\hat{\boldsymbol{g}}_{ll}^{d}\right\|^{2}}_{\text{DS\textsubscript{$l$}}}s_{l}^{d} \!+\! \underbrace{\sqrt{\!\rho^{c}}\sum\limits_{k =1}^{K}\!\!\sqrt{\eta_{k}^{c}}\hat{\boldsymbol{g}}_{ll}^{d^{H}}\boldsymbol{g}_{lk}^{c}s_{k}^{c}}_{\text{ICFUE\textsubscript{$l$}}} + \underbrace{\sqrt{\rho^{d}}\sum\limits_{l^{\prime} \neq l}^{L}\!\!\sqrt{\eta_{l^{\prime}}^{d}}\hat{\boldsymbol{g}}_{ll}^{d^{H}}\boldsymbol{g}_{ll^{\prime}}^{d}s_{l^{\prime}}^{d}}_{\text{IDUE\textsubscript{$l$}}}\\
  &+ \underbrace{\hat{\boldsymbol{g}}_{ll}^{d^{H}}\sqrt{\!\rho^{d}\eta_{l}^{d}}\boldsymbol{\varepsilon}_{ll}^{d}s_{l}^{d}}_{\text{TEE\textsubscript{$l$}: Total Estimation Error}} + \underbrace{\hat{\boldsymbol{g}}_{ll}^{d^{H}}\boldsymbol{n}_{l}^{d}}_{\text{TN\textsubscript{$l$}}}.
  \end{split}
\end{equation}
Note that $\mathbb{E}\!\left\lbrace \boldsymbol{\varepsilon}_{ll}^{d}\boldsymbol{\varepsilon}_{ll}^{d^{H}} \right\rbrace = (\psi_{ll}^{d} - \gamma_{ll}^{d})\boldsymbol{I}_{N\times N}$, and therefore, by treating the interfering terms in \eqref{eqn:IPCSI_MRCd2d} as an equivalent Gaussian noise, the ergodic rate for the $l$th DUE is obtained similar to \eqref{eqn:PCSI_d2d1}.% as \eqref{eqn:IPCSI_MRCd2d1}
 %\begin{equation}\label{eqn:IPCSI_MRCd2d1}
 %    R_{l}^{\text{DUE\textsubscript{ip}}} \!=\!\varsigma\mathbb{E}\!\left\lbrace\! \log_{2}\!\left(\!\!1\! +\! \dfrac{\rho^{d}\eta_{l}^{d}\left\|\hat{\boldsymbol{g}}_{ll}^{d}\right\|^{4}}{\!\left\|\hat{\boldsymbol{g}}_{ll}^{d}\right\|^{2}\!\!\rho^{c}\!\sum\limits_{k =1}^{K}\!\eta_{k}^{c}\psi_{lk}^{c} \!+ \!\left\|\hat{\boldsymbol{g}}_{ll}^{d}\right\|^{2}\!\rho^{d}\!\sum\limits_{l^{\prime} \neq l}^{L}\!\eta_{l^{\prime}}^{d}\psi_{ll^{\prime}}^{d}\!\!+\! \left\|\hat{\boldsymbol{g}}_{ll}^{d}\right\|^{2}\!\!\rho^{d}\eta_{l}^{d}\!\left(\psi_{ll}^{d}\! -\! \gamma_{ll}^{d}\right) \!+\! N_{0}\!\left\|\hat{\boldsymbol{g}}_{ll}^{d}\right\|^{2}}\!\!\right)\!\right\rbrace\!.
% \end{equation}
 
 \begin{theorem} The closed form achievable data rate for $l$th DUE with imperfect CSI is given by
 \begin{equation}\label{eqn:IPCSI_apxd2d}
     \!\!\tilde{R}_{l}^{\text{DUE\textsubscript{ip}}} =\!\varsigma \log_{2}\!\!\left(\!\!1\! +\! \dfrac{\rho^{d}\eta_{l}^{d}\gamma_{ll}^{d}(N-1)}{\rho^{c}\!\sum\limits_{k =1}^{K}\!\eta_{k}^{c}\psi_{lk}^{c} \!+ \!\rho^{d}\!\sum\limits_{l^{\prime} \neq l}^{L}\!\eta_{l^{\prime}}^{d}\psi_{ll^{\prime}}^{d}\!\!+\!\rho^{d}\eta_{l}^{d}\!\left(\psi_{ll}^{d}\! -\! \gamma_{ll}^{d}\right) \!+\! N_{0}}\!\!\right)\!\!.
 \end{equation}
\end{theorem}
\begin{proof}[Sketch of Proof]
By applying \eqref{eqn:jensen} to the ergodic rate obtained from \eqref{eqn:IPCSI_MRCd2d}  %in \eqref{eqn:IPCSI_MRCd2d1}  
and following a similar approach as that in the case of perfect CSI, achievable data rate is derived (see Appendix B).
\end{proof}
% \textcolor{purple}{This section, demonstrated that using low resolution ADCs limits the performance of the system in many ways. Also, the presence of mutual interference between DUEs and CFUEs in SE and channel estimation expressions verified. The results of this section are used to design proper pilot and power allocation schemes in the next two sections.}
%%%%%%%%%%%%%%%%%%%%%%%%%%%%%%%%%%%%%%%%%%%%%%%%%%%%%%%%%%%%%%
\section{Pilot assignment strategies}\label{sec4PilotandPower}
In this section two pilot assignment strategies are presented to control the pilot contamination effect in the training phase and improve the channel estimation quality. \textcolor{black}{To this end, here, the results of previous section are used.}

As it can be seen from the channel estimates given in \eqref{eqn:ChEst} and \eqref{eqn:ChEstd2d}, sharing pilots between CFUEs and DUEs introduces the second term in the denominator of \eqref{eqn:ChEst} and \eqref{eqn:ChEstd2d}, expressing the pilot contamination between DUEs on CFUEs, which degrades the channel estimation quality. To circumvent this drawback, we first remove this coupling by assigning different sets of orthogonal pilots to CFUEs and DUEs, and then we employ established pilot assignment techniques to allocate the pilots among the users. Although in this way the frequency of reusing a certain pilot among CFUEs or DUEs may increase, it permits us to apply modified version of proven pilot assignment techniques, such as greedy \cite{ngo2017cell} and graph coloring-based \cite{zhu2015graph} pilot assignments, to improve the system performance. It is worth to mention that this decoupling will also simplify the rate expressions, e.g. terms that include $\boldsymbol{\omega}_{k}^{H}\boldsymbol{\theta}_{l}$ will be removed. Besides, it is a rational approach from practical viewpoint to assign pilots taken from different sets for CFUEs and DUEs.
\subsection{CFUE Pilot Assignment}
Therefore, for a total of $\tau$ orthogonal pilots we consider $\tau_{d} = \max\left\lbrace\lfloor\frac{L}{L+K}\tau\rfloor,1\right\rbrace$ of them for DUEs and the remaining $\tau_{c} = \tau-\tau_{d}$ for CFUEs, so that $\boldsymbol{\omega}_{k}^{H}\boldsymbol{\theta}_{l}=0,\ \forall k, l$. Next, we resort to the greedy approach for assigning $\tau_{c}$ orthogonal pilots among $K$ CFUEs \cite{ngo2017cell}. The greedy-based CFUE pilot allocation (GCPA) is given in Algorithm 1. This algorithm starts with an initial random pilot allocation; then, the SE of the CFUEs are computed and the user with the minimum data rate is selected and allocated with a pilot that minimizes the pilot contamination term resulting from other CFUEs. Then, these steps are repeated for the newly assigned pilots for a limited number of iterations\footnote{Although after $T$ iterations the algorithm determines pilot sequences assigned for $K$ CFUEs, its convergence to an optimal solution or a solution that guarantees a performance improvement is not granted \cite{ngo2017cell}. In fact, in this algorithm when a pilot is allocated to the user with minimum SE in an iteration, This new pilot allocation may cause even larger interference to some other users and therefore reduce the minimum SE or sum SE of the system. However, after around $T=5$ iterations on average it can bring some improvements to the performance of system \cite{ngo2017cell}. This overall performance improvement is also observed from Fig.~\ref{fig:pilt} in the numerical results.}. Besides the greedy approach, recently a better performing graph coloring (GC) based approach  has been presented \cite{liu2020graph}. Such GC-based strategy can be applied in our context to the CFUEs following the same procedure reported in  \cite{liu2020graph};  we will show its performance in the numerical results and omit providing further details for the sake of brevity.
\begin{algorithm}[t]
\small
		\caption{Greedy-based CFUE pilot allocation (GCPA)} \label{Algorithm:GCPA}
		\textbf{Input}: Large-scale fading coefficients $\beta_{mk}^{c}, \forall \{m,k\}$, set of available orthogonal pilots $\boldsymbol{\Omega} = \{\boldsymbol{\omega}_{1},...,\boldsymbol{\omega}_{\tau_{c}}\}$, number of iterations $T$, iteration index $t=1$.
		\begin{itemize}
			\item[\textbf{I}.] \emph{Iteration} $t$:
			\begin{itemize}
				\item[I.1.] Use \eqref{eqn:IPCSI_apx} and find the following
				\begin{equation*}
				    \hat{k} = \text{arg}~\underset{k}{\min}\ \ R_{k}^{\text{CFUE\textsubscript{ip}}}.
				\end{equation*}
				\item[I.2.] Choose $\omega_{\hat{k}}$ from the set of CFUE pilots that minimizes following term
				\begin{equation*}
                  \boldsymbol{\omega}_{\hat{k}} = \text{arg}~\underset{\boldsymbol{\pi}_{\hat{k}}\in\boldsymbol{\Omega}}{\min}\ \ \sum\limits_{m\in\mathcal{M}_{\hat{k}}}^{}\sum\limits_{k^{\prime} \neq \hat{k}}^{K}\mu_{k^{\prime}}^{c}\beta_{mk^{\prime}}^{c}\left|\boldsymbol{\pi}_{\hat{k}}^{H}\boldsymbol{\omega}_{k^{\prime}}\right|^{2}
				\end{equation*}
%				\item[I.3.] Update $\varrho_{k}^{(n)}$ as follows
			\end{itemize}
			\item[\textbf{II}.] \textit{If} $t=T$ stop. Otherwise $t=t+1$ and go to \textbf{I}.
			%\item[\textbf{III}.] Set $\varrho_{k}^{(n-1)}\leftarrow\varrho_{k}^{(n)}$ $\forall k$ and $n = n+1$, then go to \textbf{I}.
		\end{itemize}
		\textbf{Output}: The assigned pilots $\boldsymbol{\omega}_{k} \forall k.$
\end{algorithm}
\subsection{D2D Pilot Assignment}
Here, we propose a modified version of graph coloring (GC) algorithm \cite{zhu2015graph,xu2016pilot} for allocating $\tau_{d}$ orthogonal pilots among $L$ DUEs. By using GC-based pilot allocation the potential interference due to reusing pilots between $l$th and $l^{\prime}$th DUE transmitters at their desired receiver is denoted by $\varphi_{ll^{\prime}}$ and is defined as
\begin{equation}\label{eqn:GCcriterion}
    \varphi_{ll^{\prime}} =
    \begin{cases}
    0, &\text{if}\ l = l^{\prime}\\
    \frac{\psi_{ll^{\prime}}^{d}}{\psi_{ll}^{d}} + \frac{\psi_{l^{\prime}l}^{d}}{\psi_{l^{\prime}l^{\prime}}^{d}}.\ \ &\text{if}\ l\neq l^{\prime}
    \end{cases}
\end{equation}
Accordingly, a large $\varphi_{ll^{\prime}}$ infers a strong interference at the receivers of $l$th and $l^{\prime}$th D2D pairs by the other D2D pair's transmitter. The Modified GC-based DUE pilot assignment (MGCDPA) algorithm, which is given in Algorithm 2\footnote{Note that, Algorithm 1 calculates the rate of all users in every iteration, determines the user with the lowest data rate and from among all the available pilots, selects the one that minimizes the pilot contamination term in the received training signal of that user. This procedure is performed for a predefined number of iterations. Algorithm 2, on the other hand, goes through exactly $L$ iterations, i.e. the number of users, which marks the first clear difference of the two algorithms. Plus, Algorithm 2 computes a level of interference between each pair of pilot transmitting users at their corresponding receivers based on \eqref{eqn:GCcriterion} and stores the values in a matrix, then using this matrix in each iteration the most interfering user is identified and a pilot from the available pilots is allocated in a way that distinct pilots are firstly allocated to the most interfering users if available, otherwise, they are assigned such that the usage frequency of pilots become almost the same while minimizing the pilot contamination.}, attempts to allocate the pilots such that the users with the same pilot experience a low value of $\varphi_{ll^{\prime}}$\footnote{Since for a particular user Algorithm 2 only runs once to assign a pilot and no refinements are performed after assigning a pilot to that user \cite{zhu2015graph}, there is no convergence issues for this algorithm. So, all D2D users are assigned with a pilot only when the algorithm has run for $\ell = L$ iterations.}. 
\begin{algorithm}[t]
\small
		\caption{Modified GC-based DUE pilot assignment} \label{Algorithm:GCDPA}
		\textbf{Input}: Interference strengths $\varphi_{ll^{\prime}}, \forall \{l,l^{\prime}\}$, set of available orthogonal pilots $\boldsymbol{\Theta} = \{\boldsymbol{\theta}_{1},...,\boldsymbol{\theta}_{\tau_{d}}\}$, set of all the transmitters of D2D pairs $\mathcal{L}$, set of users that are assigned with pilots $\mathcal{U}=\emptyset$ which is empty initially, iteration index $\ell=1$.
		\begin{itemize}
			\item[\textbf{I}.] \emph{Iteration} $\ell$:
			\begin{itemize}
				\item[I.1.] Among all the D2D transmitters which are not assigned with pilot find the one that experiences or causes the largest interference
				\begin{equation*}
				    \hat{l} = \text{arg}~\underset{l^{\prime}\in\mathcal{L}\setminus\mathcal{U}}{\max}\ \ \sum\limits_{l\in\mathcal{L}}\varphi_{ll^{\prime}}.
				\end{equation*}
				\item[I.2.] From the set of available pilots, i.e. $\Theta$, select the one that minimizes interference to the users with the same pilot sequence, 
				\begin{equation*}
				\begin{split}
                  \boldsymbol{\theta}_{\hat{l}} &= \text{arg}~\underset{\boldsymbol{\pi}_{\hat{l}}\in\boldsymbol{\Theta}}{\min}\ \ \sum\limits_{l\in\mathcal{U}}\varphi_{l\hat{l}}\left|\boldsymbol{\theta}_{l}^{H}\boldsymbol{\pi}_{\hat{l}}\right|^{2},\\
                  \mathcal{U} &= \mathcal{U}\cup\hat{l}.
                  \end{split}
				\end{equation*}
%				\item[I.3.] Update $\varrho_{k}^{(n)}$ as follows
			\end{itemize}
			\item[\textbf{II}.] \textit{If} $\ell=L$ stop. Otherwise $\ell=\ell+1$ and go to \textbf{I}.
			%\item[\textbf{III}.] Set $\varrho_{k}^{(n-1)}\leftarrow\varrho_{k}^{(n)}$ $\forall k$ and $n = n+1$, then go to \textbf{I}.
		\end{itemize}
		\textbf{Output}: The assigned pilots $\boldsymbol{\theta}_{l} \forall l$.
\end{algorithm}
\subsection{Complexity of Algorithms 1 and 2}
%Since the number of pilot signals are limited and they are known beforehand, we assume that the inner product of pilot sequences are computed and known beforehand so the complexity of the algorithm does not scale by their size. 
     Step I$.1$ in Algorithm 1 has a complexity of $\mathcal{O}\left(\mathcal{A}(M,L,K)\right)$ where $\mathcal{A}(M,L,K)=MK\log K + (L+K)\sum\limits_{k=1}^{K}\left|\mathcal{M}_{k}\right|$. It is because we first need to sort the large-scale fading vectors of size $K\times 1$ at each AP with complexity of $\mathcal{O}(K\log K)$ to perform user-centric strategy \cite{cormen2009introduction}. Since there are $M$ such vectors, the overall complexity becomes $\mathcal{O}(MK\log K)$. The second term in the complexity expression of step I$.1$, comes from the computational complexity of achievable data rate for all $K$ users. The complexity of computing other parts in this step are not larger than the two discussed parts. Also, step I$.2$, has a complexity of $\mathcal{O}(\tau_{c}K|\mathcal{M}_{\hat{k}}|)$. Therefore, the overall complexity of GCPA algorithm over $T$ iterations becomes
     $\mathcal{O}\left(T\left(\mathcal{A}(M,L,K)+\tau_{c}K|\mathcal{M}_{\hat{k}}|\right)\right)$.
     
     For Algorithm 2, in $\ell$th iteration, step I.1 requires a complexity of $\mathcal{O}\left(L\times (L-\ell)\right)$ which comes from $\left|\mathcal{L}\setminus\mathcal{U}\right|=L-\ell$ summations of complexity $\mathcal{O}(L)$. Step I.2 requires a complexity of $\mathcal{O}\left(L\times \ell\right)$ which is the cost of $\left|\mathcal{U}\right|=\ell$ summations of complexity $\mathcal{O}(L)$. Therefore. the overall complexity in $\ell$th iteration would be $\mathcal{O}\left(L^{2}\right)$ and since the algorithm runs for $L$ iterations, the over all complexity is $\mathcal{O}\left(L^3\right)$.

% \textcolor{purple}{This section provided ways to assign pilot sequences in order to alleviate the mutual interference between users and therefore improve the SE and channel estimation quality. The complexity of each algorithm was also analysed.}
\section{Power Allocation} \label{Sec:powerallocation}
\textcolor{black}{Mutual interference between users in the data transmission phase is another important factor that degrades the SE of the system. Therefore it is important to} consider transmit power allocation to further improve the system performance. The following two optimization problems are considered to this end:
\begin{itemize}
    \item \textbf{M}ax \textbf{S}um \textbf{R}ate of \textbf{C}FUEs subject to \textbf{Q}uality of services for \textbf{D}UEs (MSRCQD),
    \item \textbf{W}eighted \textbf{M}ax \textbf{P}roduct of SINRs of \textbf{C}FUEs and \textbf{D}UEs (WMPCD).
\end{itemize}
In MSRCQD the sum data rate of the CFUEs are maximized while DUEs are constrained to have larger data rates than a predetermined value. In WMPCD, the objective is to maximize the weighted product of SINRs of CFUEs and DUEs. This utility function improves the overall performance of the system while also ensuring a degree of fairness between the users, so that all the users are served with a non-zero data rate \cite{bjornson2017massive}. MSRCQD may lead to solutions that no power is allocated to some cell-free users \cite{bjornson2017massive}. On the other hand, in WMPCD not only all users are guaranteed to be allocated with a non-zero power but also since the objective is the product of the SINR of all users, a fair performance for all users is provided. Moreover, in the case of WMPCD we have introduced weights to prioritize one type of users over the others.

\subsection{MSRCQD Optimization Problem}
The MSRCQD optimization problem is formulated as follows
 \begin{equation}\label{power1}
  \mathcal{P}_{1}:
    \begin{cases}
      \begin{aligned}
        &\underset{\{\eta_{k}^{c}\geq0,\ \eta_{l}^{d}\geq0\}_{k,l}}{\text{{\textit{\emph{maximize}}}}}
                             && \sum\limits_{k=1}^{K}R_{k}^{\text{CFUE\textsubscript{ip}}}\\
        &\text{{\textit{\emph{subject to}}}} && \tilde{R}_{l}^{\text{DUE\textsubscript{ip}}}\!\geq\! R_{l,\min},  \ l = 1,2,...,L,\\
        &                    && \eta_{k}^{c}\leq 1, \ \ \  \eta_{l}^{d}\leq 1,\ \ \ \ k = 1,...,K, \ l = 1,...,L.
      \end{aligned}
    \end{cases}
 \end{equation}
 $\mathcal{P}_{1}$ is a non-convex and NP-hard problem to solve optimally. So, for solving this optimization  problem we present two different solutions, namely successive convex optimization (SCA) and geometric programming where the former leads to an iterative solution while the latter relies on the high SINR approximation that can be solved using available GP solvers. \textcolor{black}{While SCA can achieve KKT optimality conditions of $\mathcal{P}_1$ using an \textit{iterative} algorithm, one can find a \textit{suboptimal} solution using a \textit{non-iterative} and \textit{lower-complexity} GP method. However, as it is observed in the numerical results at higher SINRs, the performance gap of GP and SCA is negligible.}
 
 \medskip
 
 {\bf Successive convex optimization.}
       Here, successive lower-bound maximization procedure is used to solve $\mathcal{P}_{1}$. To see this, assume that $A_{k}\left(\eta_{k}^{c}\right)$ and $C_{l}\left(\eta_{l}^{d}\right)$ are the numerator of the SINR of the $k$th CFUE and the $l$th DUE, respectively. Also, $B_{k}\left(\eta^{c},\eta^{d}\right)$ and $D_{l}\left(\eta^{c},\eta^{d}\right)$ are the corresponding denominators, and $\eta^{c} = \{\eta_{1}^{c},...,\eta_{K}^{c}\},\ \eta^{d} = \{\eta_{1}^{d},...,\eta_{L}^{d}\}$. The quantities $A_k(\cdot)$, $B_k(\cdot, \cdot)$, $C_l(\cdot)$, $D_l(\cdot, \cdot)$ are all linear functions of the variable to be optimized $\bbeta$. Consider the generic minimum rate constraint
 $\tilde{R}_{l}^{\text{DUE\textsubscript{ip}}}\!\geq\! R_{l,\min}$. With basic algebra, this constraint can be reformulated as
\begin{equation}
C_{l}\left(\eta_{l}^{d}\right) \geq \left(2^{R_{l,\min}/\varsigma}-1\right) D_{l}\left(\eta^{c},\eta^{d}\right) \; ,
\end{equation}
which shows that the constraints in $\mathcal{P}_1$ are linear. The non-convexity of $\mathcal{P}_{1}$ is thus due to the objective function only. This function, neglecting the irrelevant constant $\varsigma$, can be written as:
\begin{equation}
\sum\limits_{k=1}^{K} \log_{2} \left(1 + \frac{A_{k}\left(\eta_{k}^{c}\right)}{B_{k}\left(\eta^{c},\eta^{d}\right)}\right)=
\underbrace{\sum\limits_{k=1}^{K} \log_{2} \left(A_{k}\left(\eta_{k}^{c}\right) + B_{k}\left(\eta^{c},\eta^{d}\right)\right)}_{g_1(\bbeta)} -
\underbrace{\sum\limits_{k=1}^{K} \log_{2} \left(B_{k}\left(\eta^{c},\eta^{d}\right)\right)}_{g_2(\bbeta)}.
\end{equation}
The functions $g_1(\bbeta)$ and $g_2(\bbeta)$ are both concave. Recall now that any concave function is upper-bounded by its first Taylor expansion around any given point $\bbeta_0$, i.e. we have
\begin{equation}
g_2(\bbeta) \leq g_2(\bbeta_0) + \nabla^T_{\bbeta}g_2(\bbeta)|_{\bbeta=\bbeta_0}(\bbeta -\bbeta_0) \; ,
\end{equation}
with $\nabla_{\bbeta}g_2(\bbeta)|_{\bbeta=\bbeta_0}$ the gradient of the function $g_2(\cdot)$ with respect to $\bbeta$ and computed for $\bbeta=\bbeta_0$.
Accordingly, the objective function of $\mathcal{P}_1$ can be lower-bounded as
\begin{equation}
\sum\limits_{k=1}^{K}R_{k}^{\text{CFUE\textsubscript{ip}}} \geq g_1(\bbeta) - g_2(\bbeta_0) - \nabla^T_{\bbeta}g_2(\bbeta)|_{\bbeta=\bbeta_0}(\bbeta -\bbeta_0)
\triangleq G(\bbeta, \bbeta_0) \; .
\end{equation}
It is easy to realize that the lower-bounding function $G(\bbeta, \bbeta_0)$ is a concave function, and that for $\bbeta=\bbeta_0$ the bound holds with equality. Otherwise stated, properties  \textbf{P1} -- \textbf{P3} \cite{buzzi2019user} hold and the successive lower-bound maximization strategy can be applied.
Summing up, the proposed procedure works as follows.
\begin{itemize}
	\item[1.] Set $i=0$ and choose any feasible point $\bbeta_0$.
	\item[2.] Solve the following convex optimization problem with any standard numerical procedure (e.g., fmincon routine)
	\begin{equation}\label{poweri}
	\mathcal{O}_{i}:
	\begin{cases}
	\begin{aligned}
	&\underset{\{\eta_{k}^{c}\geq0,\ \eta_{l}^{d}\geq0\}_{k,l}}{\text{{\textit{\emph{maximize}}}}}
	&&   G(\bbeta,\bbeta_i) \\
	&\text{{\textit{\emph{subject to}}}} && C_{l}\left(\eta_{l}^{d}\right) \geq \left(2^{R_{l,\min}/\varsigma}-1\right) D_{l}\left(\eta^{c},\eta^{d}\right),  \ l = 1,2,...,L,\\
	&                    && \eta_{k}^{c}\leq 1, \ \ \  \eta_{l}^{d}\leq 1,\ \ \ \ k = 1,...,K, \ l = 1,...,L.
	\end{aligned}
	\end{cases}
	\end{equation}
	Let $\bbeta'$ denote the solution to problem $\mathcal{O}_i$.
	\item[3.] Set $i=i+1$ and $\bbeta_i=\bbeta'$.
	\item[4.] Repeat steps 2-3 until convergence or maximum allowed number of iterations, i.e. $N_{\text{max}}$, has been reached.
\end{itemize}
Based on the theory discussed in \cite[Section V, Subsection A]{buzzi2019user}, the following can be stated:
\begin{theorem}
	After each repetition of steps 2-3, the sum-rate value, i.e. the objective of problem $\mathcal{P}_1$ is not decreased, and the resulting sequence of values converges. At the convergence, the found power allocation fulfills the KKT first-order optimality conditions of problem $\mathcal{P}_1$.
\end{theorem}

\medskip

 {\bf Geometric programming.}
     This solution is described in the following theorem.
 \begin{theorem}
  Solution of problem $\mathcal{P}_{1}$ can approximately be efficiently obtained using the following GP problem.
  \begin{equation}\label{power1GP}
  \mathcal{P}^{\prime}_{1}:
    \begin{cases}
      \begin{aligned}
        &\underset{{\substack{\{\eta_{k}^{c}\geq0,\ v_{k}\geq 0 \}_{k}\\ \{\eta_{l}^{d}\geq0\}_{l}}}}{\text{{\textit{\emph{maximize}}}}}
                             && \prod\limits_{k=1}^{K}v_{k}\\
        &\text{{\textit{\emph{subject to}}}} && \frac{v_{k}B_{k}\left(\eta^{c},\eta^{d}\right)}{A_{k}\left(\eta_{k}^{c}\right)}\leq 1,\\
        &                    && \left( 2^{R_{l,\min}/\varsigma}-1\right)\frac{D_{l}\left(\eta^{c},\eta^{d}\right)}{C_{l}\left(\eta_{k}^{c}\right)}\leq 1,\\
        &                    && \eta_{k}^{c}\leq 1, \ \ \  \eta_{l}^{d}\leq 1,\ \ \ \ k = 1,...,K, \ l = 1,...,L.
      \end{aligned}
    \end{cases}
 \end{equation}
%  Where $A_{k}\left(\eta_{k}^{c}\right)$ and $C_{l}\left(\eta_{l}^{d}\right)$ are the numerator of the SINR of the $k$th CFUE and the $l$th DUE reported in \eqref{eqn:IPCSI_apx} and \eqref{eqn:IPCSI_apxd2d}, respectively. Also, $B_{k}\left(\eta^{c},\eta^{d}\right)$ and $D_{l}\left(\eta^{c},\eta^{d}\right)$ are the corresponding denominators, and $\eta^{c} = \{\eta_{1}^{c},...,\eta_{K}^{c}\},\ \eta^{d} = \{\eta_{1}^{d},...,\eta_{L}^{d}\}$.
 \end{theorem}
 \begin{proof}
 By assuming high SINR approximation for CFUEs, the objective function in $\mathcal{P}_{1}$ after ignoring ``1" in rate expression inside the logarithm in \eqref{eqn:IPCSI_apx} becomes $\sum\limits_{k=1}^{K}R_{k}^{\text{CFUE\textsubscript{ip}}}\approx \varsigma\sum\limits_{k=1}^{K}\log\left(\frac{A_{k}\left(\eta_{k}^{c}\right)}{B_{k}\left(\eta^{c},\eta^{d}\right)}\right) = \varsigma\log\left(\prod\limits_{k=1}^{K}\frac{A_{k}\left(\eta_{k}^{c}\right)}{B_{k}\left(\eta^{c},\eta^{d}\right)}\right) $. Then by removing the constant coefficient $\varsigma$ and ignoring the monotonically increasing function, i.e. the logarithm, the optimizing values of optimization variables will remain unchanged. 
 
 Next, we introduce the auxiliary variable $v_{k}$ such that $\frac{A_{k}\left(\eta_{k}^{c}\right)}{B_{k}\left(\eta^{c},\eta^{d}\right)}\geq v_{k}$ which results the first constraint and the objective of  $\mathcal{P}^{\prime}_{1}$. Since $A_{k}\left(\eta_{k}^{c}\right)$, $C_{l}\left(\eta_{l}^{d}\right)$, and the objective in $\mathcal{P}_{1}^{\prime}$ are monomial and $B_{k}\left(\eta^{c},\eta^{d}\right)$, $D_{l}\left(\eta^{c},\eta^{d}\right)$ are posynomial, the inequality constraints are posynomial, and thus problem \eqref{power1GP} is a GP problem\footnote{We have used MOSEK in CVX to solve GP problems. It exploits the fact that any GP problem can be reformulated in the form of a convex optimization problem. So, it applies interior-point technique that uses Newton's method to find the solution. Hence, in each iteration, by applying Newton's method to the unconstrained convex problem, it moves along the descending direction of the objective function and finally converges to a solution that its gap with the optimal solution can be made arbitrarily small in cost of increased number of Newton iterations \cite[Chapter 11]{boyd2004convex}.}.
 \end{proof}

\subsection{WMPCD Optimization Problem}
The second problem that we investigate is WMPCD which is formulated as $\mathcal{P}_{2}$
 \begin{equation}\label{power2}
  \mathcal{P}_{2}:
    \begin{cases}
      \begin{aligned}
        &\underset{\{\eta_{k}^{c}\geq0,\ \eta_{l}^{d}\geq0\}_{k,l}}{\text{{\textit{\emph{maximize}}}}}
                             && \left(\prod\limits_{k=1}^{K}\frac{A_{k}\left(\eta_{k}^{c}\right)}{B_{k}\left(\eta^{c},\eta^{d}\right)}\right)^{w^{c}}\left(\prod\limits_{l=1}^{L}\frac{C_{l}\left(\eta_{l}^{d}\right)}{D_{l}\left(\eta^{c},\eta^{d}\right)}\right)^{w^{d}}\\
        &\text{{\textit{\emph{subject to}}}} && \eta_{k}^{c}\leq 1, \ k = 1,2,...,K,\\
        &                    && \eta_{l}^{d}\leq 1, \ l = 1,2,...,L.
      \end{aligned}
    \end{cases}
 \end{equation}
 It is worth mentioning that the first and the second terms in the objective function of $\mathcal{P}_{2}$ are the product of the SINRs of CFUEs and DUEs respectively, and $w^{c}\geq 0$, $w^{d}\geq 0$ are the respective weights. Solution of the above optimization problem is addressed in the following theorem.
 \begin{theorem}
  Solution of the optimization problem $\mathcal{P}_{2}$ can be obtained from following GP problem.
  \begin{equation}\label{power2GP}
  \mathcal{P}^{\prime}_{2}:
    \begin{cases}
      \begin{aligned}
        &\underset{{\substack{\{\eta_{k}^{c}\geq0,\ v_{k}\geq 0 \}_{k}\\ \{\eta_{l}^{d}\geq0,\ x_{l}\geq0\}_{l}}} }{\text{{\textit{\emph{maximize}}}}}
                             && \left(\prod\limits_{k=1}^{K}v_{k}\right)^{w^{c}}\left(\prod\limits_{l=1}^{L}x_{l}\right)^{w^{d}}\\
        &\text{{\textit{\emph{subject to}}}} && \frac{v_{k}B_{k}\left(\eta^{c},\eta^{d}\right)}{A_{k}\left(\eta_{k}^{c}\right)}\leq 1,\\
        &                    && \frac{x_{l}D_{l}\left(\eta^{c},\eta^{d}\right)}{C_{l}\left(\eta_{k}^{c}\right)}\leq 1,\\
        &                    && \eta_{k}^{c}\leq 1, \ \ \  \eta_{l}^{d}\leq 1,\ \ \ \ k = 1,...,K, \ l = 1,...,L.
      \end{aligned}
    \end{cases}
 \end{equation}
 \end{theorem}
 \begin{proof}
 For obtaining $\mathcal{P}^{\prime}_{2}$ we first introduce the auxiliary variables $v_{k}$ and $x_{l}$ such that $\frac{A_{k}\left(\eta_{k}^{c}\right)}{B_{k}\left(\eta^{c},\eta^{d}\right)}\geq v_{k}$ and $\frac{C_{l}\left(\eta_{l}^{c}\right)}{D_{l}\left(\eta^{c},\eta^{d}\right)}\geq x_{l}$, $\forall l,\ k$. Then, after rearranging these inequalities, the two first constraints are derived. Similar to the proof of theorem 6, the inequality constraints are posynomial while the objective function is monomial therefore $\mathcal{P}_{2}^{\prime}$ is a GP problem.
 \end{proof}
 \subsection{Complexity of SCA- and GP-based Power Control}
Major complexity of SCA-based power optimization comes from solving the approximated convex problem \eqref{poweri}. Fmincon routine that employs interior-point (IP) method to solve the convex problem is used to solve \eqref{poweri}. A common way in IP is to reformulate an inequality constrained convex optimization problem in the form of an unconstrained convex optimization problem and solve using Newton's method \cite[Chapter 11]{boyd2004convex}. We denote the number of iterations for Fmincon to reach an $\epsilon$-gap solution of \eqref{poweri} by $N$ \cite[Chapter 11, Section 11.5]{boyd2004convex}. The value of $N$ depends on the number of inequality constraints in \eqref{poweri} which is $2L+K$. Also, each Newton step requires a computation of $\mathcal{O}\left(\mathcal{B}(L,M,K)\right)$ where $\mathcal{B}(M,L,K)=MK\log K + (L+K)\left(L^2+K\sum\limits_{k=1}^{K}\left|\mathcal{M}_{k}\right|\right)$. Therefore the total complexity will be $\mathcal{O}\left(N_{\text{max}}N\mathcal{B}(M,L,K)\right)$ where $N_{\text{max}}$ is the maximum number of iterations for the SCA-based algorithm.
 
 For solving GP problems, we have used MOSEK in the CVX which exploits the fact that any GP problem can be reformulated as a convex problem \cite{boyd2004convex} and employs IP method for solving the resulted convex form. Similar to the SCA-based method,  total number of the Newton iterations to reach an $\epsilon$-gap solution of the optimal solution denoted by $N$ \cite[Chapter 11, Section 11.5]{boyd2004convex}. The number of inequality constraints is $2(L+K)$. In addition, the computational complexity of the unconstrained problems of MSRCQD and WMPCD per iteration of IP is $\mathcal{O}\left(\mathcal{B}(M,L,K)\right)$. Thus, the total complexity for GP problems is $\mathcal{O}\left(N\mathcal{B}(M,L,K)\right)$.
 
% \textcolor{purple}{In this section two power allocation schemes are developed to reduce mutual interference so as to enhance the system performance for different objectives. Also, complexity analysis for each algorithm was provided.} 
In the following section, numerical results are presented to assess the system performance of the proposed pilot and power allocation problems.
%%%%%%%%%%%%%%%%%%%%%%%%%%%%%%%%%%%%%%%%%%%%%%%%%%%%%%%%%%%%%%%%%%%%%%%https://www.overleaf.com/project/5e1c55e61b65750001d906b5
\section{Numerical Results and Discussions}\label{sec5numerical}
We focus on a simulation scenario with $K=20$ CFUEs, $L=10$ pairs of D2D users, $M=200$ APs, and $b=4$ bits, unless specifically mentioned, which are uniformly and randomly distributed within area of $D=1 \times 1$ [km$^2$]. Moreover, for each pair of D2D users we assume that the transmitter and the receiver are randomly placed within a distance of 10 [m] up to 100 [m] from one another. The areas is wrapped around to avoid boundary effects. To model the large-scale fading coefficients and noise power we the same model and parameters as in \cite{ngo2017cell,masoumi2019performance}. In particular $B=20$ [MHz] and $f=1.9$ [GHz] are system bandwidth and carrier frequency, respectively.
%The large-scale fading follows a three-slope model as below \cite{masoumi2019performance}.
%\begin{equation}\label{LSFC}
%\begin{split}
%  \!\beta_{mk} \ \ \!\! =& PL_{mk} + \sigma_{\text{sh}}z_{mk}\\
% \!P\!L_{mk}\!\!=&
%  \begin{cases}
%    -\widetilde{L}-10\log_{10}\left(d_{mk}^{3.5}\right),& \text{if } d_{mk}\geq d_{1}\\
%    -\widetilde{L}-10\log_{10}\left(d_{1}^{1.5}d_{mk}^{2}\right),& \text{if } d_{0}< d_{mk} \leq d_{1}\\
%    -\widetilde{L}-10\log_{10}\left(d_{1}^{1.5}d_{0}^{2}\right),& \text{if } d_{mk} \leq d_{0}
%  \end{cases}\\
%  \!\widetilde{L}\! =& 46.3 + 33.9\log_{10}(f)-13.82\log_{10}(h_{\text{AP}})-(1.1 \log_{10}(f)\!-\!0.7)h_u +
%      (1.56 \log_{10}(f)\!-\!0.8),
%  \end{split}
%\end{equation}
%where, $\sigma_{\text{sh}} = 8 \ [\text{dB}]$, $z_{mk}\sim\mathcal{CN}(0,1)$ are shadowing parameters, and $P\!L_{mk}$ is path-loss in [dB]. Also,  $h_{\text{AP}} = 15\ [\text{m}], \ h_u = 1.65\ [\text{m}], \ f = 1900\ [\text{MHz}]$ are the APs height, user antenna height and carrier frequency, respectively, additionally, $d_{0} = 10\ [\text{m}], \ d_{1} = 50\ [\text{m}]$. The noise power is computed by $N_{0} = B\times k_{B}\times T_{0}\times N\!F$ where, $B = 20\ [\text{MHz}]$, $k_{B} = 1.381\times 10^{-23}\ [\text{Joule/Kelvin}]$, $T_{0} = 290\ [\text{Kelvin}]$ and $N\!F = 9\ [\text{dB}]$ are system bandwidth, Boltzmann constant, temperature and noise figure, respectively. 
Also, $T = 200$ samples, $\rho^{c} = \rho^{d} = 100\ [\text{mW}]$, $\rho_{p}^{c} = \rho_{p}^{d} = 200\ [\text{mW}]$ and $\tau = 10$.
\begin{figure}[!t]
\centering
\psfragfig[scale=.6]{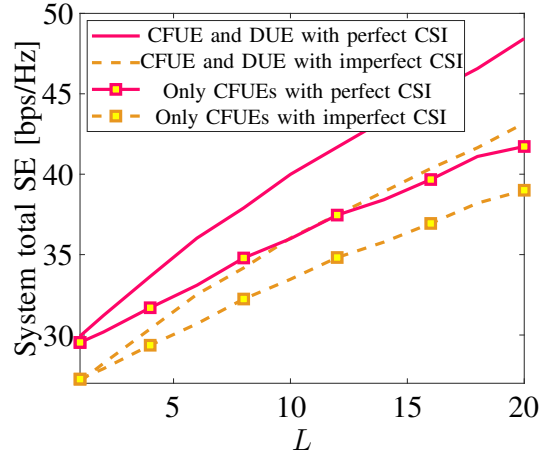}{\psfrag{L}[][][0.96]{$L$}
\psfrag{5}[][][0.96]{5}
\psfrag{10}[][][0.96]{10}
\psfrag{15}[][][0.96]{15}
\psfrag{20}[][][0.96]{20}
\psfrag{30}[][][0.96]{30}
\psfrag{35}[][][0.96]{35}
\psfrag{40}[][][0.96]{40}
\psfrag{45}[][][0.96]{45}
\psfrag{50}[][][0.96]{50}
\psfrag{AverageSSE}[][][0.96]{System total SE [bps/Hz]}
\psfrag{CFmMIMOwithD2DPCSIAAA}[][][0.68]{\ \ CFUE and DUE with perfect CSI}
\psfrag{CFmMIMOwithD2DIPCSIAAA}[][][0.68]{\ \  CFUE and DUE with imperfect CSI}
\psfrag{EquivalentCFmMIMOPCSIAAA}[][][0.68]{Only CFUEs with perfect CSI}
\psfrag{EquivalentCFmMIMOIPCSIAAA}[][][0.68]{Only CFUEs with imperfect CSI}}
\vspace{-0.2cm}
\caption{{The advantages of considering D2D users in a CF-mMIMO system.}}\label{fig:N1}
\end{figure}

\subsection{Impact of Underlaid D2D Links and of Finite Resolution ADCs}
In order to emphasize the beneficial impact that the activation of D2D links has on the system overall throughput, 
we compare the performance of two systems. The former serves simultaneously  $L$ D2D pairs and $K$ CFUEs, while the latter  \emph{only} supports $K+L$ CFUEs without serving D2D users; both cases have an equal number of served users. In Fig. \ref{fig:N1}, the system total SE, i.e. sum of DUEs and CFUEs SEs, versus $L$ is depicted for perfect and imperfect CSI with full power (FP) allocation, i.e. $\eta_{k}^{c}=\eta_{l}^{d}=1$. The results clearly show the positive impact that the activation of D2D links has with respect to the case in which all the communications flow through the network infrastructure. 
% Here, a point worth mentioning is that if a user is close enough to its destination such that a D2D communication is a viable option, it is better to establish this D2D link instead of using mobile infrastructure to improve the system data rate.
% ========= Different ADC resolution ===========

\begin{figure}[!t]
\centering
\subfloat[]{\psfragfig[scale=.59]{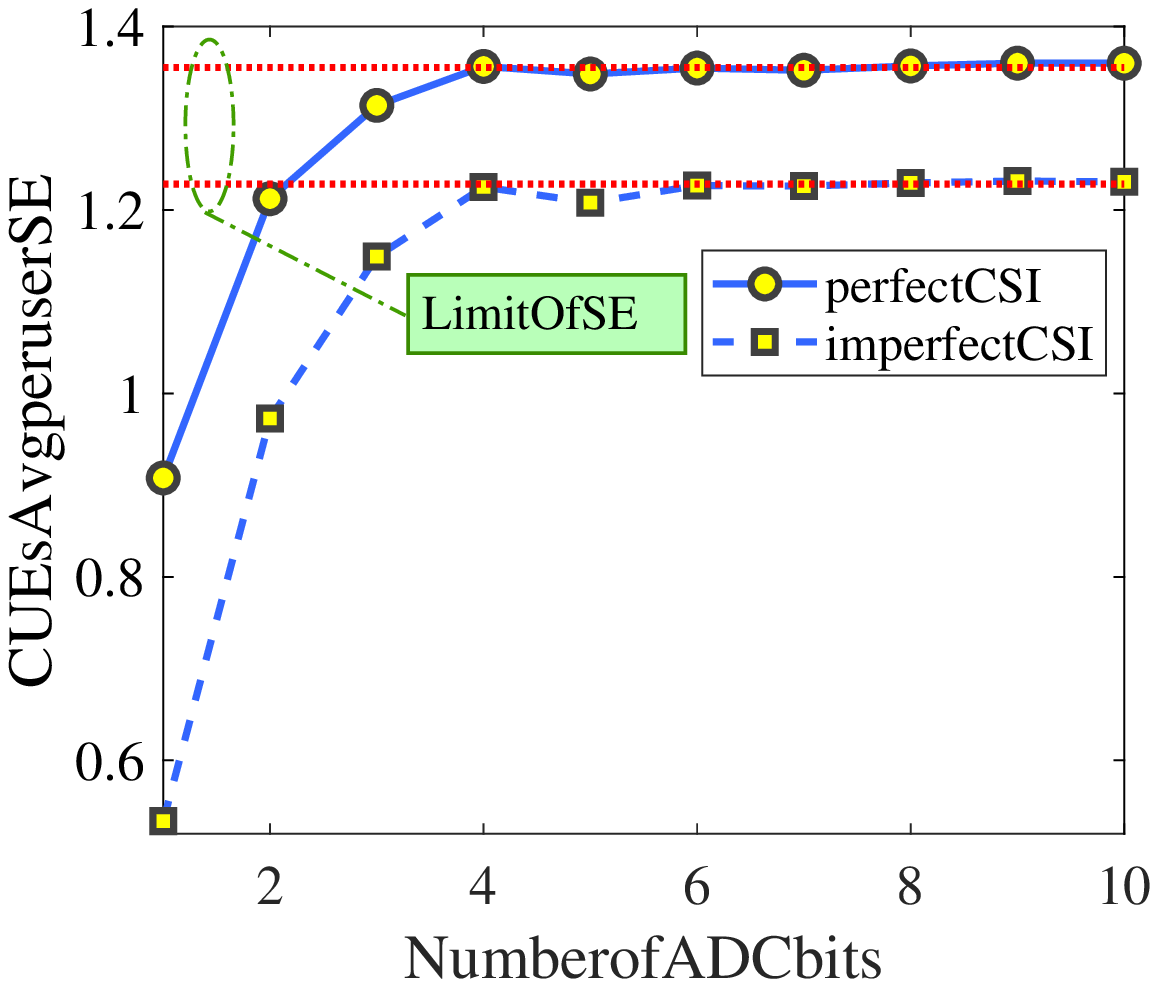}{\psfrag{NumberofADCbits}[][][0.96]{$b$ [bits]}
\psfrag{8}[][][0.96]{8}
\psfrag{1}[][][0.96]{1}
\psfrag{2}[][][0.96]{2}
\psfrag{4}[][][0.96]{4}
\psfrag{6}[][][0.96]{6}
\psfrag{10}[][][0.96]{10}
\psfrag{1.4}[][][0.96]{1.4}
\psfrag{1.2}[][][0.96]{1.2}
\psfrag{0.6}[][][0.96]{0.6}
\psfrag{0.8}[][][0.96]{0.8}
\psfrag{CUEsAvgperuserSE}[][][0.96]{CFUE per user SE [bps/Hz]}
\psfrag{LimitOfSE}[][][0.70]{\ Limit of SE}
\psfrag{perfectCSI}[][][0.70]{\ Perfect CSI}
\psfrag{imperfectCSI}[][][0.70]{\ Imperfect CSI}}\label{fig:N3a}}
\hfil
\subfloat[]{\psfragfig[scale=0.6]{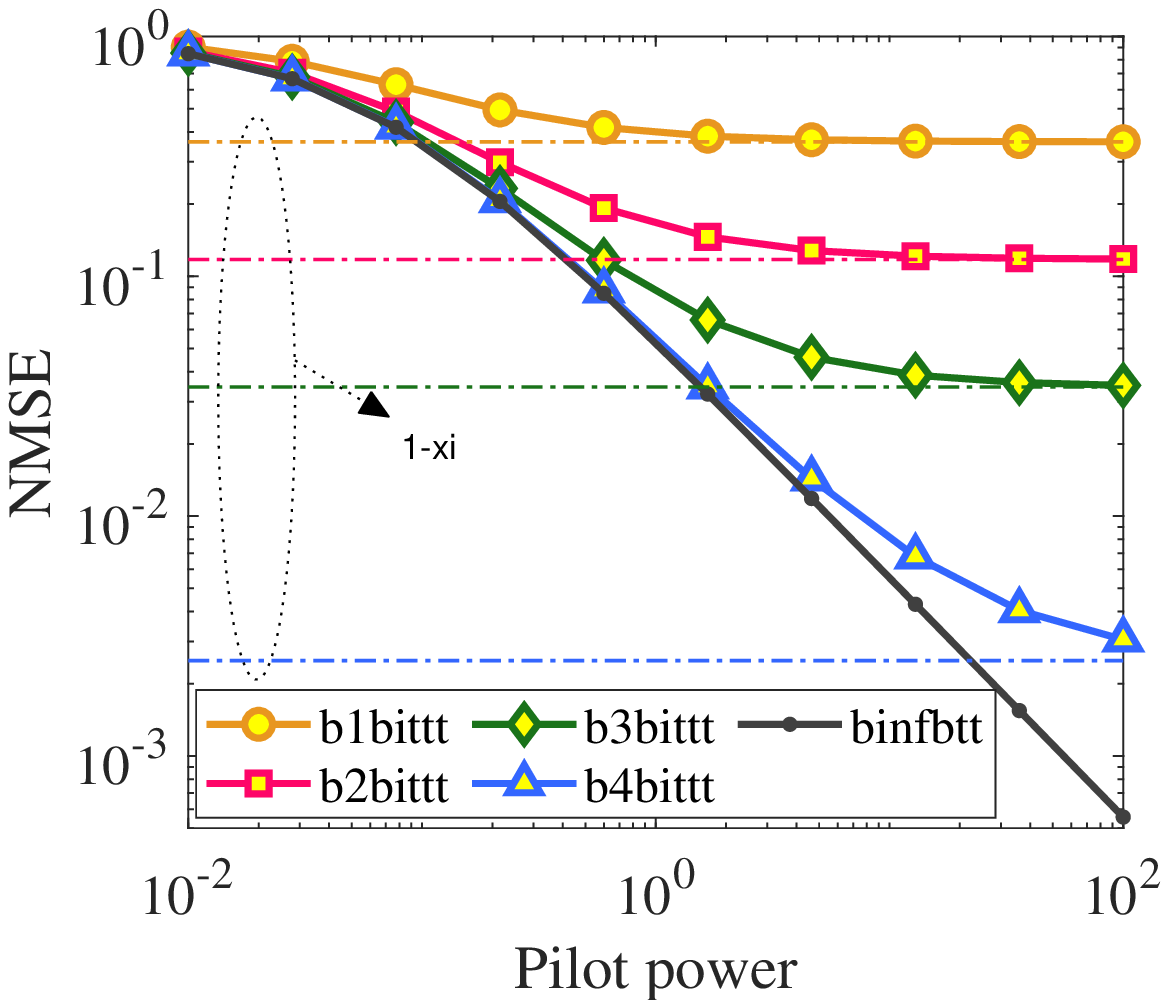}{\psfrag{0}[][][0.8]{0}
   \psfrag{-1}[][][0.8]{-1}
   \psfrag{-2}[][][0.8]{-2}
   \psfrag{-3}[][][0.8]{-3}
   \psfrag{2}[][][0.8]{2}
   \psfrag{10}[][][0.9]{10}
   \psfrag{NMSE}[][][0.96]{NMSE}
   \psfrag{Pilot power}[][][1.05]{$\rho_{p}^{c}$}
   \psfrag{b1bittt}[][][0.8]{$b\!=\!1$}
   \psfrag{b2bittt}[][][0.8]{$\ b\!=\!2$}
   \psfrag{b3bittt}[][][0.8]{$\ b\!=\!3$}
   \psfrag{b4bittt}[][][0.8]{$\ b\!=\!4$}
   \psfrag{binfbtt}[][][0.8]{$b\!=\!\infty$}
   \psfrag{1-xi}[][][0.8]{\ \ \ \ \ \ NMSE limit: $1-\xi$}}\label{fig:N3b}}
   \hfil
  \subfloat[]{\psfragfig[scale=0.6]{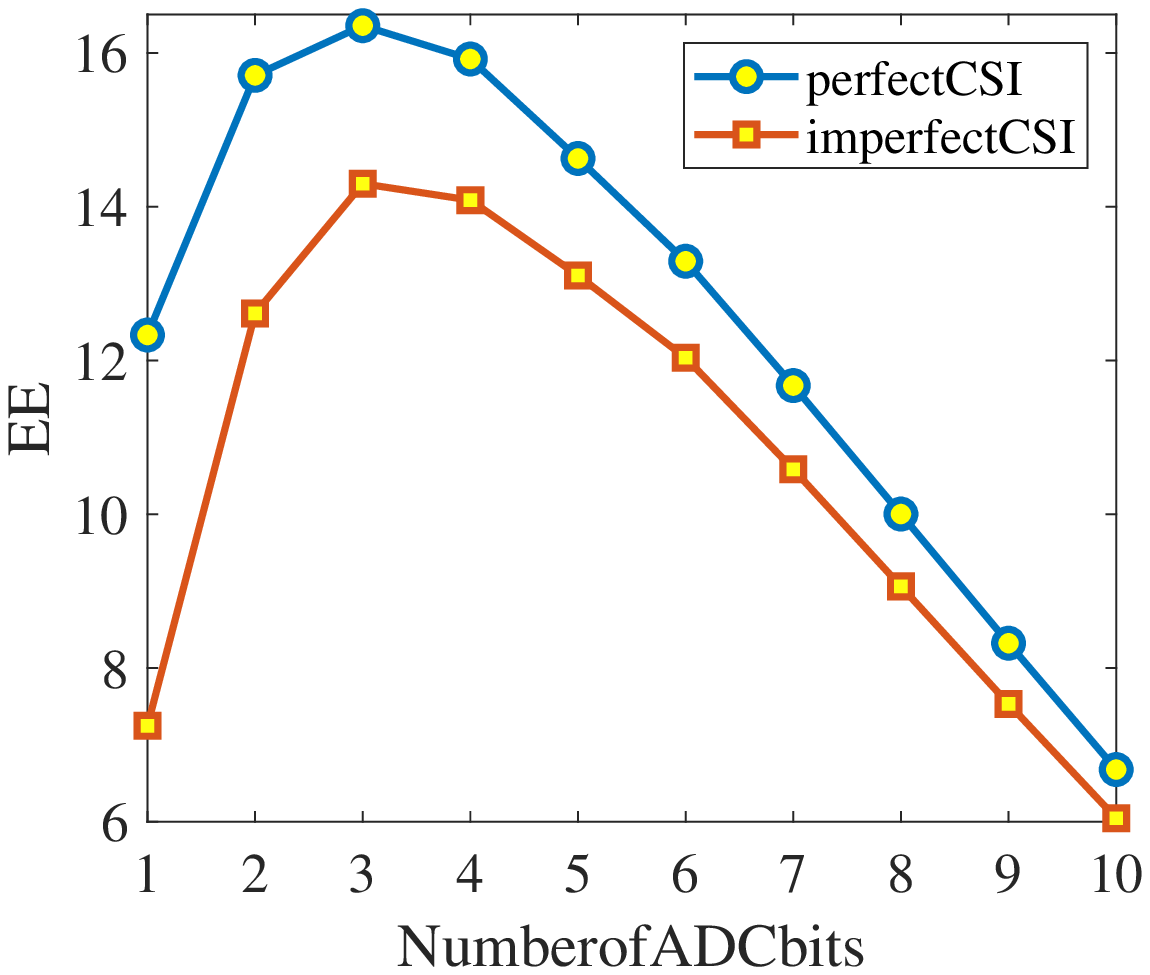}{\psfrag{1}[][][0.96]{1}
   \psfrag{2}[][][0.96]{2}
   \psfrag{3}[][][0.96]{3}
   \psfrag{4}[][][0.96]{4}
   \psfrag{5}[][][0.96]{5}
   \psfrag{6}[][][0.96]{6}
   \psfrag{7}[][][0.96]{7}
   \psfrag{8}[][][0.96]{8}
   \psfrag{9}[][][0.96]{9}
   \psfrag{10}[][][0.96]{10}
   \psfrag{12}[][][0.96]{12}
   \psfrag{14}[][][0.96]{14}
   \psfrag{16}[][][0.96]{16}
   \psfrag{EE}[][][1]{EE [Mbit/J]}
   \psfrag{NumberofADCbits}[][][0.96]{$b$ [bits]}
   \psfrag{perfectCSI}[][][0.7]{Perfect CSI}
   \psfrag{imperfectCSI}[][][0.7]{Imperfect CSI}}\label{fig:N3c}}
\vspace{-0.2cm}
\caption{The impact of low resolution ADCs on the performance of the system with full power allocation. (a) Per user SE of CFUEs versus $b$. (b) NMSE of estimated channel versus transmitted pilot power with orthogonal pilots of length $\tau=30$. (c) EE of CFUEs versus $b$.}\label{fig:N3}
\end{figure}

Next, we study the impact of the number of quantization bits at the ADCs.
From Fig.~\ref{fig:N3a} one can observe that by increasing the number of ADCs' bits the quantization noise is decreased and the performance improves. Figure~\ref{fig:N3b} depicts normalized MSE (NMSE) defined as $\text{NMSE}_{mk}=\dfrac{\text{MSE}_{mk}}{\beta_{mk}^{c}}$. This graph confirms the results in Remark 3 that using low resolution ADC limits the MSE of the estimated channels and using high resolution ADCs reduces this limit. Figure \ref{fig:N3c} studies the effect of low resolution ADCs from the energy efficiency (EE)\footnote{We define $\text{EE}=B\times \text{SSE}/P_T$ where SSE is the sum SE of CFUEs and $P_T = K(\rho_{p}^{c}+\rho^{c}) + M(P_{\text{mix}}+P_{\text{LNA}}+P_{\text{IFA}}+P_{\text{filter}}+P_{\text{AGC}}+P_{\text{syn}} + P_{ADC})$ is the total power consumption \cite{zhang2018mixed}, \cite{cui2005energy}
where $P_{\text{mix}}$, $P_{\text{LNA}}$, $P_{\text{IFA}}$, $P_{\text{filter}}$, $P_{\text{AGC}}$, $P_{\text{syn}}$ and  $P_{\text{ADC}}$ are power consumption due to the mixer, low noise amplifier, intermediate frequency amplifier (IFA), active filters, automatic gain control, frequency synthesizer and the ADC at the access points, respectively. Also $P_{\text{ADC}}=\frac{3 V_{dd}^{2} L_{\min }\left(2 B+f_{c o r}\right)}{10^{-0.1525 b+4.838}}$
where $V_{dd}$ denotes the power
supply of converter, $L_{\min}$ is the minimum
channel length for the given CMOS technology and $f_{cor}$ is the
corner frequency of the $1/f$ noise. Typical values of the mentioned parameters are given in \cite{zhang2018mixed} and \cite{cui2005energy}.} angle. For $b=3$ [bits] highest EE is achieved while for smaller and higher quantization bits lower EE is achieved due to SE degradation and high power consumption, respectively. Results show that using 4-bit ADCs is sufficient. Therefore, we use choose $b=4$ in the remainder of this section.
\subsection{User-centric AP-UE Association}
In previous results all CFUEs were served by all the  APs. Now, we turn our attention to the UC CF-mMIMO case where only a limited number of CFUEs are served by the APs.
% =========== UC CF mMIMO vs. CF mMIMO ============
\begin{figure}[!t]
  \centering
  \subfloat[]{\psfragfig[scale=.6]{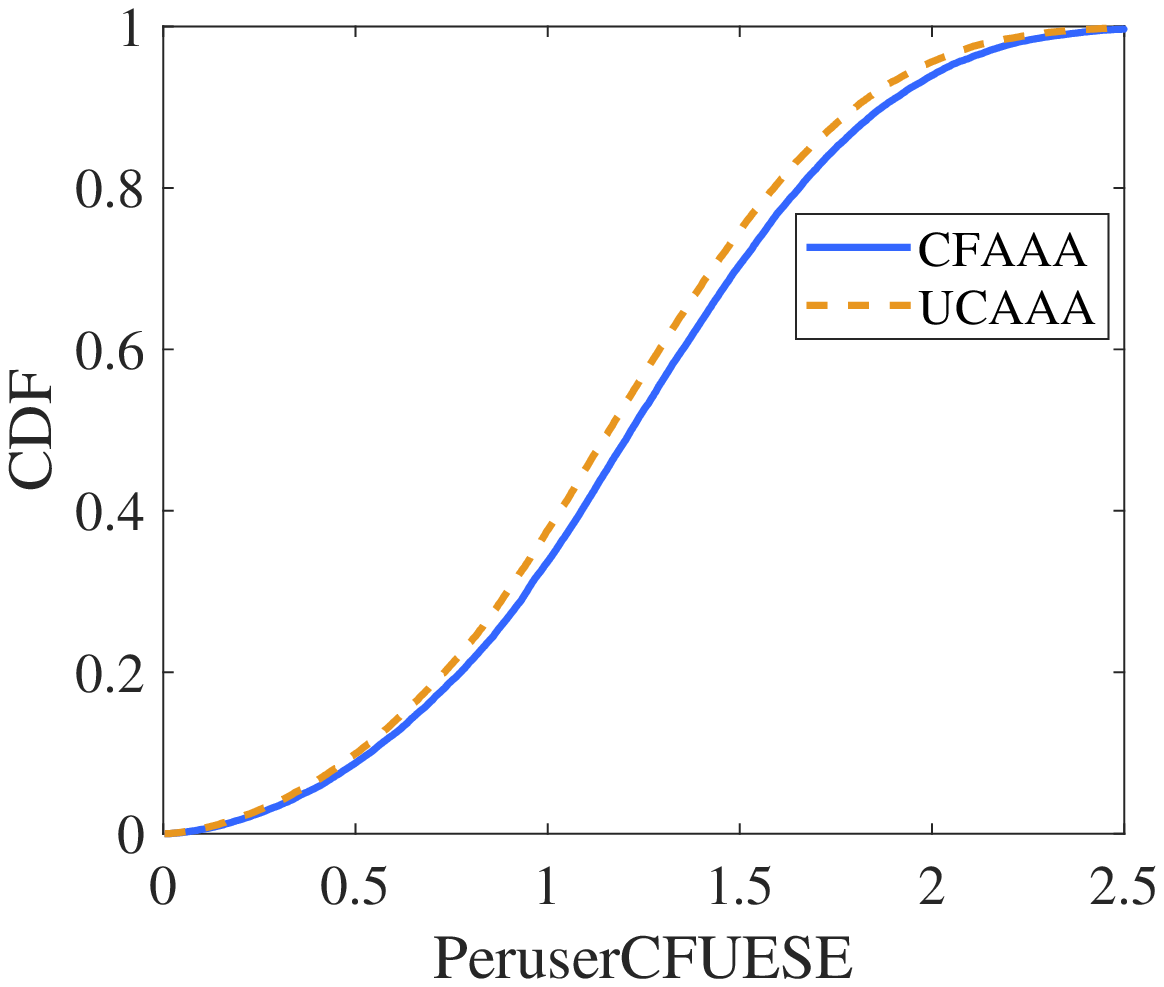}{\psfrag{0}[][][0.9]{0}
  \psfrag{0.5}[][][0.9]{0.5}
  \psfrag{1}[][][0.9]{1}
  \psfrag{1.5}[][][0.9]{1.5}
  \psfrag{2.5}[][][0.9]{2.5}
  \psfrag{2}[][][0.9]{2}
  \psfrag{0.2}[][][0.9]{0.2}
  \psfrag{0.4}[][][0.9]{0.4}
  \psfrag{0.6}[][][0.9]{0.6}
  \psfrag{0.8}[][][0.9]{0.8}
  \psfrag{CDF}[][][0.9]{CDF}
  \psfrag{CFAAA}[][][0.74]{Non-UC}
  \psfrag{UCAAA}[][][0.74]{UC}
  \psfrag{PeruserCFUESE}[][][0.9]{Per user SE [bps/Hz]}}\label{fig:uc1A}}
  \hfil
  \subfloat[]{\psfragfig[scale=.6]{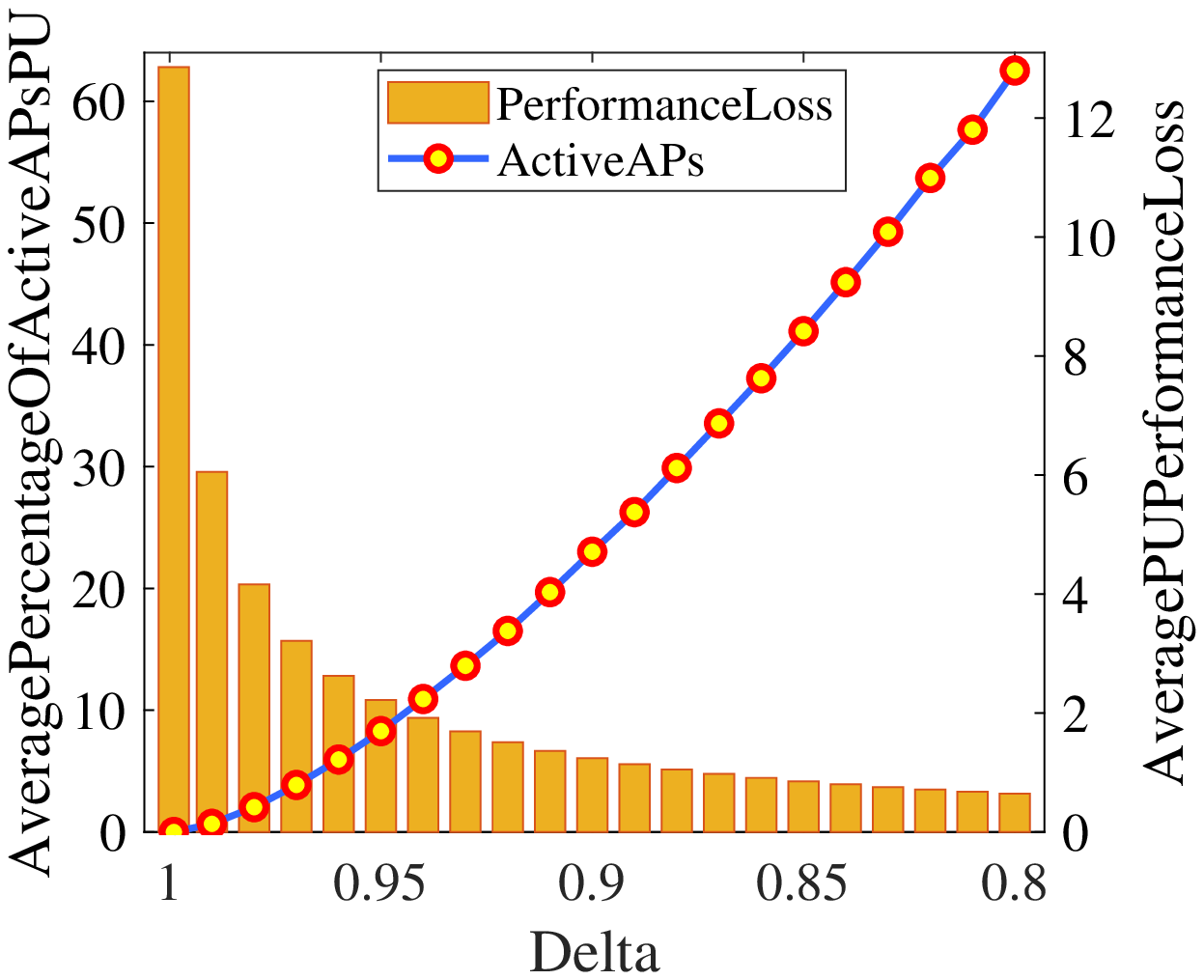}{
  \psfrag{0}[][][0.9]{0}
  \psfrag{10}[][][0.9]{10}
  \psfrag{1}[][][0.9]{1}
  \psfrag{20}[][][0.9]{20}
  \psfrag{30}[][][0.9]{30}
  \psfrag{40}[][][0.9]{40}
  \psfrag{50}[][][0.9]{50}
  \psfrag{60}[][][0.9]{60}
  \psfrag{0.8}[][][0.9]{0.8}
  \psfrag{0.85}[][][0.9]{0.85}
  \psfrag{0.9}[][][0.9]{0.9}
  \psfrag{0.95}[][][0.9]{0.95}
  \psfrag{2}[][][0.9]{2}
  \psfrag{4}[][][0.9]{4}
  \psfrag{6}[][][0.9]{6}
  \psfrag{8}[][][0.9]{8}
  \psfrag{12}[][][0.9]{12}
  \psfrag{Delta}[][][1.1]{$\delta$}
  \psfrag{AveragePercentageOfActiveAPsPU}[][][0.9]{Average serving APs per user \%}
  \psfrag{PerformanceLoss}[][][0.9]{Serving APs}
  \psfrag{ActiveAPs}[][][0.9]{\ \ \ SE loss}
  \psfrag{AveragePUPerformanceLoss}[][][0.9]{Average per user SE loss \%}}\label{fig:uc1B}}
  \caption{Performance of CFUEs with UC CF-mMIMO system, random pilot allocation, and full power transmission. (a) CDF of per user SE for $\delta = 0.9$. (b) Percentage of the APs serving each user and corresponding SE loss for $0.8\leq \delta \leq 0.999$.}\label{fig:ucc1}
  \end{figure}
  
 In Fig.~\ref{fig:ucc1}, performance of the system using UC approach for different values of design parameter $\delta$ is examined. Fig.~\ref{fig:uc1A} depicts the CDF of per user SE where Non-UC denotes the case that users are served by all APs. When UC is employed users experience performance loss; however, this loss is minimal, in particular the 5\%-outage rate loss is negligible, since the out-of-service APs for each user are those that locate farther away from that users which are determined by the design parameter $\delta = 0.9$. Fig.~\ref{fig:uc1B} shows the percentage of the APs that serve each user on average and the average per user SE loss for different $\delta$. From this figure, it is seen that for $\delta = 0.95$ only 12\% of the APs or 24 APs out of 200 APs are involved in serving each user on average and as the SE loss curve shows it causes less than 2\% loss in exchange for this massive reduction in the number of serving APs.
  % ========= Diff. pilot alloc. methods ===========
\begin{figure}[!t]
  \centering
  \subfloat[]{\psfragfig[scale=.59]{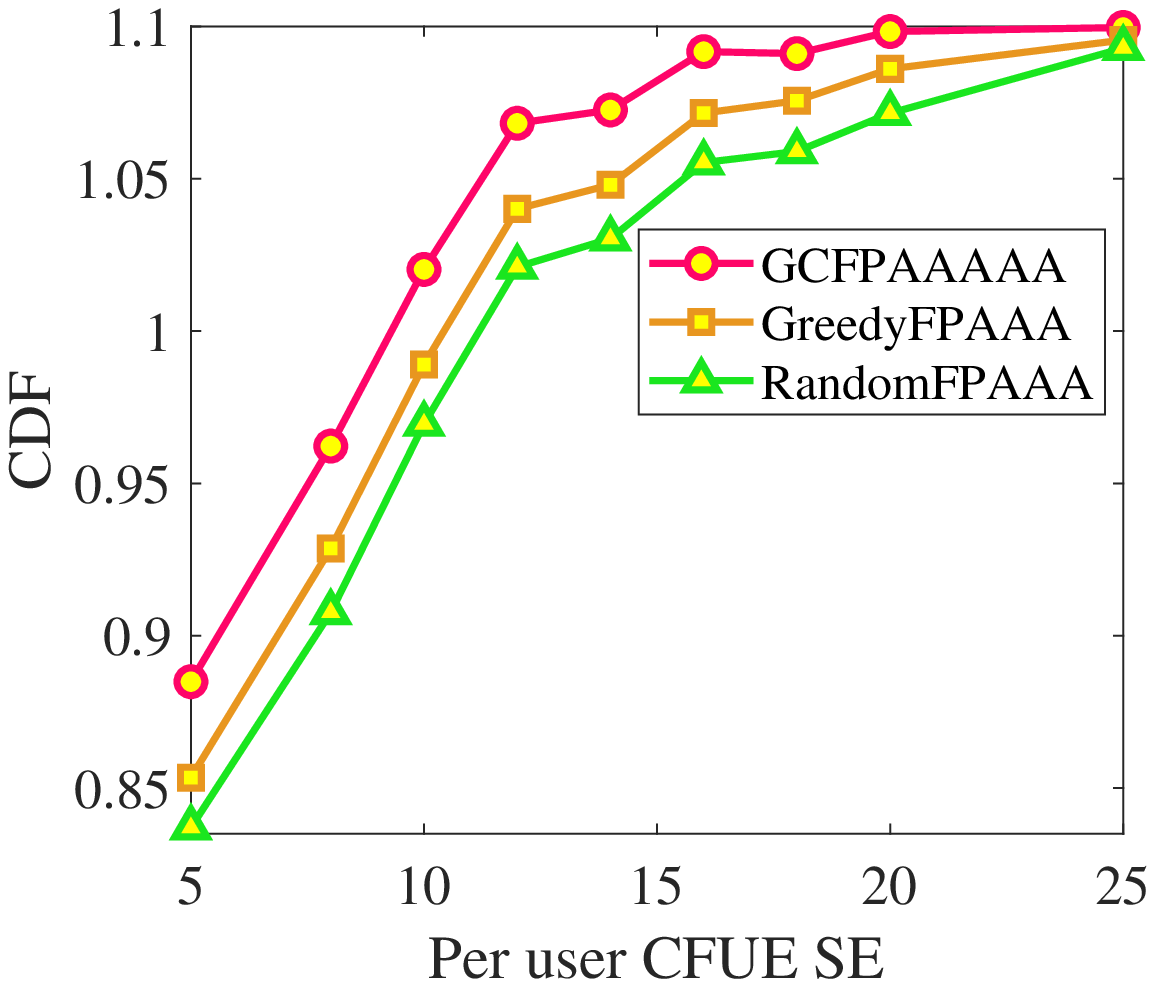}{\psfrag{1}[][][0.96]{1}
  \psfrag{1.1}[][][0.96]{1.1}
  \psfrag{1}[][][0.96]{1}
  \psfrag{1.05}[][][0.96]{1.05}
  \psfrag{0.95}[][][0.96]{0.95}
  \psfrag{0.9}[][][0.96]{0.9}
  \psfrag{0.85}[][][0.96]{0.85}
  \psfrag{5}[][][0.96]{5}
  \psfrag{10}[][][0.96]{10}
  \psfrag{15}[][][0.96]{15}
  \psfrag{20}[][][0.96]{20}
  \psfrag{25}[][][0.96]{25}
  \psfrag{CDF}[][][1]{Average per user SE [bps/Hz]}
  \psfrag{GCFPAAAAA}[][][0.8]{Graph coloring}
  \psfrag{GreedyFPAAA}[][][0.8]{Greedy}
  \psfrag{RandomFPAAA}[][][0.8]{Random}
  \psfrag{Per user CFUE SE}[][][0.96]{$\tau$}}\label{fig:pilta}}
  \hfil
  \subfloat[]{\psfragfig[scale=.59]{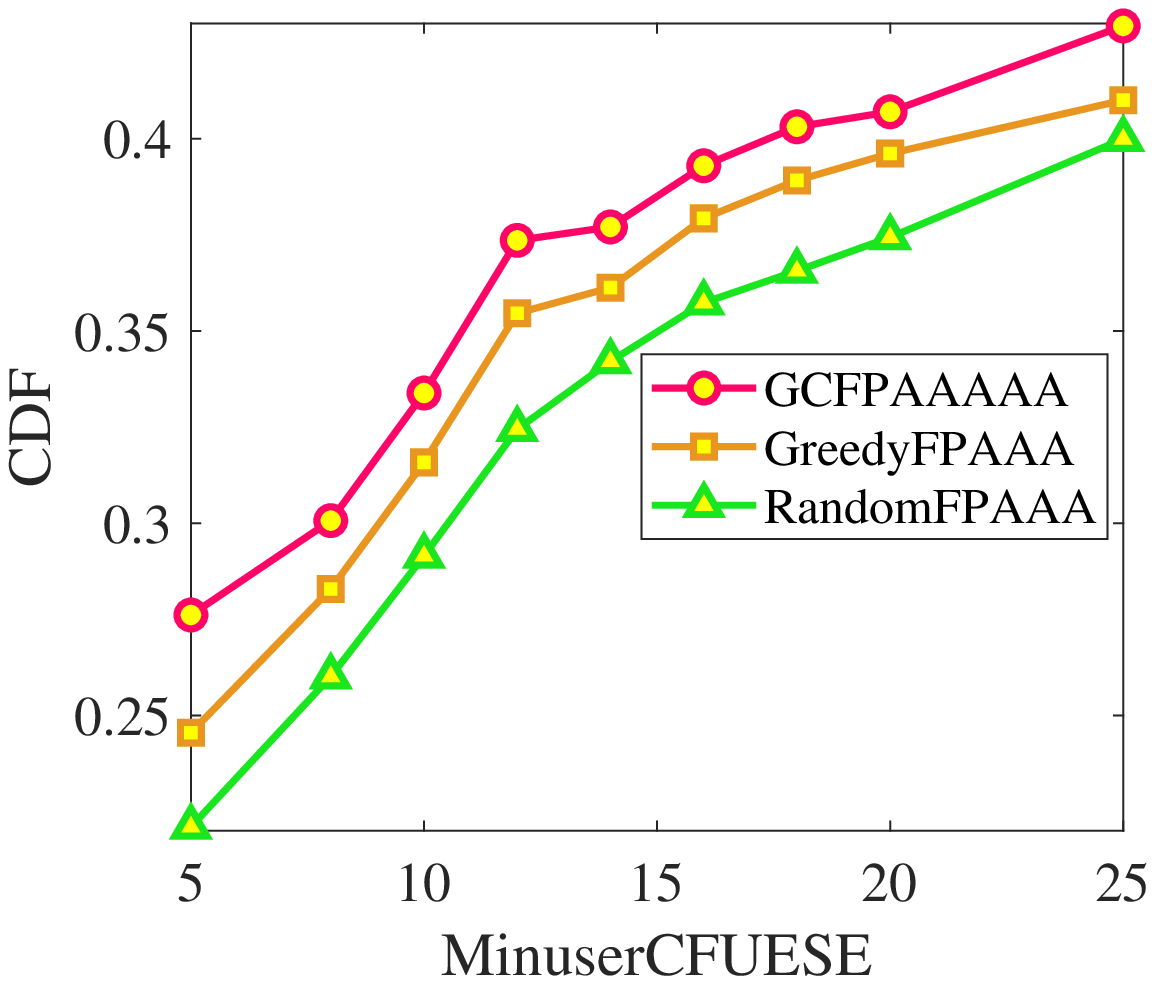}{\psfrag{0.25}[][][0.96]{0.25}
  \psfrag{0.3}[][][0.96]{0.3}
  \psfrag{0.35}[][][0.96]{0.35}
  \psfrag{0.4}[][][0.96]{0.4}
  \psfrag{5}[][][0.96]{5}
  \psfrag{10}[][][0.96]{10}
  \psfrag{15}[][][0.96]{15}
  \psfrag{20}[][][0.96]{20}
  \psfrag{25}[][][0.96]{25}
  \psfrag{CDF}[][][1]{Average minimum SE [bps/Hz]}
  \psfrag{GCFPAAAAA}[][][0.8]{Graph coloring}
  \psfrag{GreedyFPAAA}[][][0.8]{Greedy}
  \psfrag{RandomFPAAA}[][][0.8]{Random}
  \psfrag{MinuserCFUESE}[][][0.96]{$\tau$}}\label{fig:piltb}}
  \caption{Comparison of different pilot allocation methods for CFUEs. (a) Average per user SE. (b) Average minimum SE.}\label{fig:pilt}
  \end{figure}
  
\subsection{Impact of Pilot Allocation and Power Control Strategies}  
We now examine the performance gains granted by the proposed pilot allocation and power control algorithms. 

Figure \ref{fig:pilta} and Fig.~\ref{fig:piltb} compare the performance of graph coloring, greedy, and random pilot allocations. As shown in both figures, the GC-based pilot allocation outperforms the other two methods. Based on these results, by adopting GC pilot allocation one can reach a performance gain of up to 12\% for minimum SE and 4\% for per user SE in comparison to the greedy method, and gains of up to 25\% for minimum SE and 6\% for per user SE in comparison with random pilot allocation.
%Also, increasing the length of pilot, i.e. $\tau$, enhances the SE since it improves the channel estimation quality.
% =========== SCA vs. GP ===========
\begin{figure}[!t]
  \centering
  \subfloat[]{\psfragfig[scale=.59]{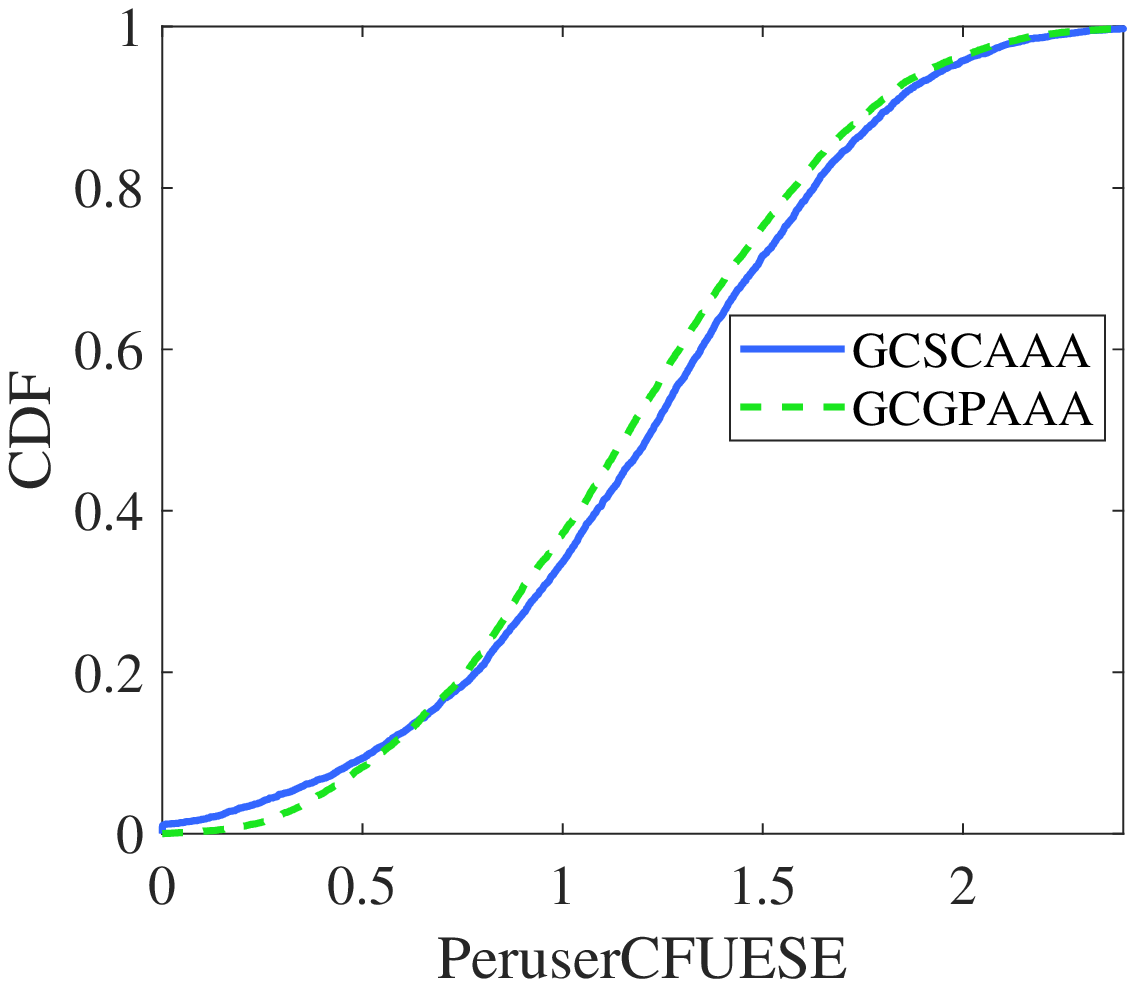}{\psfrag{0}[][][0.96]{0}
  \psfrag{0.5}[][][0.96]{0.5}
  \psfrag{1}[][][0.96]{1}
  \psfrag{1.5}[][][0.96]{1.5}
  \psfrag{2}[][][0.96]{2}
  \psfrag{0.2}[][][0.96]{0.2}
  \psfrag{0.4}[][][0.96]{0.4}
  \psfrag{0.6}[][][0.96]{0.6}
  \psfrag{0.8}[][][0.96]{0.8}
  \psfrag{CDF}[][][1]{CDF}
  \psfrag{GCSCAAA}[][][1]{SCA}
  \psfrag{GCGPAAA}[][][1]{GP}
  \psfrag{PeruserCFUESE}[][][1]{Per user SE [bps/Hz]}
  \psfrag{Per user CFUE SE}[][][0.96]{$\tau$}}\label{fig:gpscaA}}
  \hfil
  \subfloat[]{\psfragfig[scale=.59]{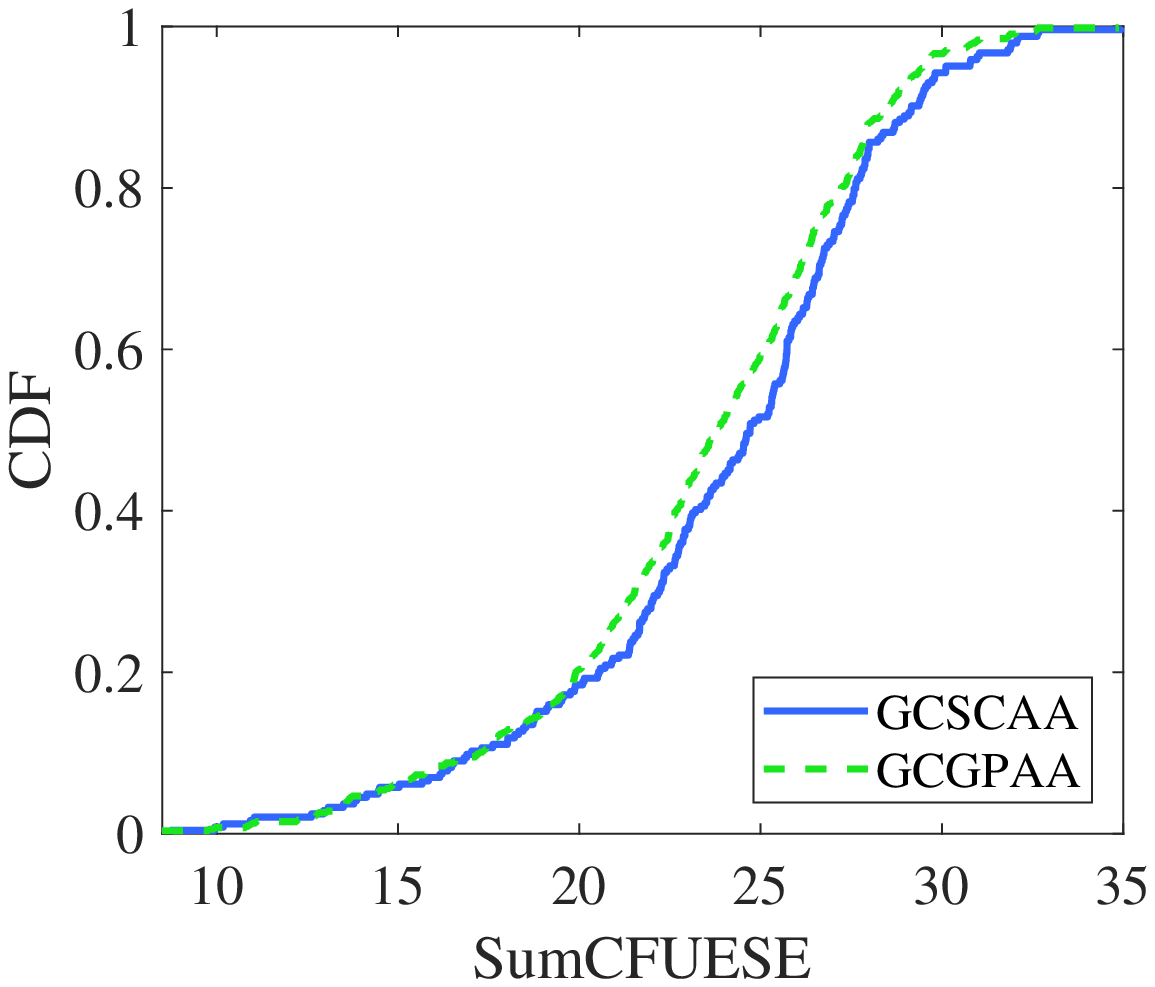}{
  \psfrag{0}[][][0.96]{0}
  \psfrag{15}[][][0.96]{15}
  \psfrag{10}[][][0.96]{10}
  \psfrag{35}[][][0.96]{35}
  \psfrag{30}[][][0.96]{30}
  \psfrag{1}[][][0.96]{1}
  \psfrag{25}[][][0.96]{25}
  \psfrag{20}[][][0.96]{20}
  \psfrag{0.2}[][][0.96]{0.2}
  \psfrag{0.4}[][][0.96]{0.4}
  \psfrag{0.6}[][][0.96]{0.6}
  \psfrag{0.8}[][][0.96]{0.8}
  \psfrag{CDF}[][][1]{CDF}
  \psfrag{GCSCAA}[][][1]{SCA}
  \psfrag{GCGPAA}[][][1]{GP}
  \psfrag{SumCFUESE}[][][1]{Sum SE [bps/Hz]}}\label{fig:gpscaB}}
  \caption{Comparison of GP-based and SCA-based power allocation solutions for $\mathcal{P}_{1}$ and $R_{l,\min} = 0.1$ for $l=1,2,...,L$. (a) Per user SE. (b) Sum SE.}\label{fig:gpsca}
  \end{figure}
  
  In Fig.~\ref{fig:gpsca} the two power allocation approaches, i.e. SCA and GP, for solving $\mathcal{P}_{1}$ are compared. As we expected for per user SE in Fig.~\ref{fig:gpscaA} at high SNR regime the performance of the two power allocations overlap while at low SNR, GP has better performance for per user data rates because it approximates the sum of the SEs with the product of SNRs that in turn prevents SNRs of having near-zero values. On the other hand, in terms of the objective function which is the sum SE of CFUEs, as it is evident from Fig.~\ref{fig:gpscaB} SCA outperforms the GP solution.
% =========== P_1 ===========
\begin{figure}[!t]
  \centering
  \subfloat[]{\psfragfig[scale=.59]{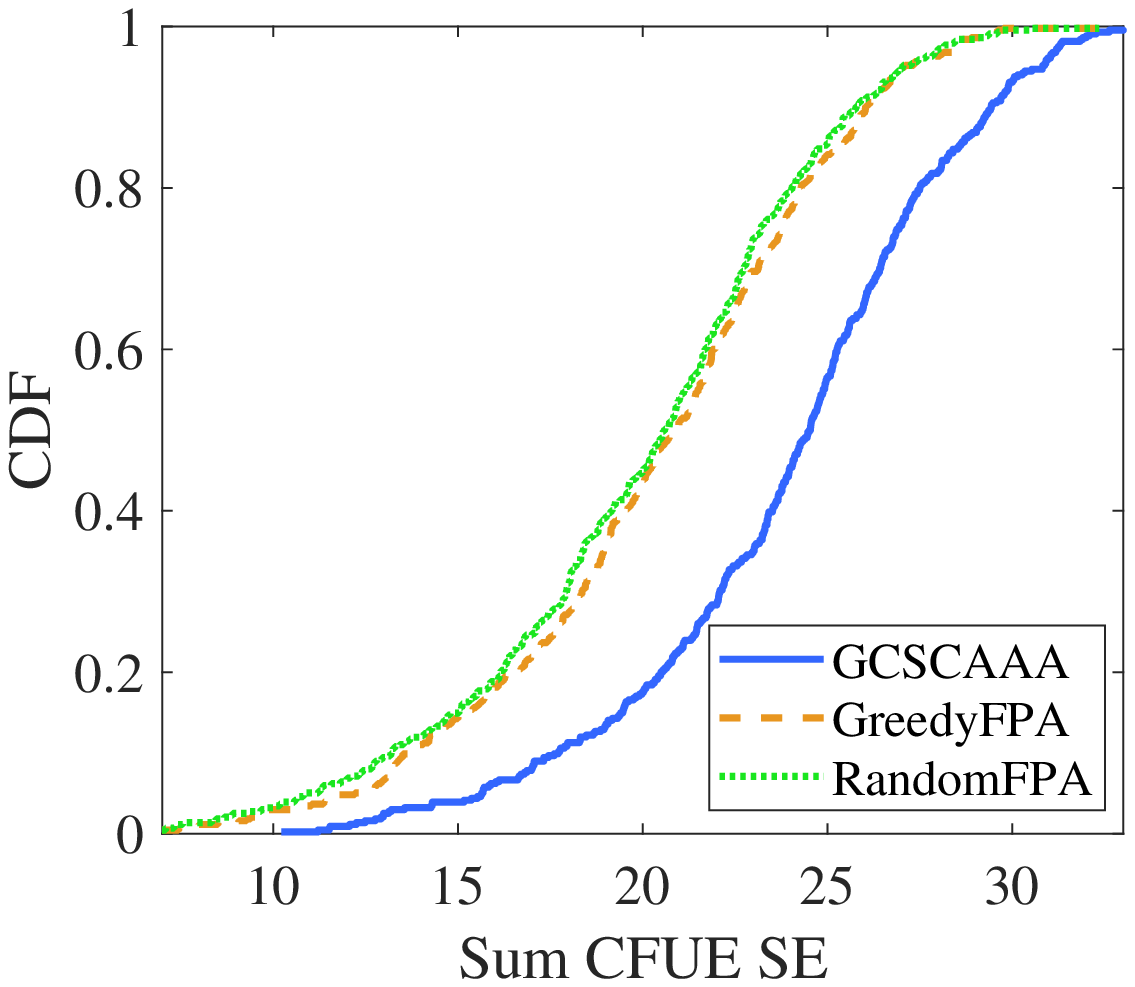}{
  \psfrag{0}[][][0.96]{0}
  \psfrag{0.5}[][][0.96]{0.5}
  \psfrag{1}[][][0.96]{1}
  \psfrag{1.5}[][][0.96]{1.5}
  \psfrag{2}[][][0.96]{2}
  \psfrag{0.2}[][][0.96]{0.2}
  \psfrag{0.4}[][][0.96]{0.4}
  \psfrag{0.6}[][][0.96]{0.6}
  \psfrag{0.8}[][][0.96]{0.8}
  \psfrag{CDF}[][][1]{CDF}
  \psfrag{GCSCAAA}[][][0.8]{SCA-GC}
  \psfrag{GreedyFPA}[][][0.8]{FP-Greedy}
  \psfrag{RandomFPA}[][][0.8]{FP-Random}
  \psfrag{Sum CFUE SE}[][][1]{Sum SE [bps/Hz]}}\label{fig:sca1A}}
  \hfil
  \subfloat[]{\psfragfig[scale=.59]{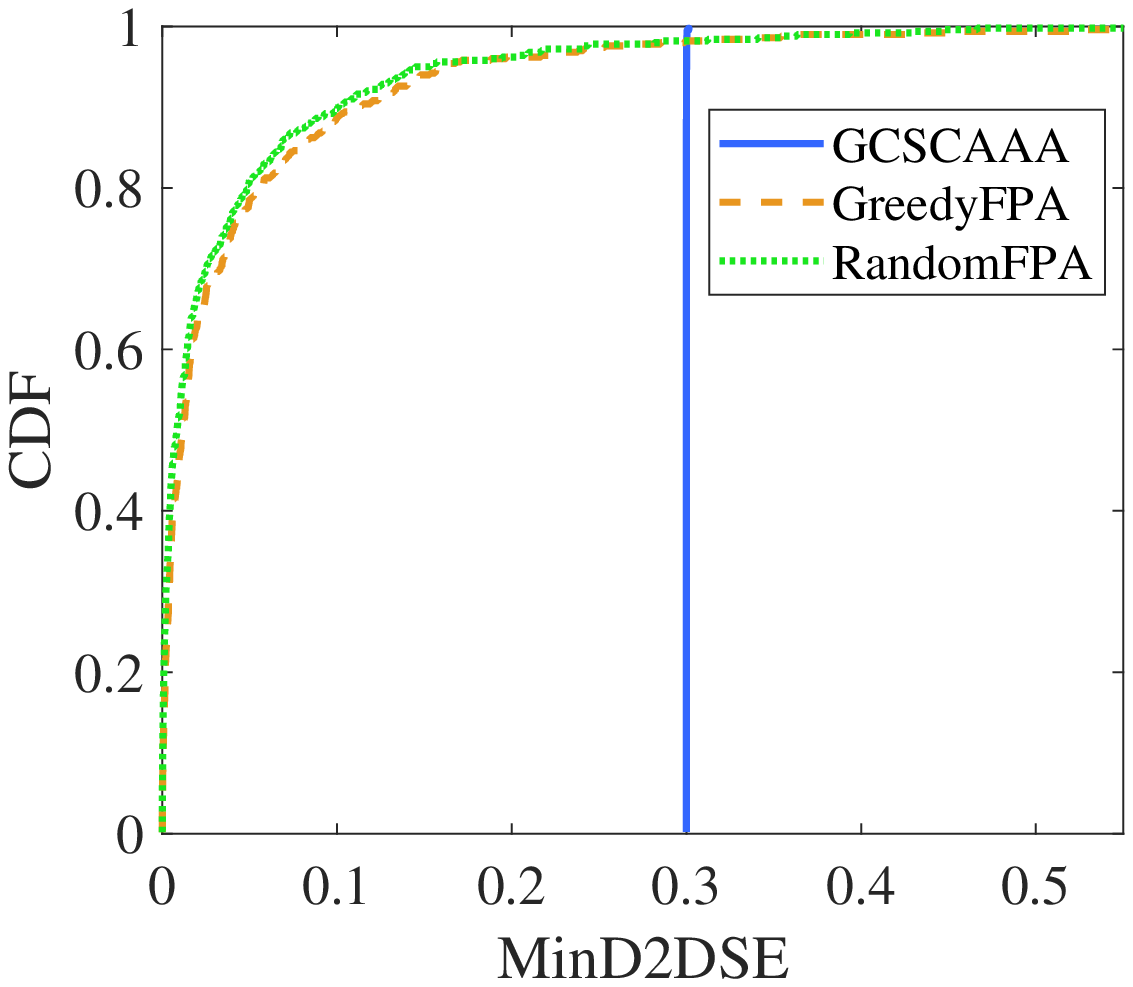}{
  \psfrag{0}[][][0.96]{0}
  \psfrag{0.1}[][][0.96]{0.1}
  \psfrag{1}[][][0.96]{1}
  \psfrag{0.5}[][][0.96]{0.5}
  \psfrag{0.3}[][][0.96]{0.3}
  \psfrag{0.2}[][][0.96]{0.2}
  \psfrag{0.4}[][][0.96]{0.4}
  \psfrag{0.6}[][][0.96]{0.6}
  \psfrag{0.8}[][][0.96]{0.8}
  \psfrag{CDF}[][][1]{CDF}
  \psfrag{GCSCAAA}[][][0.8]{SCA-GC}
  \psfrag{GreedyFPA}[][][0.8]{FP-Greedy}
  \psfrag{RandomFPA}[][][0.8]{FP-Random}
  \psfrag{MinD2DSE}[][][1]{Minimum per user SE [bps/Hz]}}\label{fig:sca1B}}
  \hfil
  \subfloat[]{\psfragfig[scale=.7]{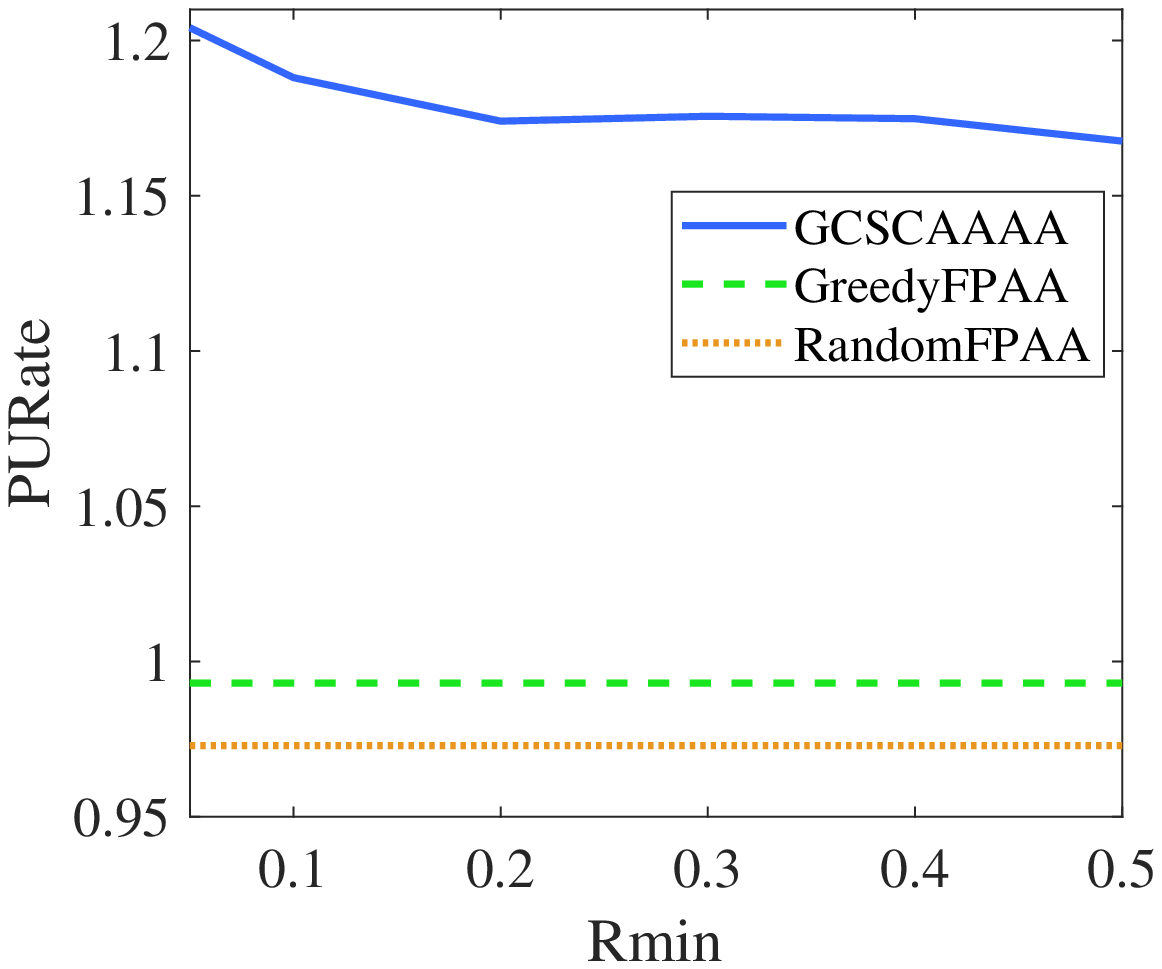}{
  \psfrag{0}[][][0.96]{0}
  \psfrag{0.1}[][][0.96]{0.1}
  \psfrag{1}[][][0.96]{1}
  \psfrag{0.5}[][][0.96]{0.5}
  \psfrag{0.3}[][][0.96]{0.3}
  \psfrag{0.2}[][][0.96]{0.2}
  \psfrag{0.4}[][][0.96]{0.4}
  \psfrag{0.6}[][][0.96]{0.6}
  \psfrag{0.8}[][][0.96]{0.8}
  \psfrag{0.95}[][][0.96]{0.95}
  \psfrag{1.05}[][][0.96]{1.05}
  \psfrag{1.1}[][][0.96]{1.1}
  \psfrag{1.15}[][][0.96]{1.15}
  \psfrag{1.2}[][][0.96]{1.2}
  \psfrag{PURate}[][][1]{Average per user SE [bps/Hz]}
  \psfrag{GCSCAAAA}[][][0.8]{SCA-GC}
  \psfrag{GreedyFPAA}[][][0.8]{FP-Greedy}
  \psfrag{RandomFPAA}[][][0.8]{FP-Random}
  \psfrag{Rmin}[][][1]{$R_{\min}$ [bps/Hz]}}\label{fig:rminP1}}
  \caption{Performance of the system with SCA-based power allocation and GC-based pilot allocation compared to full power transmission and random/greedy-based pilot allocation for $\mathcal{P}_{1}$, and $R_{l,\min} = R_{\min}$, $l=1,2,...,L$. (a) CDF of sum SE for CFUEs with $R_{l,\min} = 0.3$. (b) CDF of minimum SE for DUEs with $R_{l,\min} = 0.3$. (c) Average per user SE of CFUEs versus $R_{\min}$.}\label{fig:sca}
  \end{figure}
  
\textcolor{black}{The combination of full power (FP) transmission with greedy-based and random pilot allocation which are widely used in many prior works \cite{ngo2017cell,nayebi2017precoding,bjornson2019making,mai2018pilot,liu2019tabu,jin2020spectral,hu2019cell,zhu2015graph} are considered as baseline approaches to highlight the performance of the proposed method.} Fig.~\ref{fig:sca1A} shows tremendous improvements in sum data rate of CFUEs. In terms of 5\%-outage, SCA with GC-based pilot allocation achieves the performance gains of 28\% up to 43\% in comparison to FP transmission with greedy-based and random pilot allocation, respectively. Moreover, for DUEs Fig.~\ref{fig:sca1B} verifies that all the users comply with the minimum data rate constraint.
Next, Fig.~\ref{fig:rminP1} depicts average per user SE for CFUEs versus $R_{\min}$. As it is shown, by adopting SCA-based power allocation and GC-based pilot allocation, one can reach at least 17\% better per user SE in comparison to other two trivial power and pilot allocation methods.
  \begin{figure}[t!]
  \centering
  \subfloat[]{\psfragfig[scale=.6]{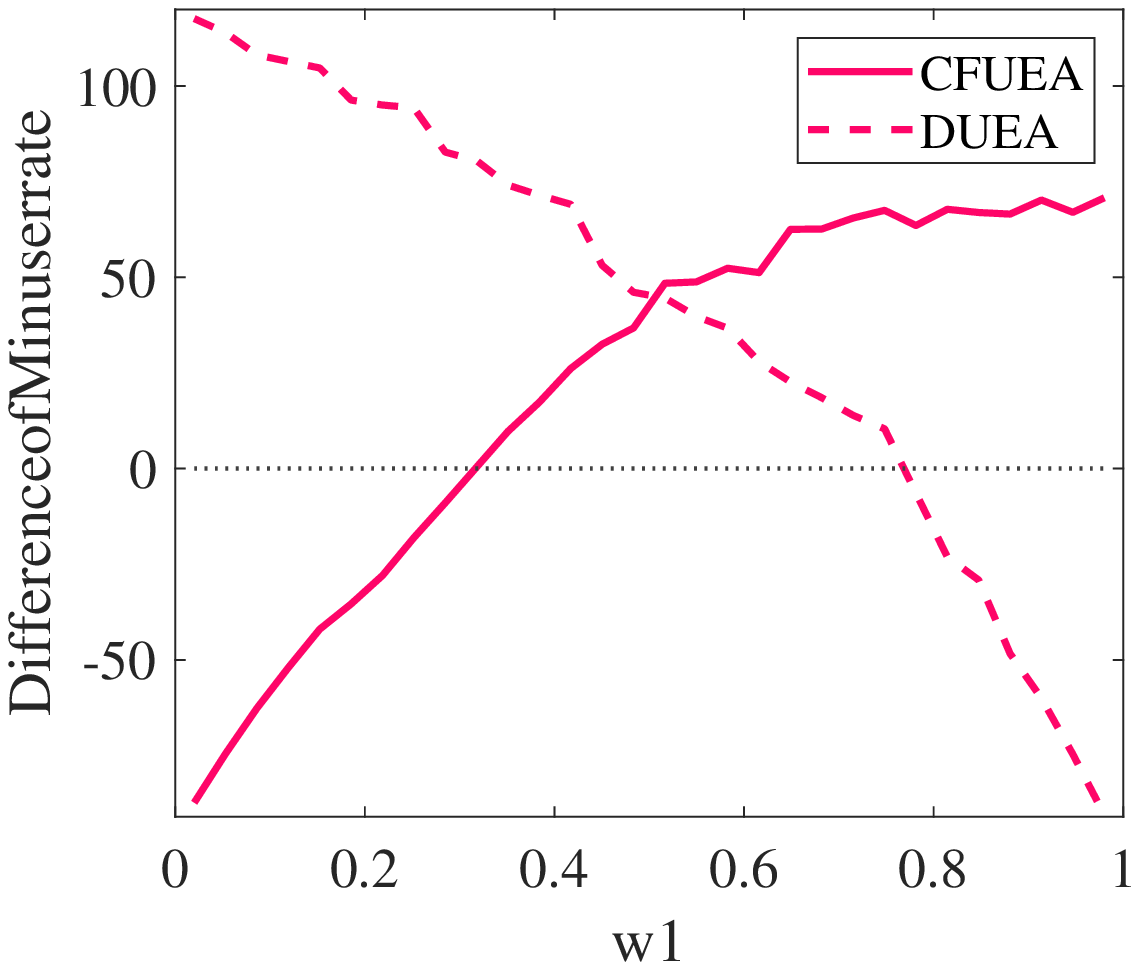}{\psfrag{0}[][][0.96]{0}
  \psfrag{0.6}[][][0.96]{0.6}
  \psfrag{0.2}[][][0.96]{0.2}
  \psfrag{0.8}[][][0.96]{0.8}
  \psfrag{0.4}[][][0.96]{0.4}
  \psfrag{1}[][][0.96]{1}
  \psfrag{-50}[][][0.96]{-50}
  \psfrag{100}[][][0.96]{100}
  \psfrag{50}[][][0.96]{50}
  \psfrag{DifferenceofMinuserrate}[][][0.96]{$100\!\!\times\!\!(R_{min}^{opt}-R_{min})/R_{min}$}
  \psfrag{CFUEA}[][][0.88]{CFUE}
  \psfrag{DUEA}[][][0.87]{DUE}
  \psfrag{w1}[][][0.96]{$w^{c}$}}\label{fig:N7a}}
  \hfil
  \subfloat[]{\psfragfig[scale=.6]{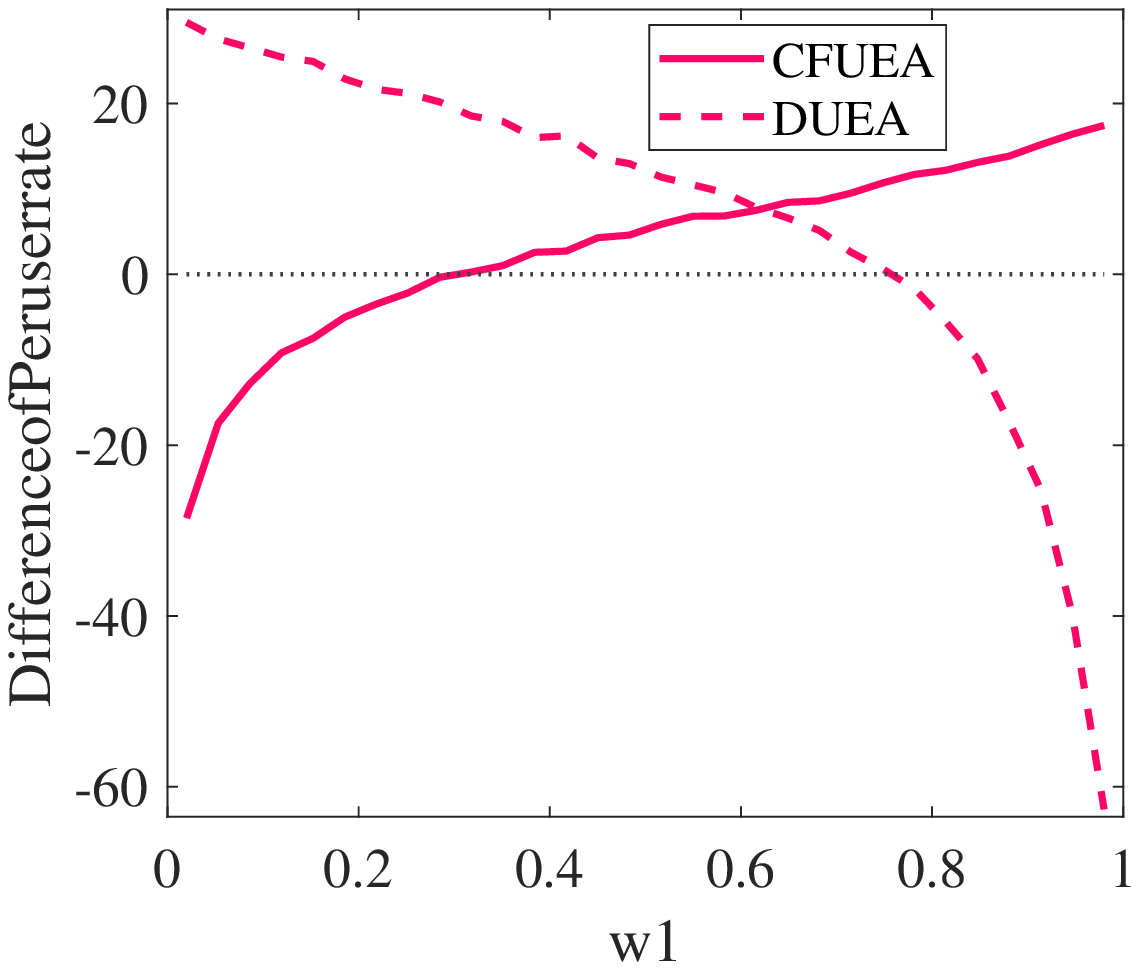}{\psfrag{0}[][][0.74]{0}
  \psfrag{0.6}[][][0.96]{0.6}
  \psfrag{0.2}[][][0.96]{0.2}
  \psfrag{0.8}[][][0.96]{0.8}
  \psfrag{0.4}[][][0.96]{0.4}
  \psfrag{1}[][][0.96]{1}
  \psfrag{-60}[][][0.96]{-60}
  \psfrag{-40}[][][0.96]{-40}
  \psfrag{-20}[][][0.96]{-20}
  \psfrag{20}[][][0.96]{20}
  \psfrag{DifferenceofPeruserrate}[][][0.96]{$100\!\!\times\!\!(R^{opt}-R)/R$}
  \psfrag{CFUEA}[][][0.86]{CFUE}
  \psfrag{DUEA}[][][0.87]{DUE}
  \psfrag{w1}[][][0.96]{$w^{c}$}}\label{fig:N7b}}
  %\hfil
  %\subfloat[]{\psfragfig[scale=.6]{./Figures/FinalResults/opSUEvsW1}{\psfrag{0}[][][0.96]{0}
  %\psfrag{0.6}[][][0.96]{0.6}
  %\psfrag{0.2}[][][0.96]{0.2}
  %\psfrag{0.8}[][][0.96]{0.8}
  %\psfrag{0.4}[][][0.96]{0.4}
  %\psfrag{1}[][][0.96]{1}
  %\psfrag{-60}[][][0.96]{-60}
  %\psfrag{-40}[][][0.96]{-40}
  %\psfrag{-20}[][][0.96]{-20}
  %\psfrag{20}[][][0.96]{20}
  %\psfrag{Differenceofsumuserrate}[][][0.96]{$100\!\!\times\!\!(R_{sum}^{opt}-R_{sum})/R_{sum}$}
  %\psfrag{CFUEA}[][][0.86]{CFUE}
  %\psfrag{DUEA}[][][0.87]{DUE}
  %\psfrag{w1}[][][0.96]{$w^{c}$}}\label{fig:N7c}}
  \caption{Perfomance of DUEs and CFUEs, without pilot allocation and power control and with pilot allocation and WMPCD power control for different values of weight $w^{c} = 1 - w^{d}$ where $w^{c}\in [0,1]$. (a) minimum SE. (b) Per user SE.}\label{fig:N7}% (c) Sum SE.
  \end{figure} %Note that, $R_{min}^{opt}$, $R^{opt}$ and $R_{sum}^{opt}$ indicate minimum rate among the users, per user rate and sum rate of the users with WMPCD power control and GC-based ...
  % Regarding the per user SE and sum SE, the same performance enhancement for both types ...
  
  In Fig.~\ref{fig:N7}, performance of the system with the WMPCD power control for different weights is evaluated. Note that, $R_{min}^{opt}$ and $R^{opt}$ indicate minimum rate among the users and per user rate with WMPCD power control and GC-based pilot assignment, respectively, while those without the superscript $opt$ denote the case of full power and random pilot assignment. When the curves are above the horizontal dot-line, we have performance improvements compared to that of without resource allocation. For smaller $w^{c}$, DUEs have higher priority than the CFUEs, so the objective function is  maximized by improving the DUEs' data rate and cutting down the interference from CFUEs. As a result, less power is dedicated for data transmission of CFUEs. In contrast, when $w^{c}$ gets closer to 1, CFUEs gain higher priority and the behaviour of system is justified in the similar way. Moreover, the performance improvement is much more pronounced for the minimum SE, as depicted in Fig.~\ref{fig:N7a}, which shows near 50\% improvement for both DUEs and CFUEs (for the case $w^{c}\approx 0.5$). Regarding the per user SE, the same performance enhancement for both types of users occurs around $w^{c}\approx 0.6$. Also, as it is observed from these figures for $0.3\leq w^{c}\leq 0.8$ we have improvements for both DUEs and CFUEs, though by varying $w^{c}$ in the given interval one can reach different trade-offs between CFUEs and DUEs.
  
%   \textcolor{purple}{This section provided extensive numerical result to illustrate and evaluate the effectiveness of proposed algorithms and the impacts of different system parameters like low resolution ADCs.}

\section{Conclusions}\label{sec6conclusion}
We considered the issue of the coexistence between D2D links and the uplink of a CF-mMIMO system, assuming that finite resolution ADCs were used at the APs. Also, to make the system scalable we assumed that each AP only serves a few of the users, an approach is known as the user-centric. For the cases of both perfect and imperfect CSI at the receivers, closed-form expressions of achievable data rates for both CFUEs and DUEs have been derived. In order to obtain estimates of D2D pairs' and CFUEs' channels,  greedy-based and graph coloring-based algorithms to assign pilot sequences among different users have been proposed, so as to control the resultant pilot contamination. Next, two power allocation problems have been explored. The first one maximizes the CFUEs' sum SE subject to QoS constraints on DUEs data rates; and the second one maximizes the weighted product of SINRs of CFUEs and DUEs. GP theory and SCA were used to solve both problems and the convergence and complexity of proposed algorithms and optimization problems are discussed. Numerical results have shown that the activation of D2D links provides a considerable gain to the network throughput. Also, the performance loss due to finite resolution ADCs can be kept under control by using an adequate level of resolution; in particular, in the considered scenario,  using 4-bit ADCs provided almost the same performance as that of the infinite resolution ADCs' in terms of SE. Results also showed that properly designed pilot allocation and power control schemes can bring remarkable performance improvements to the network throughput. Further work in this area is currently addressing the analysis of the downlink CF-mMIMO system, as well as the design of proper D2D link activation protocols. 
%------------------------------------------------------------------------------------------------------------
%%%%%%%%%%%%%%%%%%%%%%%%%%%%%%%%

%\newpage
%\appendix
\vspace{-.4cm}
\section*{Appendix A}\label{app:1}
It can be shown that all the terms in $\mathcal{I}_{k}^{c}$ are mutually uncorrelated, therefore
\begin{equation}\label{eqn:DSBUapp1}
\begin{split}
  \mathbb{E}\{\text{DS}_{k}\} &= \xi\sqrt{\eta_{k}^{c}\rho^{c}}\sum\limits_{m\in\mathcal{M}_{k}}^{}\mathbb{E}\left\lbrace|h_{mk}^{c}|^{2}\right\rbrace =  \xi\sqrt{\eta_{k}^{c}\rho^{c}}\sum\limits_{m\in\mathcal{M}_{k}}^{}\beta_{mk}^{c},\\
  \mathbb{E}\{\text{BU}_{k}\} &= \xi^{2}\eta_{k}^{c}\rho^{c}\mathbb{E}\!\left\lbrace\! \left\lvert \sum\limits_{m\in\mathcal{M}_{k}}^{}\!|h_{mk}^{c}|^{2} \!-\!\! \sum\limits_{m\in\mathcal{M}_{k}}^{}\!\!\mathbb{E}\!\left\lbrace|h_{mk}^{c}|^{2}\right\rbrace \right\rvert^{2} \!\right\rbrace\! = \xi^{2}\eta_{k}^{c}\rho^{c}\mathbb{E}\!\!\left\lbrace\! \left\lvert \sum\limits_{m\in\mathcal{M}_{k}}^{}\!\left(|h_{mk}^{c}|^{2}\!\! - \!\beta_{mk}^{c}\right) \right\rvert^{2} \!\right\rbrace.
  \end{split}
\end{equation}
  Since the terms $\left(|h_{mk}^{c}|^{2}\!\! - \!\beta_{mk}^{c}\right)$ are uncorrelated for different $m$, the above equation \eqref{eqn:DSBUapp1} can be simplified as
  \begin{equation}\label{eqn:BUapp1}
  \mathbb{E}\!\{\!\text{BU}_{k}\!\} \!=\! \xi^{2}\eta_{k}^{c}\rho^{c} \!\sum\limits_{m\in\mathcal{M}_{k}}^{}\!\mathbb{E}\!\!\left\lbrace\!\left\lvert|h_{mk}^{c}|^{2}\!\! - \!\beta_{mk}^{c}\right\rvert^{2} \!\right\rbrace \!=\! \xi^{2}\eta_{k}^{c}\rho^{c} \!\sum\limits_{m\in\mathcal{M}_{k}}^{}\!\!\left(\!\mathbb{E}\!\left\lbrace\!|h_{mk}^{c}|^{4} \!\right\rbrace\!\! - \!\beta_{mk}^{c^{2}}\right) \!\overset{(a)}{=}\! \xi^{2}\eta_{k}^{c}\rho^{c} \!\sum\limits_{m\in\mathcal{M}_{k}}^{}\!\beta_{mk}^{c^{2}},
\end{equation}
where $(a)$ is due to $\mathbb{E}\!\left\lbrace\!|h_{mk}^{c}|^{4} \!\right\rbrace\! = 2\beta_{mk}^{c^{2}}$. Also, for computing $\text{Var}(\mathcal{I}_{k}^{c})$ we have
\begin{equation}\label{eqn:varapp1}
\begin{split}
\text{Var}(\mathcal{I}_{k}^{c}) \!=\! \mathbb{E}\!\left\lbrace\! \left\lvert \mathcal{I}_{k}^{c} \right\rvert^{2} \!\right\rbrace \!-\! \underbrace{\left\lvert \mathbb{E}\left\lbrace \mathcal{I}_{k} \right\rbrace \right\rvert^{2}}_{=0} \!=\! \mathbb{E}\!\left\lbrace\! \left\lvert \text{ICFUE\textsubscript{$k$}} \right\rvert^{2} \!\right\rbrace \!+\! \mathbb{E}\!\left\lbrace\! \left\lvert \text{IDUE\textsubscript{$k$}} \right\rvert^{2} \!\right\rbrace\! +\! \mathbb{E}\!\left\lbrace\! \left\lvert \text{TN\textsubscript{$k$}} \right\rvert^{2} \!\right\rbrace\! +\! \mathbb{E}\!\left\lbrace\! \left\lvert \text{QN\textsubscript{$k$}} \right\rvert^{2} \!\right\rbrace.
  \end{split}
\end{equation}
By computing each of the terms in the above equation, we have
\begin{equation}\label{eqn:ICUEapp1}
\begin{split}
\mathbb{E}\left\lbrace \left\lvert \text{ICFUE\textsubscript{$k$}} \right\rvert^{2} \right\rbrace &= \xi^{2}\rho^{c}\! \sum\limits_{k^{\prime}\neq k} ^ {K}\!\! \eta_{k^{\prime}}^{c}\mathbb{E}\!\left\lbrace\! \left\lvert \sum\limits_{m\in\mathcal{M}_{k}}^{}h_{mk}^{c^*}h_{mk^{\prime}}^{c}\right\rvert^{2} \right\rbrace = \xi^{2}\rho^{c}\! \sum\limits_{k^{\prime}\neq k} ^ {K}\!\! \eta_{k^{\prime}}^{c}\!\sum\limits_{m\in\mathcal{M}_{k}}^{}\!\mathbb{E}\!\left\lbrace\! \left\lvert h_{mk}^{c^*}h_{mk^{\prime}}^{c}\right\rvert^{2} \right\rbrace\\
&= \xi^{2}\rho^{c}\! \sum\limits_{k^{\prime}\neq k} ^ {K}\!\! \eta_{k^{\prime}}^{c}\!\sum\limits_{m\in\mathcal{M}_{k}}^{}\!\! \beta_{mk}^{c}\beta_{mk^{\prime}}^{c},
  \end{split}
\end{equation}
\begin{equation}\label{eqn:IDUEapp1}
\begin{split}
\mathbb{E}\left\lbrace \left\lvert \text{IDUE\textsubscript{$k$}} \right\rvert^{2} \right\rbrace &= \xi^{2}\!\rho^{d}\!\sum\limits_{l^{\prime}=1} ^ {L}\!\!\eta_{l^{\prime}}^{d}\mathbb{E}\!\left\lbrace\! \left\lvert \sum\limits_{m\in\mathcal{M}_{k}}^{}\!\!h_{mk}^{c^*}h_{ml^{\prime}}^{d}\right\rvert^{2} \right\rbrace = \xi^{2}\rho^{d}\! \sum\limits_{l^{\prime}=1} ^ {L}\!\! \eta_{l^{\prime}}^{d}\!\sum\limits_{m\in\mathcal{M}_{k}}^{}\!\mathbb{E}\!\left\lbrace\! \left\lvert h_{mk}^{c^*}h_{ml^{\prime}}^{d}\right\rvert^{2} \right\rbrace\\
&= \xi^{2}\rho^{d}\! \sum\limits_{l^{\prime}=1} ^ {L}\!\! \eta_{l^{\prime}}^{d}\!\sum\limits_{m\in\mathcal{M}_{k}}^{}\!\! \beta_{mk}^{c}\beta_{ml^{\prime}}^{d},
  \end{split}
\end{equation}
\begin{equation}\label{eqn:TNapp1}
\begin{split}
\mathbb{E}\left\lbrace \left\lvert \text{TN\textsubscript{$k$}} \right\rvert^{2} \right\rbrace &= \xi^{2}\mathbb{E}\!\left\lbrace\! \left\lvert \sum\limits_{m\in\mathcal{M}_{k}}^{}\!h_{mk}^{c^*}n_{m}\right\rvert^{2} \right\rbrace = \xi^{2}\!\sum\limits_{m\in\mathcal{M}_{k}}^{}\!\mathbb{E} \!\left\lbrace\! \left\lvert h_{mk}^{c^*}n_{m}\right\rvert^{2} \right\rbrace =  \xi^{2}N_{0}\!\sum\limits_{m\in\mathcal{M}_{k}}^{}\!\! \beta_{mk}^{c},
  \end{split}
\end{equation}
\begin{equation}\label{eqn:QNapp1}
\begin{split}
&\mathbb{E}\left\lbrace \left\lvert \text{QN\textsubscript{$k$}} \right\rvert^{2} \right\rbrace \!=\! \mathbb{E}\!\left\lbrace\! \left\lvert \sum\limits_{m\in\mathcal{M}_{k}}^{}\!h_{mk}^{c^*}q_{m}\right\rvert^{2} \!\right\rbrace =\! \!\sum\limits_{m\in\mathcal{M}_{k}}^{}\!\mathbb{E} \!\left\lbrace\! \left\lvert h_{mk}^{c}\right\rvert^{2}\mathbb{E} \!\left\lbrace\! \left\lvert q_{m}\right\rvert^{2}\Big| \{h_{mk}^{c},h_{ml}^{d}\} \!\right\rbrace\! \right\rbrace\! \overset{\eqref{eqn:Qvariance}}{=} \!\sum\limits_{m\in\mathcal{M}_{k}}^{}\!\mathbb{E} \!\Bigg\{\!\! \left\lvert h_{mk}^{c}\right\rvert^{2}(1\!-\!\xi)\xi\\
&\times\!\!\left(\!\rho^{c}\sum\limits_{k^{\prime} =1}^{K}\eta_{k^{\prime}}^{c}\left|h_{mk^{\prime}}^{c}\right|^2 \!+\! \rho^{d}\sum\limits_{l^{\prime} =1}^{L}\eta_{l^{\prime}}^{d}\left|h_{ml^{\prime}}^{d}\right|^{2} + N_{0}\right) \!\!\Bigg\} \! = \!(1-\xi)\xi\!\Bigg(\!\rho^{c}\sum\limits_{m\in\mathcal{M}_{k}}^{}\!\sum\limits_{k^{\prime} =1}^{K}\!\!\eta_{k^{\prime}}^{c}\mathbb{E} \!\left\lbrace\! \left\lvert h_{mk}^{c}\right\rvert^{2}\left|h_{mk^{\prime}}^{c}\right|^2\right\rbrace \\
&\!+\! \rho^{d}\!\sum\limits_{m\in\mathcal{M}_{k}}^{}\!\sum\limits_{l^{\prime} =1}^{L}\!\!\eta_{l^{\prime}}^{d}\mathbb{E} \!\left\lbrace\! \left\lvert h_{mk}^{c}\right\rvert^{2}\left|h_{ml^{\prime}}^{d}\right|^{2}\right\rbrace \!+\! N_{0}\!\sum\limits_{m\in\mathcal{M}_{k}}^{}\!\! \beta_{mk}^{c} \Bigg) = \!(1\!-\!\xi)\xi\rho^{c}\!\!\sum\limits_{k^{\prime} =1}^{K}\!\!\eta_{k^{\prime}}^{c}\!\!\sum\limits_{m\in\mathcal{M}_{k}}^{}\! \beta_{mk}^{c}\beta_{mk^{\prime}}^{c} \\
&+\! (1\!-\!\xi)\xi\rho^{c}\eta_{k}^{c}\!\sum\limits_{m\in\mathcal{M}_{k}}^{}\! \beta_{mk}^{c^{2}} \!+\! (1\!-\!\xi)\xi\rho^{d}\!\sum\limits_{l^{\prime} =1}^{L}\!\!\eta_{l^{\prime}}^{d}\!\sum\limits_{m\in\mathcal{M}_{k}}^{}\!\beta_{mk}^{c}\beta_{ml^{\prime}}^{d} \!+\! (1\!-\!\xi)\xi N_{0}\!\sum\limits_{m\in\mathcal{M}_{k}}^{}\!\! \beta_{mk}^{c},
\end{split}
\end{equation}
By combining equations \eqref{eqn:DSBUapp1}--\eqref{eqn:QNapp1} and using \eqref{eqn:PCSI_UatF} the achievable rate is derived.
% ********************
\section*{Appendix B}\label{app:ECF_UpperBound}
According to \eqref{eqn:jensen} we approximate $R_{l}^{\text{DUE\textsubscript{P}}}$ with $\tilde{R}_{l}^{\text{DUE\textsubscript{P}}}$ which is given by
 \begin{equation}\label{eqn:PCSI_d2dapp2_1}
     \tilde{R}_{l}^{\text{DUE\textsubscript{P}}} =\! \log_{2}\!\left(\!1\! +\! \left(\mathbb{E}\left\lbrace\dfrac{\rho^{c}\sum\limits_{k =1}^{K}\!\eta_{k}^{c}\left|\bar{g}_{lk}^{d,c}\right|^{2} \!+ \rho^{d}\sum\limits_{l^{\prime} \neq l}^{L}\!\eta_{l^{\prime}}^{d}\left|\bar{g}_{ll^{\prime}}^{d}\right|^{2} \!+\! N_{0}}{\rho^{d}\eta_{l}^{d}\left\|\boldsymbol{g}_{ll}^{d}\right\|^{2}}\right\rbrace\right)^{-1}\right)\!,
 \end{equation}
 where $\bar{g}_{lk}^{d,c} = \frac{\boldsymbol{g}_{ll}^{d^{H}}\boldsymbol{g}_{lk}^{c}}{\left\|\boldsymbol{g}_{ll}^{d}\right\|}$ and $\bar{g}_{ll^{\prime}}^{d} = \frac{\boldsymbol{g}_{ll}^{d^{H}}\boldsymbol{g}_{ll^{\prime}}^{c}}{\left\|\boldsymbol{g}_{ll}^{d}\right\|}$ which are $\mathcal{CN}(0,\psi_{lk}^{c})$ and $\mathcal{CN}(0,\psi_{ll^{\prime}}^{d})$ conditioned on $\boldsymbol{g}_{ll}^{d}$ respectively, so we can write
  \begin{equation}\label{eqn:PCSI_d2dapp2_2}
  \begin{split}
     \mathbb{E}&\left\lbrace\dfrac{\rho^{c}\sum\limits_{k =1}^{K}\!\eta_{k}^{c}\left|\bar{g}_{lk}^{d,c}\right|^{2} \!+ \rho^{d}\sum\limits_{l^{\prime} \neq l}^{L}\!\eta_{l^{\prime}}^{d}\left|\bar{g}_{ll^{\prime}}^{d}\right|^{2} \!+\! N_{0}}{\rho^{d}\eta_{l}^{d}\left\|\boldsymbol{g}_{ll}^{d}\right\|^{2}}\right\rbrace = \mathbb{E}\left\lbrace\rho^{c}\sum\limits_{k =1}^{K}\!\eta_{k}^{c}\left|\bar{g}_{lk}^{d,c}\right|^{2} \!+ \rho^{d}\sum\limits_{l^{\prime} \neq l}^{L}\!\eta_{l^{\prime}}^{d}\left|\bar{g}_{ll^{\prime}}^{d}\right|^{2} \!+\! N_{0}\right\rbrace\\
     &\times \mathbb{E}\left\lbrace\dfrac{1}{\rho^{d}\eta_{l}^{d}\left\|\boldsymbol{g}_{ll}^{d}\right\|^{2}}\right\rbrace \overset{(a)}{=} \frac{\rho^{c}\sum\limits_{k =1}^{K}\!\eta_{k}^{c}\psi_{lk}^{c} \!+ \rho^{d}\sum\limits_{l^{\prime} \neq l}^{L}\!\eta_{l^{\prime}}^{d}\psi_{ll^{\prime}}^{d} \!+\! N_{0}}{\rho^{d}\eta_{l}^{d}\psi_{ll}^{d}(N-1)},
     \end{split}
 \end{equation}
 where $(a)$ comes from the fact that $\boldsymbol{g}_{ll}^{d}$ can be written as $\boldsymbol{g}_{ll}^{d} = \sqrt{\psi_{ll}^{d}}\boldsymbol{w}$ such that $\boldsymbol{w}\sim\mathcal{CN}(0,\boldsymbol{I}_{N})$, also $\|\boldsymbol{w}\|^{2}=\frac{1}{2}\sum\limits_{n=1}^{2N}x_{n}^{2}$ and $x_{n}\sim\mathcal{N}(0,1),\ \forall n$. Therefore, $\frac{1}{\|\boldsymbol{g}_{ll}^{d}\|^{2}} = \frac{1}{\psi_{ll}^{d}\|\boldsymbol{w}\|^{2}} = \frac{2}{\psi_{ll}^{d}\sum_{n=1}^{2N}x_{n}^{2}}$ where $\frac{1}{\sum_{n=1}^{2N}x_{n}^{2}}$ has inverse Chi-squared distribution with $2N$ degrees of freedom with $\mathbb{E}\left\lbrace \frac{1}{\sum_{n=1}^{2N}x_{n}^{2}} \right\rbrace = \frac{1}{2N - 2}$ which results in
 \begin{equation}
 \begin{split}
     \mathbb{E}\left\lbrace\dfrac{1}{\rho^{d}\eta_{l}^{d}\left\|\boldsymbol{g}_{ll}^{d}\right\|^{2}}\right\rbrace &= \mathbb{E}\left\lbrace\frac{2}{\rho^{d}\eta_{l}^{d}\psi_{ll}^{d}\sum_{n=1}^{2N}x_{n}^{2}}\right\rbrace = \frac{2}{\rho^{d}\eta_{l}^{d}\psi_{ll}^{d}}\mathbb{E}\left\lbrace\frac{1}{\sum_{n=1}^{2N}x_{n}^{2}}\right\rbrace\\
     &= \frac{2}{\rho^{d}\eta_{l}^{d}\psi_{ll}^{d}}\times\frac{1}{2(N - 1)} = \frac{1}{\rho^{d}\eta_{l}^{d}\psi_{ll}^{d}(N - 1)}.
     \end{split}
 \end{equation}

\bibliography{References}

\begin{thebibliography}{10}

\bibitem{ngo2017cell}
H.~Q. Ngo, A.~Ashikhmin, H.~Yang, E.~G. Larsson, and T.~L. Marzetta,
  ``Cell-free massive {MIMO} versus small cells,'' {\em IEEE Transactions on
  Wireless Communications}, vol.~16, no.~3, pp.~1834--1850, 2017.

\bibitem{nayebi2017precoding}
E.~Nayebi, A.~Ashikhmin, T.~L. Marzetta, H.~Yang, and B.~D. Rao, ``Precoding
  and power optimization in cell-free massive {MIMO} systems,'' {\em IEEE
  Transactions on Wireless Communications}, vol.~16, no.~7, pp.~4445--4459,
  2017.

\bibitem{rajatheva2020white}
N.~Rajatheva {\em et~al.}, {\em White paper on broadband connectivity in {6G}}.
\newblock 6G Research Visions, No. 10, University of Oulu, 2020.

\bibitem{zhang2020prospective}
J.~Zhang, E.~Bj{\"o}rnson, M.~Matthaiou, D.~W.~K. Ng, H.~Yang, and D.~J. Love,
  ``Prospective multiple antenna technologies for beyond {5G},'' {\em IEEE
  Journal on Selected Areas in Communications}, vol.~38, no.~8, pp.~1637--1660,
  2020.

\bibitem{liu2020graph}
H.~Liu, J.~Zhang, S.~Jin, and B.~Ai, ``Graph coloring based pilot assignment
  for cell-free massive {MIMO} systems,'' {\em IEEE Transactions on Vehicular
  Technology}, vol.~69, no.~8, pp.~9180--9184, 2020.

\bibitem{zhang2020envisioning}
S.~Zhang, J.~Liu, H.~Guo, M.~Qi, and N.~Kato, ``Envisioning device-to-device
  communications in {6G},'' {\em IEEE Network}, vol.~34, no.~3, pp.~86--91,
  2020.

\bibitem{buzzi2019user}
S.~Buzzi, C.~D’Andrea, A.~Zappone, and C.~D’Elia, ``User-centric {5G}
  cellular networks: Resource allocation and comparison with the cell-free
  massive {MIMO} approach,'' {\em IEEE Transactions on Wireless
  Communications}, vol.~19, no.~2, pp.~1250 -- 1264, 2019.

\bibitem{bjornson2019making}
E.~Bj{\"o}rnson and L.~Sanguinetti, ``Making cell-free massive {MIMO}
  competitive with {MMSE} processing and centralized implementation,'' {\em
  IEEE Transactions on Wireless Communications}, vol.~19, no.~1, pp.~77--90,
  2020.

\bibitem{zhang2019cellR}
J.~Zhang, S.~Chen, Y.~Lin, J.~Zheng, B.~Ai, and L.~Hanzo, ``Cell-free massive
  {MIMO}: A new next-generation paradigm,'' {\em IEEE Access}, vol.~7,
  pp.~99878--99888, 2019.

\bibitem{mai2018pilot}
T.~C. Mai, H.~Q. Ngo, M.~Egan, and T.~Q. Duong, ``Pilot power control for
  cell-free massive {MIMO},'' {\em IEEE Transactions on Vehicular Technology},
  vol.~67, no.~11, pp.~11264--11268, 2018.

\bibitem{liu2019tabu}
H.~Liu, J.~Zhang, X.~Zhang, A.~Kurniawan, T.~Juhana, and B.~Ai,
  ``Tabu-search-based pilot assignment for cell-free massive {MIMO} systems,''
  {\em IEEE Transactions on Vehicular Technology}, vol.~69, no.~2,
  pp.~2286--2290, 2020.

\bibitem{jin2020spectral}
S.-N. Jin, D.-W. Yue, and H.~H. Nguyen, ``Spectral and energy efficiency in
  cell-free massive {MIMO} systems over correlated {Rician} fading,'' {\em IEEE
  Systems Journal}, 2021.

\bibitem{interdonato2020local}
G.~Interdonato, M.~Karlsson, E.~Bj{\"o}rnson, and E.~G. Larsson, ``Local
  partial zero-forcing precoding for cell-free massive {MIMO},'' {\em IEEE
  Transactions on Wireless Communications}, vol.~19, no.~7, pp.~4758--4774,
  2020.

\bibitem{liu2019spectral}
P.~Liu, K.~Luo, D.~Chen, and T.~Jiang, ``Spectral efficiency analysis of
  cell-free massive {MIMO} systems with zero-forcing detector,'' {\em IEEE
  Transactions on Wireless Communications}, vol.~19, no.~2, pp.~795--807, 2020.

\bibitem{hu2019cell}
X.~Hu, C.~Zhong, X.~Chen, W.~Xu, H.~Lin, and Z.~Zhang, ``Cell-free massive
  {MIMO} systems with low resolution {ADC}s,'' {\em IEEE Transactions on
  Communications}, vol.~67, no.~10, pp.~6844--6857, 2019.

\bibitem{masoumi2019performance}
H.~Masoumi and M.~J. Emadi, ``Performance analysis of cell-free massive {MIMO}
  system with limited fronthaul capacity and hardware impairments,'' {\em IEEE
  Transactions on Wireless Communications}, vol.~19, no.~2, pp.~1038--1053,
  2020.

\bibitem{zhang2018performanceR}
J.~Zhang, Y.~Wei, E.~Bj{\"o}rnson, Y.~Han, and S.~Jin, ``Performance analysis
  and power control of cell-free massive {MIMO} systems with hardware
  impairments,'' {\em IEEE Access}, vol.~6, pp.~55302--55314, 2018.

\bibitem{zhang2020rfR}
Y.~Zhang, M.~Zhou, Y.~Cheng, L.~Yang, and H.~Zhu, ``{RF} impairments and
  low-resolution {ADCs} for nonideal uplink cell-free massive {MIMO} systems,''
  {\em IEEE Systems Journal}, 2021.

\bibitem{zheng2020efficientR}
J.~Zheng, J.~Zhang, L.~Zhang, X.~Zhang, and B.~Ai, ``Efficient receiver design
  for uplink cell-free massive {MIMO} with hardware impairments,'' {\em IEEE
  Transactions on Vehicular Technology}, vol.~69, no.~4, pp.~4537--4541, 2020.

\bibitem{alonzo2019energy}
M.~Alonzo, S.~Buzzi, A.~Zappone, and C.~D’Elia, ``Energy-efficient power
  control in cell-free and user-centric massive {MIMO} at millimeter wave,''
  {\em IEEE Transactions on Green Communications and Networking}, vol.~3,
  no.~3, pp.~651--663, 2019.

\bibitem{rezaei2020underlaid}
F.~Rezaei, A.~R. Heidarpour, C.~Tellambura, and A.~Tadaion, ``Underlaid
  spectrum sharing for cell-free massive {MIMO-NOMA},'' {\em IEEE
  Communications Letters}, vol.~24, no.~4, pp.~907--911, 2020.

\bibitem{d2020analysis}
C.~D’Andrea, A.~Garcia-Rodriguez, G.~Geraci, L.~G. Giordano, and S.~Buzzi,
  ``Analysis of {UAV} communications in cell-free massive {MIMO} systems,''
  {\em IEEE Open Journal of the Communications Society}, vol.~1, pp.~133--147,
  2020.

\bibitem{wang2019performance}
D.~Wang, M.~Wang, P.~Zhu, J.~Li, J.~Wang, and X.~You, ``Performance of
  network-assisted full-duplex for cell-free massive {MIMO},'' {\em IEEE
  Transactions on Communications}, vol.~68, no.~3, pp.~1464--1478, 2020.

\bibitem{bashar2019performance}
M.~Bashar, K.~Cumanan, A.~G. Burr, H.~Q. Ngo, L.~Hanzo, and P.~Xiao, ``On the
  performance of cell-free massive {MIMO} relying on adaptive {NOMA/OMA}
  mode-switching,'' {\em IEEE Transactions on Communications}, vol.~68, no.~2,
  pp.~792--810, 2020.

\bibitem{srinivasan2019joint}
M.~Srinivasan, A.~Subhash, and S.~Kalyani, ``Joint power and resource
  allocation for {D2D} communication with low-resolution {ADC},'' in {\em 2019
  53rd Asilomar Conference on Signals, Systems, and Computers}, pp.~995--999,
  IEEE, 2019.

\bibitem{xu2016pilot}
H.~Xu, N.~Huang, Z.~Yang, J.~Shi, B.~Wu, and M.~Chen, ``Pilot allocation and
  power control in {D2D} underlay massive {MIMO} systems,'' {\em IEEE
  Communications Letters}, vol.~21, no.~1, pp.~112--115, 2016.

\bibitem{yang2016downlink}
Z.~Yang, N.~Huang, H.~Xu, Y.~Pan, Y.~Li, and M.~Chen, ``Downlink resource
  allocation and power control for device-to-device communication underlaying
  cellular networks,'' {\em IEEE Communications Letters}, vol.~20, no.~7,
  pp.~1449--1452, 2016.

\bibitem{pan2017resourceR}
Y.~Pan, C.~Pan, Z.~Yang, and M.~Chen, ``Resource allocation for {D2D}
  communications underlaying a {NOMA}-based cellular network,'' {\em IEEE
  Wireless Communications Letters}, vol.~7, no.~1, pp.~130--133, 2017.

\bibitem{he2017spectral}
A.~He, L.~Wang, Y.~Chen, K.-K. Wong, and M.~Elkashlan, ``Spectral and energy
  efficiency of uplink {D2D} underlaid massive {MIMO} cellular networks,'' {\em
  IEEE Transactions on Communications}, vol.~65, no.~9, pp.~3780--3793, 2017.

\bibitem{lin2015interplay}
X.~Lin, R.~W. Heath, and J.~G. Andrews, ``The interplay between massive {MIMO}
  and underlaid {D2D} networking,'' {\em IEEE Transactions on Wireless
  Communications}, vol.~14, no.~6, pp.~3337--3351, 2015.

\bibitem{liu2018performance}
X.~Liu, Y.~Li, L.~Xiao, and J.~Wang, ``Performance analysis and power control
  for multi-antenna {V2V} underlay massive {MIMO},'' {\em IEEE Transactions on
  Wireless Communications}, vol.~17, no.~7, pp.~4374--4387, 2018.

\bibitem{xu2017pilot}
H.~Xu, W.~Xu, Z.~Yang, J.~Shi, and M.~Chen, ``Pilot reuse among {D2D} users in
  {D2D} underlaid massive {MIMO} systems,'' {\em IEEE Transactions on Vehicular
  Technology}, vol.~67, no.~1, pp.~467--482, 2017.

\bibitem{zhang2017mixedADCrician}
J.~Zhang, L.~Dai, Z.~He, S.~Jin, and X.~Li, ``Performance analysis of
  mixed-{ADC} massive {MIMO} systems over {Rician} fading channels,'' {\em IEEE
  Journal on Selected Areas in Communications}, vol.~35, no.~6, pp.~1327--1338,
  2017.

\bibitem{zhou2016fronthaul}
Y.~Zhou and W.~Yu, ``Fronthaul compression and transmit beamforming
  optimization for multi-antenna uplink {C-RAN},'' {\em IEEE Transactions on
  Signal Processing}, vol.~64, no.~16, pp.~4138--4151, 2016.

\bibitem{zhang2018low}
J.~Zhang, L.~Dai, X.~Li, Y.~Liu, and L.~Hanzo, ``On low-resolution {ADCs} in
  practical {5G} millimeter-wave massive {MIMO} systems,'' {\em IEEE
  Communications Magazine}, vol.~56, no.~7, pp.~205--211, 2018.

\bibitem{zhu2015soft}
X.~Zhu, Z.~Wang, C.~Qian, L.~Dai, J.~Chen, S.~Chen, and L.~Hanzo, ``Soft pilot
  reuse and multicell block diagonalization precoding for massive {MIMO}
  systems,'' {\em IEEE Transactions on Vehicular Technology}, vol.~65, no.~5,
  pp.~3285--3298, 2015.

\bibitem{bjornson2017massive}
E.~Bj{\"o}rnson, J.~Hoydis, L.~Sanguinetti, {\em et~al.}, ``Massive {MIMO}
  networks: Spectral, energy, and hardware efficiency,'' {\em Foundations and
  Trends in Signal Processing}, vol.~11, no.~3-4, pp.~154--655, 2017.

\bibitem{marzetta2016fundamentals}
T.~L. Marzetta, E.~G. Larsson, H.~Yang, and H.~Q. Ngo, {\em Fundamentals of
  massive {MIMO}}.
\newblock Cambridge University Press, 2016.

\bibitem{zhu2015graph}
X.~Zhu, L.~Dai, and Z.~Wang, ``Graph coloring based pilot allocation to
  mitigate pilot contamination for multi-cell massive {MIMO} systems,'' {\em
  IEEE Communications Letters}, vol.~19, no.~10, pp.~1842--1845, 2015.

\bibitem{cormen2009introduction}
T.~H. Cormen, C.~E. Leiserson, R.~L. Rivest, and C.~Stein, {\em Introduction to
  algorithms}.
\newblock MIT press, 2009.

\bibitem{boyd2004convex}
S.~Boyd, S.~P. Boyd, and L.~Vandenberghe, {\em Convex optimization}.
\newblock Cambridge university press, 2004.

\bibitem{zhang2018mixed}
J.~Zhang, L.~Dai, Z.~He, B.~Ai, and O.~A. Dobre, ``Mixed-{ADC/DAC} multipair
  massive {MIMO} relaying systems: {Performance} analysis and power
  optimization,'' {\em IEEE Transactions on Communications}, vol.~67, no.~1,
  pp.~140--153, 2018.

\bibitem{cui2005energy}
S.~Cui, A.~J. Goldsmith, and A.~Bahai, ``Energy-constrained modulation
  optimization,'' {\em IEEE transactions on wireless communications}, vol.~4,
  no.~5, pp.~2349--2360, 2005.

\end{thebibliography}
\bibliographystyle{ieeetr}

\end{document}